\documentclass[a4paper,fontsize=11pt,oneside,abstracton]{scrartcl}

\usepackage{fixltx2e}

\usepackage[english]{babel}
\usepackage{microtype}
\usepackage[utf8]{inputenc}
\usepackage{url}
\usepackage{subfigure}
\usepackage{xspace}
\usepackage{xparse}

\usepackage{multirow}
\usepackage{booktabs}
\usepackage{longtable}

\usepackage[colorlinks = true]{hyperref}
\usepackage{xcolor}
\definecolor{darkred}  {rgb}{0.5,0,0}
\definecolor{darkblue} {rgb}{0,0,0.5}
\definecolor{darkgreen}{rgb}{0,0.5,0}
\hypersetup{
  urlcolor   = blue,         % color of external links
  linkcolor  = darkblue,     % color of internal links
  citecolor  = darkgreen,    % color of links to bibliography
  filecolor  = darkred       % color of file links
}

\usepackage{mathrsfs}
\usepackage{wasysym}

\usepackage{amsmath}

\usepackage{amsthm}
\usepackage{amssymb}
\usepackage{pifont}
\usepackage{mathtools}
\mathtoolsset{centercolon}

\usepackage{cleveref}
\crefname{lemma}{lemma}{lemmas}
\crefname{proposition}{proposition}{propositions}
\crefname{definition}{definition}{definitions}
\crefname{theorem}{theorem}{theorems}
\crefname{conjecture}{conjecture}{conjectures}
\crefname{corollary}{corollary}{corollaries}
\crefname{example}{example}{examples}
\crefname{section}{section}{sections}
\crefname{appendix}{appendix}{appendices}
\crefname{figure}{fig.}{figs.}
\crefname{equation}{eq.}{eqs.}
\crefname{table}{table}{tables}
\crefname{item}{property}{properties}
\crefname{remark}{remark}{remarks}

\newtheorem{theorem}{Theorem}
\newtheorem{definition}[theorem]{Definition}
\newtheorem{corollary}[theorem]{Corollary}
\newtheorem{proposition}[theorem]{Proposition}
\newtheorem{lemma}[theorem]{Lemma}
\newtheorem{example}[theorem]{Example}
\newtheorem{remark}[theorem]{Remark}
\newtheorem{fact}[theorem]{Fact}
\newtheorem{claim}[theorem]{Claim}

\usepackage{dsfont}
\usepackage{lmodern}
\DeclareMathAlphabet{\mathpzc}{T1}{pzc}{m}{it}

\newcommand{\set}[1] {\mathrm{#1}}
\newcommand{\sset}[1] {\mathds{#1}}
\newcommand{\map}[1] {\mathrm{#1}}
\newcommand{\op}[1] {\mathbf{#1}}
\newcommand{\sop}[1] {\mathds{#1}}

\newcommand{\ket}[1]{\left|#1\right\rangle}
\newcommand{\bra}[1]{\left\langle#1\right|}
\newcommand{\ketbra}[2]{\left|#1\middle\rangle\!\middle\langle#2\right|}
\newcommand{\braket}[2]{\left\langle #1\middle|#2\right\rangle}
\newcommand{\proj}[1]{\ketbra{#1}{#1}}

\newcommand{\point}[0] {\ \ .}
\newcommand{\comma}[0] {\ \ ,}
\newcommand{\mt}[1] {\quad\text{#1}\quad}

\newcommand{\ct}{^{\dagger}} % conjugate transpose

\newcommand{\knownth}[1] {\emph{#1}}

\newcommand{\hs}[0] {\ensuremath{\mathcal{H}}\xspace}
\newcommand{\U}{\mathrm{U}} % unitary group

\newcommand{\1}[0] {\sop{1}}

\makeatletter
\providecommand*{\cupdot}{%
  \mathbin{%
    \mathpalette\@cupdot{}%
  }%
}
\newcommand*{\@cupdot}[2]{%
  \ooalign{%
    $\m@th#1\cup$\cr
    \hidewidth$\m@th#1\cdot$\hidewidth
  }%
}
\makeatother

\newcommand{\tightoverset}[2]{%
  \mathop{#2}\limits^{\vbox to -.5ex{\kern-0.75ex\hbox{$#1$}\vss}}}
\newcommand{\isomap}[0] {\tightoverset\sim\longrightarrow}
\newcommand*{\longhookrightarrow}{\ensuremath{\lhook\joinrel\longrightarrow}}

\makeatletter
\renewcommand*\env@matrix[1][*\c@MaxMatrixCols c]{%
  \hskip -\arraycolsep
  \let\@ifnextchar\new@ifnextchar
  \array{#1}}
\makeatother

\usepackage{array}
\newcolumntype{C}[1]{>{\centering\hspace{0pt}\arraybackslash}p{#1}}
\newcolumntype{R}[1]{>{\raggedleft\hspace{0pt}\arraybackslash}p{#1}}
\newcolumntype{L}[1]{>{\raggedright\hspace{0pt}\arraybackslash}p{#1}}

\DeclarePairedDelimiter{\abs}{\lvert}{\rvert}

\DeclarePairedDelimiter{\of}{\lparen}{\rparen}

\renewcommand{\Re}{\operatorname{Re}}

\DeclareMathOperator{\supp}{supp}
\DeclareMathOperator{\poly}{poly}
\DeclareMathOperator{\diag}{diag}
\DeclareMathOperator{\spn}{span}

\usepackage[section]{placeins}
\usepackage{enumitem}

\usepackage{floatrow}
\usepackage{wrapfig}
\usepackage{caption}
\usepackage{ifthen}
\newboolean{CACHETIKZ}
\setboolean{CACHETIKZ}{false}
\usepackage{tikz}
\usetikzlibrary{calc,arrows,patterns,decorations.pathmorphing,backgrounds,fadings,external}
\ifCACHETIKZ
\tikzset{external/system call={pdflatex \tikzexternalcheckshellescape -halt-on-error -interaction=batchmode -jobname "\image" "\texsource"}}
\fi

\definecolor{tape}{rgb}{.9,.9,.9}
\definecolor{transition}{rgb}{.1,.4,.9}
\definecolor{quantum}{rgb}{1,.8,.1}
\newcommand{\ta}{.8}
\newcommand{\dx}{.175}
\tikzfading[name=fade down, top color=transparent!0, bottom color=transparent!100]
\tikzfading[name=fade up, top color=transparent!100, bottom color=transparent!0]

\pgfdeclarelayer{blurb}
\pgfdeclarelayer{pipe}
\pgfdeclarelayer{tape}
\pgfdeclarelayer{quantum}
\pgfsetlayers{background,blurb,pipe,tape,quantum,main}

\usepackage{environ}
\NewEnviron{tape}[1][1]{%
\tikzsetnextfilename{\pdfmdfivesum{\expandafter\detokenize\expandafter\expandafter\expandafter{\csname BODY\endcsname}}}%
\begin{tikzpicture}[
  tape/.style = {fill = tape, draw = tape, line width=\dx\ts, align = center, rounded corners=.1\ts},
  debug/.style = {fill = none, draw = none},
  quantum/.style = {fill = quantum, draw = quantum, line width=\dx\ts, align = center, rounded corners=.1\ts},
  transition/.style = {fill = transition, draw=none, line width=\dx\ts},
  draw=black,
  line width=.08\ts,
  baseline = #1,
  x=\ta\ts,y=\ts
  ]
  
  \BODY
  
  \begin{pgfonlayer}{tape}
      \draw[tape] (0,0) rectangle (\thetx+\thetx*\dx-\dx,1);
  \end{pgfonlayer}
  \draw[debug] (current bounding box.north east) rectangle (current bounding  box.south west);
\end{tikzpicture}%
}

\makeatletter
\newcommand\tapesymbols[1]{%
  \@tapesymbols#1\@nnil
}
\def\@tapesymbols#1{%
  \ifx\@nnil#1\relax\else%
  \@nameuse{tapesymbols@char@#1\expandafter}%
  \fi
}
\DeclareDocumentCommand \defcharcode { m O{} O{} m }{
  \@namedef{tapesymbols@char@#1}#2#3{
    \begin{scope}[shift={(\thetx + \thetx*\dx,0)}]
      \draw[debug] (0,0) rectangle (\ta\ts,\ts);
      \begin{pgfinterruptboundingbox}
        #4
      \end{pgfinterruptboundingbox}
    \end{scope}
    \stepcounter{tx}
    \@tapesymbols
  }%
}
\makeatother
\DeclareDocumentCommand \inputinkscape{ m O{.01} }{
  \begin{scope}[x=#2\ts,y=#2\ts,yscale=-1,shift={(0,-\ts)}]%have to give explicit units here
    \InputIfFileExists{./figures/#1}{}{\input{#1}}
  \end{scope}
}

\defcharcode{>}{
  \draw[transition, path fading=fade up] (-\dx,-.2) rectangle (1+\dx,.3);
  \addtocounter{tx}{-1}
}
\defcharcode{<}{
  \draw[transition, path fading=fade down] (-\dx,.7) rectangle (1+\dx,1.2);
  \addtocounter{tx}{-1}
}
\defcharcode{-}{
  \begin{pgfonlayer}{quantum}
    \draw[quantum] (0,0) rectangle (\ta\ts,\ts);
  \end{pgfonlayer}
  \addtocounter{tx}{-1}
}
\defcharcode{,}{
  \begin{pgfonlayer}{quantum}
    \draw[quantum] (0,0) rectangle (2*\ta\ts+6*\dx,\ts);
  \end{pgfonlayer}
  \addtocounter{tx}{-1}
}
\newcommand{\rhombus}[1]{
  \draw[#1] (.1,.5) -- (.5,.8) -- (.9,.5) -- (.5,.2) -- cycle;
}

\defcharcode{|}{
  \draw[double=tape] (.5,0) -- (.5,1);
}
\defcharcode{!}{
  \draw[decorate, decoration={coil, aspect=0, amplitude=.04cm, segment length=0.07cm}] (.5,0) -- (.5,1);
}
\defcharcode{?}{
  \draw[line width=1.5*\dx\ts,black] (.5,0) -- (.5,1);
}
\defcharcode{o}{
  \draw[rounded corners=.2] (.1,.1) rectangle (.9,.9);
}
\defcharcode{O}[#1]{
  \draw[rounded corners=.2,fill=black] (.05,.05) rectangle (.95,.95);
  \ifx#1T
    \node[text centered,text height=5,scale=.7,white] at (.5,.4) {\textbf{\textup{T}}};
  \else\ifx#1U
    \node[text centered,text height=5,scale=.7,white] at (.5,.4) {\textbf{\textup{U}}};
  \else\ifx#1A
    \node[text centered,text height=5,scale=.7,white] at (.5,.4) {\textbf{\textup{A}}};
  \else\ifx#1H
    \node[text centered,text height=5,scale=.7,white] at (.5,.4) {\textbf{\textup{H}}};
  \else\ifx#1S
    \node[text centered,text height=5,scale=.7,white] at (.5,.4) {\textbf{\textup{S}}};
  \else    
    \node[text centered,text depth=1,white] at (.5,.4) {$#1$};
  \fi\fi\fi\fi\fi
}
\defcharcode{.}{
  \fill (.33,.5) circle [radius=.1];
  \fill (1,.5) circle [radius=.1];
  \fill (1.66,.5) circle [radius=.1];
  \stepcounter{tx}
}
\defcharcode{G}{
  \inputinkscape{ghost}
}
\defcharcode{B}{
  \inputinkscape{ghostb}
}
\defcharcode{g}{
  \inputinkscape{ighost}
}
\defcharcode{b}{
  \inputinkscape{ighostb}
}
\defcharcode{I}{
  \inputinkscape{initializer}
}
\defcharcode{+}{
  \draw[fill,rounded corners=.2] (.1,0) -- (.1,1) -- (1,.5) -- cycle;
  \begin{scope}[line width=1,white,line cap=round]
    \draw (0.2,.5) -- (0.7,.5);
    \draw (0.45,.25) -- (.45,.75);
  \end{scope}
}
\defcharcode{D}{ % the same as # but doesn't cause problems in \newcommand
  \draw[fill,rounded corners=.2] (.1,0) -- (.1,1) -- (1,.5) -- cycle;
}
\defcharcode{X}{
  \draw[rounded corners=.2] (.1,0) -- (.1,1) -- (1,.5) -- cycle;
}
\defcharcode{Y}{
  \draw[rounded corners=.2] (.1,0) -- (.1,1) -- (1,.5) -- cycle;
  \draw[rounded corners=.2] (.3,.3) -- (.3,.7) -- (.65,.5) -- cycle;
}
\defcharcode{Z}{
    \draw[rounded corners=.2] (.9,0) -- (0,.5) -- (.9,1) -- cycle;    
}
\defcharcode{C}{
    \draw[fill, rounded corners=.2] (.9,0) -- (0,.5) -- (.9,1) -- cycle;    
}
\defcharcode{:}[#1]{
  \ifx#1T
    \node[text centered,text height=5,scale=.7] at (.5,.35) {\textbf{\textup{T}}};
  \else\ifx#1U
    \node[text centered,text height=5,scale=.7] at (.5,.35) {\textbf{\textup{U}}};
  \else\ifx#1A
    \node[text centered,text height=5,scale=.7] at (.5,.35) {\textbf{\textup{A}}};
  \else\ifx#1H
    \node[text centered,text height=5,scale=.7] at (.5,.35) {\textbf{\textup{H}}};
  \else\ifx#1S
    \node[text centered,text height=5,scale=.7] at (.5,.35) {\textbf{\textup{S}}};
  \else
    \node[text centered,text height=5] at (.5,.4) {$#1$};
  \fi\fi\fi\fi\fi
}
\defcharcode{;}[#1]{
  \node[text centered,text height=5,text depth=-0.2] at (.6,.4) {$\dot{#1}$};
}
\defcharcode{"}[#1]{
    \node[text centered,text height=5,text depth=-0.2] at (.6,.4) {$\ddot{#1}$};
}
\defcharcode{*}[#1]{
    \node[text centered,text height=5,text depth=-0.2] at (.6,.4) {$\vec{#1}$};
}
\defcharcode{@}[#1]{
  \draw[fill = black] (0,0) -- (0,1) -- (.6,1) -- (1,.5) -- (.6,0) -- cycle;
  \ifx#1r
    \rhombus{fill=white,draw=none}
  \else\ifx#1T
    \node[text centered,text height=5,scale=.7,white] at (.4,.4) {\textbf{\textup{T}}};
  \else\ifx#1U
    \node[text centered,text height=5,scale=.7,white] at (.4,.4) {\textbf{\textup{U}}};
  \else\ifx#1A
    \node[text centered,text height=5,scale=.7,white] at (.4,.4) {\textbf{\textup{A}}};
  \else\ifx#1H
    \node[text centered,text height=5,scale=.7,white] at (.4,.4) {\textbf{\textup{H}}};
  \else\ifx#1S
    \node[text centered,text height=5,scale=.7,white] at (.4,.4) {\textbf{\textup{S}}};
  \else
    \node[text centered,text height=5,white] at (.4,.4) {$#1$};
  \fi\fi\fi\fi\fi\fi
}
\defcharcode{r}{
  \rhombus{}
}
\defcharcode{R}{
  \rhombus{fill=black}
}
\defcharcode{w}{
  \draw (.1,.6) -- (.5,.9) -- (.9,.6);
  \draw (.1,.4) -- (.5,.1) -- (.9,.4);
}

\newlength{\ts}
\newcounter{tx}
\DeclareDocumentCommand \$ { O{0} m } {%
  \settoheight{\ts}{Z}%
  \setlength{\ts}{1.2\ts}%
  \setcounter{tx}{0}%
  \begin{tape}
    \draw[debug] (-#1-#1*\dx-\dx,0) -- (0,1);
    \tapesymbols{#2}
  \end{tape}%
  \xspace
}

\newenvironment{rules}{%
  \renewcommand{\arraystretch}{0.8}%
  \begin{tabular}[t]{@{}l@{}}
}{
  \end{tabular}%
}
\newcommand{\rvdots}{\multicolumn{1}{c}{\vdots}}

\usepackage{setspace}
\newcommand{\lmin}[0]{\lambda_{\min}}           % smallest eigenvalue
\newcommand{\lminpos}[1]{\lmin(#1|_{\supp #1})} % smallest positive eigenvalue
\newcommand{\qn}[1]{\abs{#1}_q}                 % number of quantum symbols

\newcommand{\np}{\textnormal{\textsf{NP}}\xspace}

\newcommand{\BQ}[1][f]{\textnormal{\textsf{BQ}}\ensuremath{(#1)}\xspace}
\newcommand{\BQP}{\textnormal{\textsf{BQP}}\xspace}
\newcommand{\BQEXP}{\textnormal{\textsf{BQEXP}}\xspace}
\newcommand{\QM}[1][f]{\textnormal{\textsf{QMA}\ensuremath{(#1)}}\xspace}
\newcommand{\QMA}{\textnormal{\textsf{QMA}}\xspace}
\newcommand{\QMAEXP}{\textnormal{\textsf{QMA\textsubscript{EXP}}}\xspace}
\newcommand{\NEXP}{\textnormal{\textsf{NEXP}}\xspace}
\newcommand{\yes}{\textnormal{\textsf{YES}}\xspace}
\newcommand{\no}{\textnormal{\textsf{NO}}\xspace}

\newcommand{\PromP}{\Pi}

\DeclareDocumentCommand{\TILH}{o}{%
  \IfNoValueTF{#1}{%
    TILH%
  }{%
    $#1$-TILH%
  }\xspace%
}
\DeclareDocumentCommand{\qi}{o}{% % initial state
  \IfNoValueTF{#1}{%
    \ket{\psi_\mathrm{in}}%
  }{%
    \ket{\psi_\mathrm{in}(#1)}%
  }\xspace%
}
\DeclareDocumentCommand{\hsi}{o}{% % family of initial states
  \IfNoValueTF{#1}{%
    \ensuremath{\{\qi[x]\}_{x\in I}}\xspace%
  }{%
    \ensuremath{\{\qi[x]\}_{x\in#1}}\xspace%
  }\xspace%
}
\newcommand{\hsf}[0] {\ensuremath{\mathcal{H}_\mathrm{halt}}\xspace} % halting subspace

\newcommand{\toffoli}{\textsc{Toffoli}\xspace}
\newcommand{\swap}{\textsc{Swap}\xspace}
\newcommand{\cnot}{\textsc{Cnot}\xspace}
\newcommand{\hadamard}{\textsc{Hadamard}\xspace}

\DeclareMathOperator{\enc}{enc}

\usepackage[backend=biber,style=numeric,sorting=none,doi=false,isbn=false,url=false,maxbibnames=20]{biblatex}
\bibliography{Bibliography/maris}
\bibliography{Bibliography/hamilton}
\bibliography{Bibliography/books}
\newbibmacro{string+doi}[1]{\iffieldundef{doi}{#1}{\href{https://dx.doi.org/\thefield{doi}}{#1}}}
\DeclareFieldFormat{title}{\usebibmacro{string+doi}{\mkbibemph{#1}}}
\DeclareFieldFormat[article]{title}{\usebibmacro{string+doi}{\mkbibquote{#1}}}
\DeclareFieldFormat[incollection]{title}{\usebibmacro{string+doi}{\mkbibquote{#1}}}

\usepackage[affil-it]{authblk}

\AtBeginDocument{\addtocontents{toc}{\protect\setlength{\parskip}{-.5pt}}}

\begin{document}
\title{\vspace{-1cm}The Complexity of Translationally-Invariant Spin Chains with Low Local Dimension}
\date{}
\author[1]{Johannes Bausch\thanks{jkrb2@cam.ac.uk}}
\author[2]{Toby Cubitt\thanks{t.cubitt@ucl.ac.uk}}
\author[1]{Maris Ozols\thanks{marozols@gmail.com}}
\affil[1]{DAMTP\\Centre for Mathematical Sciences\\University of Cambridge\\Cambridge CB3 0WB, UK}
\affil[2]{Department of Computer Science\\University College London\\London WC1E 6BT, UK}

\maketitle
\thispagestyle{empty}
\enlargethispage{2cm}

\begin{abstract}
  \normalsize
  We prove that estimating the ground state energy of a translationally-invariant, nearest-neighbour Hamiltonian on a 1D spin chain is \QMAEXP-complete, even for systems of low local dimension ($\approx 40$). This is an improvement over the best previously-known result by several orders of magnitude, and it shows that spin-glass-like frustration can occur in translationally-invariant quantum systems with a local dimension comparable to the smallest-known non-translationally-invariant systems with similar behaviour.

  While previous constructions of such systems rely on standard models of quantum computation, we construct a new model that is particularly well-suited for encoding quantum computation into the ground state of a translationally-invariant system. This allows us to shift the proof burden from optimizing the Hamiltonian encoding a standard computational model, to proving universality of a simple model.

  Previous techniques for encoding quantum computation into the ground state of a local Hamiltonian allow only a linear sequence of gates, hence only a linear (or nearly linear) path in the graph of all computational states. We extend these techniques by allowing significantly more general paths, including branching and cycles, thus enabling a highly efficient encoding of our computational model. However, this requires more sophisticated techniques for analysing the spectrum of the resulting Hamiltonian. To address this, we introduce a framework of graphs with unitary edge labels. After relating our Hamiltonian to the Laplacian of such a unitary labelled graph, we analyse its spectrum by combining matrix analysis and spectral graph theory techniques.
\end{abstract}

\newpage
\thispagestyle{empty}
\enlargethispage{5cm}
\begin{spacing}{1}
  \tableofcontents
\end{spacing}
%\listoftodos

\newpage
\section{Background and Motivation}
Complex physical behaviour can emerge from even very simple rules. Yet if the system is \emph{too} simple, one can often rule out the possibility of any exotic behaviour. Just how simple can a system be to nonetheless feature complex properties? Much of the progress in Hamiltonian complexity and related areas over the last decade can be viewed as improving our understanding of where this boundary between simple and complex lies.

For example, consider 1D spin chains with translationally-invariant nearest neighbour interactions. Hastings proved that if the Hamiltonian describing the system is gapped, the ground state entanglement has to follow an area law~\cite{Hastings2007a}. In 1D, the area law means that the entanglement entropy between any contiguous region and its complement is upper-bounded by a constant, independent of the size of the region. It was believed that even for non-gapped Hamiltonians, area-law violations would contribute at most $\log$ corrections in the system size.  Such long-range correlations in a spin chain's ground state which scale with the system's size are a common indicator of criticality, i.e.\ they show that the system is close to a quantum phase transition. The entanglement entropy is then expected to scale logarithmically with the number of spins, since critical spin chains can often be related to a conformal field theory.

However, using Hamiltonian complexity techniques, Irani \cite{Irani2007} constructed an example of a spin chain in 1D that exhibits violation of the area-law \emph{beyond} logarithmic corrections, indicating that one cannot describe such behaviour by a conformal field theory.  Irani's construction breaks translational-invariance, so it cannot directly be compared to systems satisfying area laws. A later construction \cite{Gottesman2009} can give a similar area-law violation whilst preserving translational-invariance. However, the required local dimension, $O(10^6)$, is vast.  It is therefore at best questionable whether this area-law violation could ever be observed in practice.  Does this mean that such violations only occur for some peculiar theoretical models with non-translationally-invariant couplings, or unrealistically large Hilbert space dimensions?

We now know that the answer to this question is negative. First, it was shown by \cite{Bravyi2012a} that, even for frustration-free spin-1 chains (i.e.\ local dimension 3), one can construct interactions that yield highly entangled ground states, indicating critical behaviour.  In fact, this result delineates a strict dimension threshold for the presence ground-state entanglement in frustration-free systems. For frustration-free spin-1/2 chains (i.e.\ local dimension 2) with translationally-invariant nearest neighbour interactions, it was already known that ground states are unentangled~\cite{Chen2011}.
Building on this, Movassagh et al.~\cite{Movassagh2014} constructed models which give power-law violation of the area-law for translationally-invariant spin-5/2 chains (i.e.\ local dimension\footnote{\cite{Movassagh2014} in fact prove their result for local dimension 5 but breaking strict translational invariance by adding boundary terms at the ends of the chain. Using a trick due to~\cite{Gottesman2009}, the boundary terms can be removed at the cost of increasing the local dimension by 1.}~6), significantly improving on the bound on the local dimension threshold for power-law area-law violation from Gottesman and Irani's result.

Similar dimension-related physicality questions also surround Cubitt et al.'s result which proves that deciding whether a system is gapped or gapless in the thermodynamic limit is an undecidable problem, even for 2D spin lattices with translationally-invariant local interactions~\cite{Cubitt2015}. Again, the local Hilbert space dimension in the model they describe is vast. Bravyi and Gossett recently derived necessary and sufficient conditions for a gapped or gapless phase for frustration-free spin-1/2 chains~\cite{Bravyi2015}. So at the other end of the local dimension scale, the spectral gap problem is decidable in some cases.  However, there is evidence that an astronomical local dimension may not be a fundamental ingredient in the emergent behaviour that gives rise to undecidability of the spectral gap. The abrupt change in the spectrum at very large system sizes that is behind the undecidability, can also occur on 2D lattices of far lower-dimensional spins~\cite{Bausch2015}. Again, this poses an immediate question of whether there is some local dimension threshold above which undecidability can occur, but below which it cannot.

The original and most widely-studied question in Hamiltonian complexity theory, however, is that of estimating the ground state energy of a local Hamiltonian. Kitaev showed that this problem is \QMA-hard~\cite{Kitaev2002} (i.e.\ at least as hard as every other problem in the complexity class \QMA---the quantum generalisation of \np).  Similar to a spin glass, when cooled down these \QMA-hard systems are predicted to get stuck in one of their many meta-stable configurations, and will take exponentially long (in the system size) to find their global minimum-energy configuration.  QMA-hardness-inspired constructions lie behind all the results mentioned above. Yet even though the parameters describing \QMA hard ground state Hamiltonians have been improved successively~\cite{Kempe2006,Oliveira2008,Aharonov2009,Hallgren2013,Gottesman2009}, a lower local dimension threshold below which systems cannot feature spin-glass-like frustration is not known; for non-translationally-invariant systems we know that this bound can be at most 8. For the more physically relevant case of spin chains \emph{with} translational symmetry, however, the best-known bound is $O(10^6)$, due to Gottesman and Irani~\cite{Gottesman2009}, which is unphysically large. From a physical perspective it makes a dramatic difference if the complexity threshold is e.g.\ 7, or 1000.

In this work, we improve the best-known upper bound on the local Hilbert space dimension required for \QMA-hardness in translationally-invariant spin chains by several orders of magnitude, showing that the question of estimating the ground state energy of a local translationally-invariant Hamiltonian with nearest-neighbour interactions remains hard, even for spins on a chain with local dimension $\approx 40$.

\section{Extended Introduction and Overview of Results}
\subsection{Historical Context}
Hamiltonians are the one-stop shop for describing physical properties of multi-body quantum systems,
and are of paramount interest for an array of disciplines ranging from experimental condensed matter physics to theoretical computer science~\cite{Oliveira2008,Kempe2006,Aharonov2009,Gottesman2009,Barthel2012,Cubitt2013,Piddock2015,Cubitt2015,Wei2015}. While computer scientists are interested in the computational power of different models, for physicists it is important to calculate the structure of the low-energy spectrum of quantum systems, in particular to approximate the minimum energy of the system, i.e.~the ground state energy.

The decision problem of determining whether such a local Hamiltonian operator has lowest energy---or eigenvalue---below some $\alpha$ or above some $\beta$, with $\beta>\alpha$, can be thought of as the quantum analogue of the maximum satisfiability problem \textsc{Max-Sat}. Similar to the well-known \textsc{3-Sat}, this asks for the maximum number of clauses of a Boolean formula in conjunctive normal form that can be satisfied simultaneously. In the quantum case, each local term $\op h$ of $\op H$ is analogous to a clause while a global state $\ket{\psi}$ is analogous to a global variable assignment, and the smaller $\bra{\psi} \op h \ket{\psi}$ is, the closer $\ket{\psi}$ is to satisfying the corresponding clause $\op h$. The \textsc{Local Hamiltonian} problem formalizes the notion of maximizing the number of local terms of $\op H$ which can be simultaneously minimized by some global state $\ket\psi$, in the sense that $\bra\psi\op H\ket\psi$ is small. Physically, this minimum is equal to the lowest energy of the system.

Formally, we can state the \textsc{Local Hamiltonian} problem as the following promise problem.
\begin{definition}[\textsc{$k$-Local Hamiltonian}]\leavevmode\label{def:local-ham-intro}
  \begin{description}
    \item[Input.] An integer $n$ and a $k$-local Hamiltonian $\op H$ on a multipartite Hilbert space $(\sset C^d)^{\otimes n}$, and two real numbers $\beta > \alpha$ such that $\beta - \alpha \geq 1/p(n)$, for some fixed polynomial $p(n)$. The smallest eigenvalue $\lmin$ of $\op H$ is promised to be either smaller than $\alpha$ or greater than $\beta$.
    \item[Question.] Is $\lmin < \alpha$, or $\lmin > \beta$?
  \end{description}
\end{definition}

\begin{table}[t]
  \centering
  \def\arraystretch{1.3}
  \begin{tabular}{R{4.2cm}llL{5cm}}
    \toprule
    & \textbf{locality} & \textbf{local dimension} & \textbf{geometry and symmetries}                  \\ \hline
    Kitaev (1999)                                         & 5                 & 2                        & arbitrary \\ \hline
    Kempe, Kitaev, Regev~\cite{Kempe2006}                 & 2                 & 2                        & arbitrary \\ \hline
    Oliveira, Terhal~\cite{Oliveira2008}                  & 2                 & 2                        & 2D, planar, nearest-neighbour interactions        \\ \hline
    Aharonov, Gottesman, Irani, Kempe~\cite{Aharonov2009} & 2                 & 12                       & line, nearest-neighbour                           \\ \hline
    Hallgren, Nagaj, Narayanaswami~\cite{Hallgren2013}    & 2                 & 8                        & line, nearest-neighbour                           \\ \hline
    Gottesman, Irani~\cite{Gottesman2009}                 & 2                 & huge $(\approx 10^6)$            & line, nearest-neighbour, translationally-invariant \\\bottomrule
  \end{tabular}
  \caption{Brief historic overview of \QMA (\QMAEXP for~\cite{Gottesman2009}) completeness results in Hamiltonian complexity.}
  \label{tab:history}
\end{table}

The \textsc{$k$-Local Hamiltonian} problem has a track record of long-standing interest (cf.~\cref{tab:history}). The foundations were laid with Feynman's paper~\cite{Feynman1986} on encoding quantum circuits into the ground state of a Hamiltonian, which motivated a whole series of interesting and increasingly sophisticated results showing that variants of this problem are \QMA- or \QMAEXP-complete.%
\footnote{
  \QMAEXP is to \QMA what \NEXP is to \np. This is a necessary technicality whenever
  the input has to be specified in unary. The energy gap still
  scales inverse-polynomially with system size $n$, and the physical
  implications are exactly the same as for \QMA-completeness.
  We define these complexity classes rigorously in \cref{sec:classes}, and explain their difference in detail in \cref{sec:qmaqmaexp}.
}

On the other hand, just as in classical computer science \textsc{2-SAT} is solvable in polynomial time, its quantum analogue---the \textsc{Quantum 2-SAT}, a special case of the 2-\textsc{Local Hamiltonian} problem
\footnote{More specifically, \textsc{Quantum 2-SAT} asks whether a sum of 2-local
terms, where each term is a 2-qubit projector that acts on any pair of qubits, is frustration-free, i.e.\ has a 0-energy eigenstate or, equivalently, a state that simultaneously satisfies all local constraints.}---can also
be solved deterministically in
polynomial time: \cite{Bravyi2006} proved an $O(n^4)$ runtime bound, and later a linear-time algorithm was discovered independently by~\cite{Arad2015} and~\cite{DeBeaudrap2016}. Yet the resemblance with classical results goes further: \textsc{Quantum 4-SAT} and later \textsc{Quantum 3-SAT} were shown to be \QMA\!\!\textsubscript{\textsf 1}-complete~\cite{Bravyi2006,Gosset2013}. In the same spirit, a recent result shows that in case of one-dimensional \emph{gapped} local Hamiltonians, there exists an efficient randomized algorithm for approximating the ground state as a matrix product state~\cite{Landau2013} (this result is independent of the local dimension).

However, the  \textsc{Local Hamiltonian} problem, as defined in \cref{def:local-ham-intro}, allows the Hamiltonian to be frustrated (going beyond local projectors), and encompasses Hamiltonians whose gap closes inverse-polynomially in the system's size. It is thus a natural question to ask whether this more general \textsc{Local Hamiltonian} problem remains computationally hard, even under restrictions motivated on \emph{physical} grounds (e.g.\ for translationally-invariant interactions and for qubits), or whether there is a fundamental local dimension threshold below which it becomes tractable.

To motivate this further, it is crucial to note that Hamiltonian constructions in the spirit of \cite{Feynman1986} are a proof-of-concept and may not necessarily be natural, in the sense that we would not encounter them in nature describing an actual physical system. There are three fundamental criteria for judging the ``physicality'' of a Hamiltonian: the interactions should be geometrically local, the dimension of the interacting subsystems should be small, and the interactions should exhibit translational invariance. These properties apply to physical systems we typically encounter in nature. For example, translational invariance means that if the Hamiltonian is specified on a lattice, the interactions are the same independently of the location within the lattice.

Starting with Kitaev's original proof of \QMA-completeness of \textsc{5-local Hamiltonian}~\cite{Kitaev2002}, the locality and local dimension of the constructions were improved successively~\cite{Kempe2006,Oliveira2008,Aharonov2009}, cf.\ \cref{tab:history}. For spins of local dimension 8 coupled by nearest-neighbour interactions on a chain, \QMA-hardness was proven by Hallgren et al.~\cite{Hallgren2013}. All of these results make heavy use of the non-translationally-invariant nature of interactions, which vastly simplify the encoding of the problem instance and verifier circuit into the local structure of the Hamiltonian. The \QMAEXP-hardness result by Gottesman and Irani~\cite{Gottesman2009}, which features a 2-local Hamiltonian on a line with translationally-invariant nearest-neighbour interactions, shows that having translational symmetry does not change the complexity class of the local Hamiltonian problem. But one caveat remains: the local dimension is unphysically large, on the order of $10^6$.

\subsection{Main Result}
Our goal is to significantly improve on this best-known upper bound on the local dimension. We develop a set of new methods to prove that the complexity threshold above which the \textsc{Local Hamiltonian} problem is computationally hard is at most 42, even under the strict physicality constraints outlined above. More precisely, we prove the following main theorem.
\begin{theorem}\label{th:hardness-intro}
  The \textsc{local Hamiltonian} problem with translationally-invariant interactions between neighbouring spins on a chain with local dimension 42 is \QMAEXP-complete. This holds true even for Hamiltonians with local terms of the form $\op h+p(n)\op b$, were $\op h$ and $\op b$ are fixed 2-local interactions and $p(n)$ is a fixed polynomial in the chain length $n$.
\end{theorem}
Following the notation in~\cite{Gottesman2009}, we label this class of problems \TILH[2], for translationally-invariant 2-local Hamiltonian. Analogous to all past hardness constructions, we prove our result by explicitly defining a family of \QMAEXP-hard instances of \TILH[2]. More precisely, the instances we construct are so-called \emph{history state} Hamiltonians: by choosing the local constraints in $\op H$ suitably, one can create a Hermitian operator with a ground state spanned by states that are a uniform superposition over the history of a computation, such that the state at step $t$ is \emph{entangled} with a corresponding state $\ket t$ in a separate time register (i.e.\ $\sum_t\ket t\otimes\ket{\psi_t}$). Measuring the time register at time $t$ then yields the state of the computation at this step. This ``program counter'', as Feynman describes it, can be thought of as a clock or a finite automaton driving the application of quantum gates. Originally, only linearly-evolving clock constructions were used, since analysing the spectrum of a Hamiltonian with branching computational paths is more difficult. More recently, \QMA-hardness constructions in 1D and 2D have used limited branching and cycles~\cite{Hallgren2013,Breuckmann2013}. These have also been exploited in the slightly different context of adiabatic and Hamiltonians quantum computation~\cite{Nagaj2008,Nagaj2012,Gosset2015}.

Whereas recent results~\cite{Piddock2015} make use of perturbation gadgets---approximating higher-order interactions in the low-energy subspace of the system by an effective high-energy theory---it is known that this does not work in one-dimensional systems~\cite{Aharonov2009}. The improvements in \cite{Hallgren2013} over \cite{Aharonov2009} are possible however by approximating 4-local interactions by a sum of 2-local interactions, effectively introducing illegal transitions that have to be penalized. Perturbation gadgets and locality reduction both depend on introducing a large energy scale to project out illegal subspaces. Our results, on the other hand, do \emph{not} use perturbation theory\footnote{The polynomial in \cref{th:hardness-intro} is an artefact of the construction. A standard trick from \cite{Gottesman2009} can reduce the Hamiltonian to fixed $O(1)$ interactions by slightly increasing the local dimension, see \cref{rem:no-scaling}.}.

Our findings are based on the following three main technical contributions:
\begin{enumerate}
  \item All previous constructions encode one of the standard models of quantum computation (almost always the circuit model, with the exception of \cite{Gottesman2009} which encodes a quantum Turing machine), which are not optimized for this task.  We design a new universal model of quantum computation---a \emph{quantum ring machine} (QRM)---which we prove to be quantum Turing-complete. The periodicity of the QRM's computational steps make it particularly well-suited for local Hamiltonian constructions.
  \item We next introduce \emph{unitary labelled graphs} and their associated Hamiltonians, which can accommodate a non-deterministic clock construction to drive quantum computation. This generalizes Feynman's original clock construction~\cite{Feynman1986}, which corresponds to a path graph in our setup. Mirroring Kitaev's analysis~\cite{Kitaev2002}, our Hamiltonian is also equivalent to a Laplacian of the corresponding graph, which allows us to analyse its spectrum using a combination of spectral graph theory and matrix analysis techniques. These techniques let us analyse ground states of more complicated Hamiltonians than previously possible.
  \item We define yet another computational model---a \emph{quantum Thue system}, or a quantum string rewriting system---that on the one hand is particularly well-suited for embedding a computational model into local interactions of a Hamiltonian; and on the other hand, under simple \emph{local} constraints on the rewriting rules, necessarily produces Hamiltonians that correspond to unitary labelled graphs. Quantum Thue systems can in a sense be thought of as an assembly language for compiling computational models into local translationally-invariant Hamiltonians, which could also be used for adiabatic quantum computation, or Hamiltonian quantum computers (cf.~\cite{Nagaj2008,Wei2015}).
\end{enumerate}
In light of our result being rather involved and technical, we want to give a poor man's overview of our findings, which---without any proofs---outline the technical contributions in this paper. We want to emphasize that we made an effort to keep each section largely self-contained; in particular the section on spectral analysis of graphs with unitary edge labels, quantum ring machines, and quantum Thue systems can be regarded independently of each other. The \QMAEXP-hardness proof in \cref{sec:hardness} of course utilizes all of our developed machinery, but in such a way that the proof of existence of \QMAEXP-hard instances themselves are given a separate section.

Since the latter part is somewhat technical and specific, we want to point out that one does not need to understand the construction itself to follow the idea behind the hardness proof, which hopefully facilitates an understanding of the result.\footnote{
  The reason behind unhitching the explicit construction of \QMAEXP-hard instances in this way is to allow for further optimization of the local dimension to go through without having to re-prove all of the claims; in fact, we encourage the interested reader to have a stab at finding a quantum Thue systems following the four properties given in \cref{lem:qts-special}, but with an alphabet that is smaller than ours.
}

\subsection{Proof Ideas and Techniques}
\paragraph{Spectral Analysis for Hamiltonians Encoding Non-Deterministic Computation.}
As briefly explained in the introduction, the fundamental idea behind encoding quantum computation into the ground state of a Hamiltonian is based on the concept of \emph{history states}, introduced by Feynman in 1986. For some quantum circuit represented by local gates $\op U_1,\ldots,\op U_T$ on a Hilbert space \hs, we define a Hamiltonian on the product space $\sset C^T\otimes\hs$ as
\begin{equation}
  \op H:=\sum_{t=1}^{T-1}(\proj t\otimes\1 + \proj{t+1}\otimes\1 - \ketbra{t+1}{t}\otimes\op U_t - \ketbra{t}{t+1}\otimes\op U_t^\dagger).
\end{equation}
The ground state of this Hermitian operator is spanned by states of the form $\sum_t\ket t\otimes\ket{\psi_t}$, where $\ket{\psi_t}=\op U_t\cdots\op U_1\ket\phi$ for some $\ket\phi\in\hs$. For any $\ket\phi$, $\ker\op H$ thus encodes the uniform superposition over the history of the quantum circuit acting on $\ket\phi$.
An intuitive way of thinking about these ground states is that they represent quantum computation driven by a clock, i.e.\ for each increment of the clock register, the corresponding quantum gate is applied to the computational register.\footnote{The notion of time in this context is meaningless, but simplifies an intuitive understanding on how computation is embedded into the ground state of $\op H$.}

Essentially all past result employ such history state Hamiltonians with a linear clock, i.e.\ for every computational step, there exists precisely one unique forward and backward transition.  For local Hamiltonian constructions---i.e. where $\op H$ is a sum of local terms---this implies that each \emph{local} rule has to know the exact location within the overall computation.

To be more specific, consider a spin chain of length $n$ as the Hilbert space $\hs_\text{loc}^{\otimes n}$. The interactions on this chain then take the form of a set of local \emph{rewriting rules} acting on neighbouring sections of spins: for $\ket{\psi_i}, \ket{\phi_i} \in \hs_\text{loc}^{\otimes k}$ for some constant $k<n$,  we encode the evolution $\ket{\psi_i} \mapsto \ket{\phi_i}$ by a local Hamiltonian term $\op h_i=(\ket{\psi_i}-\ket{\phi_i})(\bra{\psi_i}-\bra{\phi_i})$. The overall Hamiltonian is then a sum of these local interactions over all spins, i.e.
\begin{equation}\label{eq:hist-hamiltonian-intro}
  \op H=\sum_j\1_{1,\dotsc,j-1}\otimes\Big(\sum_i\op h_i\Big)_{j,\ldots,j+k-1}\otimes\1_{j+k,\dotsc,n}.
\end{equation}
If the global evolution defined by the terms $\op h_i$ is unique, this implies that it is always possible to locally determine the global state of the computation.\footnote{
  Hamiltonians such as in \cref{eq:hist-hamiltonian-intro} are combined with a series of local projectors which single out a computationally valid ground state, so strictly speaking the local rules will only have to discriminate the current computational state locally within this valid subspace---cf.~\cite{Hallgren2013}, where this is exploited to break down $4$-local interactions to $2$-local ones.
}
This means that locally, we have to store this state in one way or another: under this requirement it is difficult to push the limits of local Hilbert space dimension down, and much could be gained if we could e.g.\ allow the local computational state to be ambiguous to some extent (but such that if the wrong transition is applied, the computation does not proceed to tamper with the actual outcome of the embedded circuit).

\emph{In our work, we go beyond linear clock constructions, and prove a series of spectral graph-theoretic results which allow us to analyse more complicated history state Hamiltonians.} We outline these novel techniques below.

If all $\ket{\psi_i}$ and $\ket{\phi_i}$ in \cref{eq:hist-hamiltonian-intro} are standard basis vectors, then each rule corresponds to an edge in a graph $G$ with vertices labelled by the canonical basis of the spin chain. $\op H$ thus equals the Laplacian of the graph $G$  (whose spectrum is accessible) and the ground state of $\op H$ is given by the uniform superposition over connected graph components of $G$. We call ground states of such Hamiltonians as in \cref{eq:hist-hamiltonian-intro} \emph{history states}, since they encode the closure of states reachable under the given rewriting rules.

To analyse the spectrum of more general non-basis transitions $\ket{\psi_i}\mapsto\ket{\phi_i}$, one needs to prove that this choice still allows $\op H$ to be at least unitarily equivalent to a graph Laplacian $\Delta$, e.g.\ by explicitly constructing a unitary similarity transform $\op W$ such that $\op W^\dagger\op H\op W=\Delta\otimes\1$. Most if not all \QMA-hard construction since Kitaev's go along this route; however, in the language of graphs, the unitary equivalence could only be proven if $\Delta$ is the Laplacian of a path graph.\footnote{
  A more complex construction with a local clock was considered in~\cite{Breuckmann2013}, where the authors consider a 2D surface and allow executing transitions in parallel, as long as the execution front behaves in a time-like fashion. To analyse the spectrum of the resulting Hamiltonian, they relate the propagation terms to the diffusion of a string on a torus, corresponding to a ferromagnetic Heisenberg model with partially twisted periodic boundary conditions. Their analysis, while elegant, is specific to their string diffusion-type execution order of quantum gates. These Hamiltonians cannot generally be translationally-invariant, as the circuit must be laid out on the 2D surface.
}
Just as in \cref{eq:hist-hamiltonian-intro}, this graph essentially corresponds to the finite state automaton ``driving'' the computation; if it is a path graph, the computational path is limited to a sequential application of transition rules $\ket{\psi_i}\mapsto\ket{\phi_i}$ or gate applications encoded therein.

We extend this notion to allow much more complicated branching in the computational path to occur. In particular, we prove a series of results which guarantee the existence of the partially diagonalizing unitary $\op W$ solely based on properties of the rewriting rules, without the need to explicitly analyse the overall evolution of the system. This has two major benefits: it allows more powerful state transitions which are not necessarily unique for every step, and it drastically simplifies the spectral analysis of $\op H$ for whichever construction we choose to work with, as we do not need to construct the equivalence between $\op H$ and $\Delta$ explicitly. As an important example, our model is the first to allow multiple threads of computation to run in parallel, which then join at some common state.

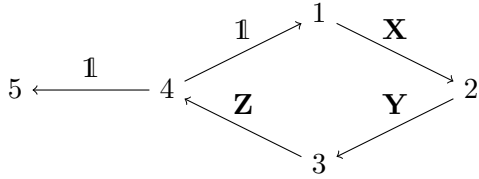
\begin{figure}
  \centering
  \begin{tikzpicture}[node distance = 14mm, text height = 1.5ex, text depth = .25ex]
    \node (5) at (0,0) {5};
    \node (4) at (2,0) {4};
    \node (1) at (4,1) {1};
    \node (3) at (4,-1) {3};
    \node (2) at (6,0) {2};
    \draw[->] (1) -- node[above] {$\op X$} (2);
    \draw[->] (2) -- node[above] {$\op Y$} (3);
    \draw[->] (3) -- node[above] {$\op Z$} (4);
    \draw[->] (4) -- node[above] {$\1$} (1);
    \draw[->] (4) -- node[above] {$\1$} (5);
  \end{tikzpicture}
  \caption{Example of a unitary labelled graph (ULG) with vertices $\{1,2,3,4,5\}$ and three non-trivial unitaries $\op X, \op Y,\op Z\in\set{U}(\mathcal H)$. The associated Hamiltonian for this ULG, as defined in \cref{eq:ulg-hamiltonian-intro}, is unitarily equivalent to the Laplacian of the underlying graph if $\op Z\op Y\op X=\1$. We provide an explicit description of this change-of-basis unitary.}
  \label{fig:ulg-intro}
\end{figure}
In a bottom-up approach, we formalize the notion of a Hamiltonian associated with a graph. Starting from a simple directed graph $G = (\set V, \set E)$ we associate a Hilbert space $\hs$ to each vertex $v \in \set V$, and a unitary $\op U_{(a,b)}:\hs\longrightarrow\hs$ for every directed edge $(a,b) \in \set E$.  We call such a graph with Hilbert space and family of unitaries a \emph{unitary labelled graph}, or ULG for short. As an example, consider \cref{fig:ulg-intro}.

The \emph{associated Hamiltonian} for the ULG is then defined as
\begin{equation}\label{eq:ulg-hamiltonian-intro}
  \op H(G):=\sum_{(a,b)\in \set E}\sum_i(\ket a\otimes\ket i-\ket b\otimes\op U_{(a,b)}\ket i)(\mt{herm.~conj.}),
\end{equation}
where the $\ket i$ label a basis of \hs. Observe that this construction is more general than a local Hamiltonian on a spin chain as in \cref{eq:hist-hamiltonian-intro}: $\op H(G)$ is simply regarded as Hermitian operator on the overall Hilbert space $\sset C^{\set V}\otimes\hs$ where the vertex labels are arbitrary, and not necessarily make $\op H(G)$ local in any sense.

The associated Hamiltonian $\op H(G)$ bears some structural resemblance with a graph Laplacian, as already mentioned. We prove the following theorem.
\begin{theorem}\label{th:simple-semiclassical-intro}
  If the product of unitaries along any loop in the graph $G$ is $\1$, a property we call \emph{simple}, then $\op H(G)$ is unitarily equivalent to $\Delta\otimes\1_n$, where $n=\dim\hs$ and $\Delta$ is the Laplacian of $G$.
\end{theorem}
\Cref{fig:ulg-intro} satisfies this theorem if and only if the product of unitaries in the loop are $\op Z\op Y\op X=\1$. We provide an explicit expression for this diagonalizing unitary, which can be constructed in poly time using a breadth-first search algorithm along a spanning tree of $G$.

\paragraph{Quantum String Rewriting.}
In order to reintroduce locality to our Hamiltonian construction, we further develop a notation which facilitates embedding transition rules as in \cref{eq:hist-hamiltonian-intro,eq:ulg-hamiltonian-intro} into the ground state of a local Hamiltonian.  This notation is heavily motivated by string rewriting models, and we extend this notion to introduce a new quantum Turing-complete model based on transitions able to perform quantum gates on part of the string's alphabet.

As mentioned, past hardness constructions (summarized in \cref{tab:history}) encode computation in \emph{local} transition rules that act on spins connected by some underlying graph of interactions. While some of these transitions are classical---i.e.~basis-preserving---others act on the spin states with a non-diagonal unitary operator, performing the actual quantum computation. Inspired by classical string rewriting systems, we interpret these quantum interactions as local quantum rewriting rules, and introduce a new abstract rewriting system called \emph{quantum Thue system}. This extends an already-existing model of string rewriting---semi-Thue systems%
\footnote{Named after the Norwegian mathematician Axel Thue. We require all rule sets for quantum string rewriting to be symmetric; a symmetric semi-Thue is simply called Thue system, explaining the name quantum Thue system.}%
---which are well-studied classically~\cite{Thomas2010}.

A (classical) semi-Thue system consists of a finite alphabet $\Sigma$ and length-preserving replacement rules for strings over this alphabet. Similar to the word problem, computation can be encoded in the question whether there exists a connecting path between some input and output strings $s_i$ and $s_f$. It is straightforward to simulate universal classical Turing machines with a Thue system, which shows that the latter is a Turing-complete model for classical computation. But what about quantum computation?

For quantum Thue systems, we require that the alphabet splits into a classical and a quantum part, i.e.~$\Sigma=\Sigma_{cl}\sqcup\Sigma_q$. Transition rules can be purely classical---between elements of $\Sigma_{cl}^*$, quantum---between elements of $\Sigma_q^*$, or a mixture thereof, in which case we require that the rule preserves the number of quantum symbols $\qn{s}$ of a string $s\in\Sigma^*$. In addition, every rule $r$ acting on at least some quantum symbols has a unitary $\op U_r\in\U(\hs^{\otimes \qn{s}})$ attached, where \hs is some fixed, finite-dimensional Hilbert space.

Starting on some string $s$ and a state vector $\ket v\in\hs^{\otimes \qn{s}}$, we apply any matching string rewriting rule $s\overset r\longmapsto s'$ in turn. For every replacement, we \emph{also} apply the corresponding unitary to the state vector, i.e.~$\op U_r\ket v=\ket {v'}$. In this fashion, we can model quantum computation, driven by a finite automaton: if we make the underlying classical Thue system implement a Turing machine that writes out a quantum circuit description on the string, and then perform this quantum circuit on a separate set of qubits attached to some quantum symbols, the final state vector will contain the output of a quantum computation.

One can then show that a quantum Thue system is itself a special case of a unitary labelled graph, which allows us to translate it into a Hamiltonian. We show that the locality of the resulting Hamiltonian only depends on the range of the largest replacement rule, e.g.~if one at most replaces a 3-character string, the resulting Hamiltonian will also be 3-local and translationally-invariant.

As replacement rules are not necessarily unique, the computation will have potential ambiguities. As such, we regard all strings connected to the initial starting string $s_i$ via some arbitrary combination of rules, and the size of this set corresponds to the number of basis states that the corresponding \emph{history state} (the ground state of the associated Hamiltonian of the unitary labelled graph defined by the quantum Thue system) is comprised of.

\paragraph{A Simpler Computational Model.}
The complexity class \QMAEXP is usually defined in terms of the circuit model, i.e.~as a uniform family of verifier circuits: a promise problem $\PromP = (\PromP_\yes, \PromP_\no)$ is in \QMAEXP if there exists a classical Turing machine, such that the verifier circuit for a problem instance $l \in \PromP_\yes \cup \PromP_\no$ can be written out by the Turing machine in $O(\exp |l|)$ steps where $|l|$ is the instance size. Being used as an all-purpose computational model, Turing machines have significant downsides: they have complicated transition functions, need a lot of internal states (which translates to an enormous local dimension when encoded in a Hamiltonian) and are rarely written out explicitly (so it is hard to get tight bounds on the required dimension). On the other hand, in past constructions, embedding a circuit directly required the use of non-local clock states marking the position within the circuit, or non-translationally-invariant terms that encode the circuit unambiguously.

\begin{figure}[!t]
  \begin{floatrow}
    \ffigbox[6cm]{
      \centering
      \includegraphics[width=5.5cm]{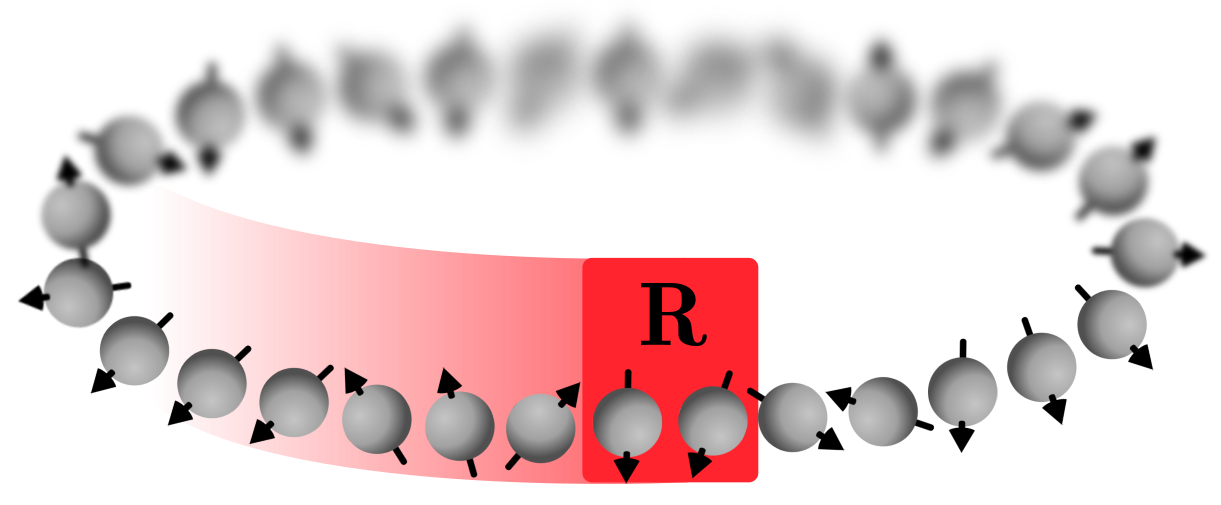}
    }{\caption{Schematic of a quantum ring machine. A fixed unitary $\op R$ is cyclically applied to a ring of qudits until one of the qudits indicates a halting configuration.}\label{fig:qrm-intro}}
    \ffigbox[\textwidth - 6.4cm]{
      \def\W{8}
      \def\w{1}
      \def\h{0.8}
      \def\gx{0.4}
      \def\gy{0.65}

      \begin{tikzpicture}[
          box/.style = {fill=white,draw = black,rounded corners=.5,line width=.8pt},
          snake/.style = {},
          snake2/.style = {decorate, decoration = {zigzag, amplitude = 4pt, segment length = 25pt}}
        ]

        \clip [snake] (-0.2,0.5*\h) rectangle (\W+0.2,-3.5*\h);
        \clip [snake2] (.2,2*\h) rectangle (\W-.2,-5*\h);

        \newcommand{\gate}[2]{
          \path[box] (#1*\w-\gx,-#2*\h-0.5*\h-\gy) rectangle (#1*\w+\gx,-#2*\h-0.5*\h+\gy);
          \node at (#1*\w,-#2*\h-0.5*\h) {$\op R$};
        }

        \foreach \i in {0,...,4} {
          \draw[line width=.8pt,double] (0,-\i*\h-.1) -- (\W,-\i*\h-.1);
          \draw[line width=.8pt] (0,-\i*\h+.1) -- (\W,-\i*\h+.1);
        }

        \gate{0}{3}
        \gate{0}{-1}
        \gate{1}{0}
        \gate{2}{1}
        \gate{3}{2}
        \gate{4}{3}
        \gate{4}{-1}
        \gate{5}{0}
        \gate{6}{1}
        \gate{7}{2}
        \gate{8}{3}
        \gate{8}{-1}
      \end{tikzpicture}
    }{\caption{Ring machine's evolution implementing a uniform quantum circuit. Double lines carry classical while single lines carry quantum information. Classical wires encode where the next quantum gate from a small universal set will be applied.}\label{fig:qrm-evo-intro}}
  \end{floatrow}
\end{figure}

We introduce a new computational model which allows us to circumvent the direct use of complicated Turing machines or quantum circuits. The so-called \emph{quantum ring machine} consists of a cyclic ring of \emph{qudits} (i.e.\ $d$-dimensional quantum systems) and a unitary $\op R$ describing a head that acts on two qudits at a time. At each time-step, the head moves in the same direction along the ring and cyclically acts on adjacent cells. We give the following definition (see \cref{sec:qrm} for more details).

\begin{definition}[Quantum ring machine]
  A \emph{quantum ring machine} consists of a ring of $n$ qudits, each of dimension $d$, and a unitary operator $\op R$ acting on a pair of qudits.
  The $n$-qudit ring is initialized in state $\qi$ and the machine proceeds by applying $\op R$ cyclically to pairs of adjacent qudits along the ring---see figure~\ref{fig:qrm-evo-intro}---until one of the qudits indicates halting: its reduced density matrix has support completely inside a certain halting subspace $\hsf$, while the reduced states of all qudits up to this point were orthogonal to $\hsf$.
\end{definition}

To show that a quantum ring machine is computationally equivalent to a uniform family of quantum circuits, we encode a classical Turing machine's transition function into $\op R$, where the internal states, including the Turing machine's halting flag, are stored as a classical information on the ring. Such ring machine can be used to write out and execute a quantum circuit ``on-the-go'': it is universal for whichever uniform circuit class is encompassed by its allowed runtime. Quantum ring machines thus bridge the gap between circuits, which are particularly simple to specify locally but have a complex global structure, and Turing machines, which are difficult to specify locally due to a possibly large number of internal states, but have a straightforward global evolution as the tape only changes in at most one location at each step. A schematic of the ring machine can be found in \cref{fig:qrm-intro,fig:qrm-evo-intro}.

The ring machine's simple mechanism allows its evolution to be described by a set of local quantum rewriting rules. These rules operate at a \emph{physical} level while the ring machine operates at a \emph{logical} level---each application of ring machine's head $\op R$ on a pair of logical qubits is implemented by a sequence of physical operations acting on a much larger number of qubits. At any given time the ring machine's head is positioned on a specific pair of logical qubits, and after each application of $\op R$ this location is updated in a similar fashion as Turing machine's head---it is shifted either up or down along the ring by one position. Overall, $\op R$ is a large controlled unitary that acts at a given logical location only if the ring machine's internal state---stored as a classical bit on the physical tape---is in an active configuration.

\paragraph{QMAEXP hardness of \TILH[2].}

The final proof of \cref{th:hardness-intro} is based on the following lemma.
\begin{lemma}\label{lem:qts-special-intro}
  There exists a \BQEXP-universal quantum Thue system with 39 symbols, 3 of which are quantum, with attached Hilbert space $\sset C^2$ and 2-local rules.
\end{lemma}
We prove this by writing out a quantum Thue system which executes a \BQEXP-universal quantum ring machine. The quantum Thue system makes heavy use of the new possibilities of ambiguous replacement rules, which allow the history state path to branch. For the \QMAEXP hardness proof itself we combine this Thue system with a series of local penalty terms, which allow us to single out the history state as lowest-energy ground state for any encoded \yes instance.

Furthermore, we prove that the quantum Thue system has a simple history state in the sense of \cref{th:simple-semiclassical-intro}, which allows us to analyse the spectrum of the resulting Hamiltonian. More specifically, we prove a variant of \emph{Kitaev's geometrical lemma} (cf.\ \cref{lem:kitaev-graphs,lem:kitaev-ulgs}) which facilitates the spectral analysis of Hamiltonians that are sums of a unitary labelled graph Hamiltonian and local projectors. This finally allows us to prove our main result, \cref{th:hardness-intro}, that \textsc{2-Local-Hamiltonian} is \QMAEXP-hard, even for translationally-invariant nearest-neighbour interactions between spins of local dimension 42.

For completeness, we also want to give a brief overview over the family of hard QTS instances that we construct, but---as mentioned before---the \QMAEXP hardness proof does not depend on the precise workings of it; assuming that \cref{lem:qts-special-intro} can be proven, \cref{th:hardness-intro} stands independently.

Treating the Hilbert space of the \TILH[2] problem as a physical tape of length $n$---some symbols quantum, some classical---we write a set of transition rules to perform the following steps that simulate the quantum ring machine.

\begin{figure}[t]
  \centering
  \subfigure{
    \includegraphics[width=0.47\linewidth]{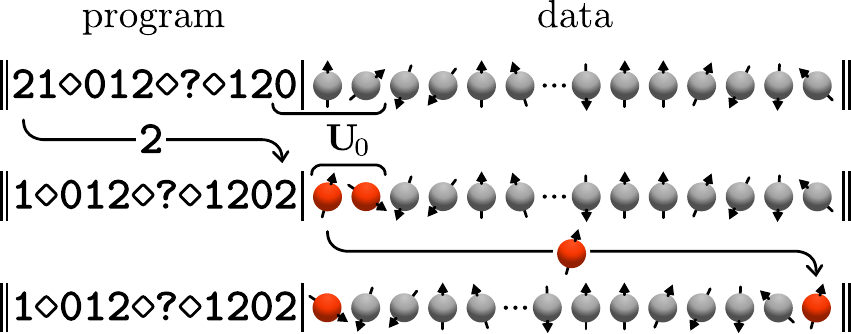}
  }
  \subfigure{
    \includegraphics[width=0.47\linewidth]{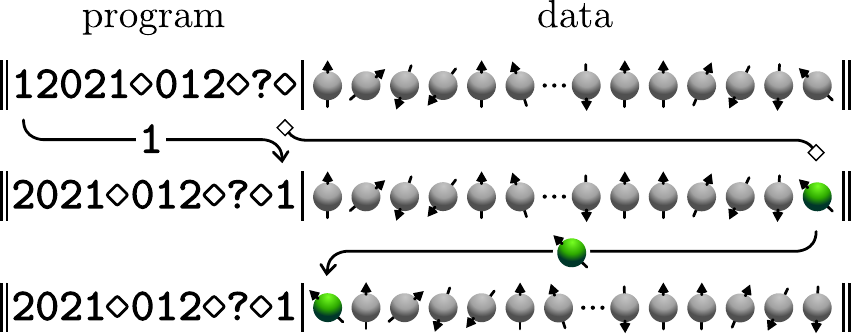}
  }
  \caption{
    Illustration of \emph{Turing's wheelbarrow} construction (see \cref{sec:wheelbarrow} for more details). It consists of a tape that stores a program string on the left- and data qubits on the right-hand side. Two types of actions are supported: application of a quantum gate (left figure) and rewinding of the tape (right figure). The rightmost program bit always indicates the next action. For example, {\${:0}} indicates that a unitary gate $\op U_0$ should be applied to the two leftmost data qubits, and the ring of qubits should then be cyclically rotated one position to the left. On the right, the action of the special symbol {\$r} is depicted: it signals the rightmost qubit to move back to the left end of the tape. After each action the program string is cyclically rotated one position to the right.
  }
  \label{fig:wheelbarrow-gates-intro}
\end{figure}
\begin{enumerate}
  \item As in the construction by Gottesman and Irani~\cite{Gottesman2009}, we use a counter to translate the chain length $n$ into a \emph{program} string of length $O(\log n)$ on the left hand side of the chain, while on the right hand side we store the physical data qubits, i.e.\ the ring of qubits our ring machine is executed on.
  \item The program on the left hand side contains a physical-level description of a quantum circuit (over a small, finite, universal gate set) for implementing one step of the quantum ring machine, i.e.~one application of the ring machine's head $\op R$. The program's rightmost bit always indicates the next gate in the circuit, and this gate is always applied to the two leftmost data qubits on the physical ring (see \cref{fig:wheelbarrow-gates-intro}).
  \item Using the two types of basic commands---``apply gate'' and ``rewind tape''---shown in \cref{fig:wheelbarrow-gates-intro}, the quantum circuit implementing $\op R$ can be executed cyclically on the physical data qubits, some of which are initialized to ancillary $\ket 0$'s to be used in the computation.
  \item The computation runs until a certain internal classical counter (stored on the ring) terminates. In our construction, we explicitly encode transitions for the gates \swap, \toffoli and a controlled quantum-universal unitary; since \swap and \toffoli are also universal for classical computation, the classical control machinery in the ring machine's head $\op R$ (i.e.\ the Turing machine used to write out the quantum circuit) can be executed exactly (without error). This means that the computation will halt deterministically (as otherwise there could be some overlap with a non-halting state). The transition rules for applying a gate as in \cref{fig:wheelbarrow-gates-intro} then have another control gate which only proceeds if the data bit to its right is in a specific configuration, terminating the machine's execution otherwise.
  \item The length of the chain is chosen so that the program encodes a quantum ring machine equivalent to a \BQEXP verifier circuit. It discriminates between \yes and \no instances of the corresponding \QMAEXP language depending on whether the ring machine accepts or rejects, and a special symbol in the program description allows us to locally penalize a wrong initialization of ancillas and a \no output of the computation.
\end{enumerate}

Our construction is universal in the sense that it can be used to implement an arbitrary quantum computation without the need to increase the local dimension (in the same spirit as a universal Turing machine can implement any computation without the need to increase the number of internal states). Since we leave parts of the input unconstrained, we conclude from \BQEXP-completeness of these instances that they can be used as a \QMAEXP verifier, finalizing our claims.

\subsection{Structure of the Paper}
We summarize several standard definitions in \cref{sec:prelims}. In \cref{sec:qrm}, we define the aforementioned quantum ring machine and show that it is indeed Turing-complete for quantum computation. \Cref{sec:qts} formalizes the notion of quantum replacement rules and introduces the model of quantum Thue systems. \Cref{sec:wheelbarrow} contains a constructive proof of a universal quantum Thue system, and \cref{sec:hardness} combines everything into our main hardness result.

\section{Preliminaries}\label{sec:prelims}
\subsection{Reversible Turing Machines}

We give the following standard definition of a (non-deterministic) Turing machine (for more background on Turing machines, see chapter~8 of~\cite{Hopcroft1979}).

\begin{definition}[Turing machine]
  A \emph{Turing machine}---or TM for short---is a triple $(\set Q,\Sigma,\delta)$, where $\set Q$ is a finite set of internal states containing a distinct initial and halting state $q_0$ and $q_f$, respectively, and $\Sigma$ is a finite set of tape symbols containing a designated blank symbol $0$. Let $\set D := \{\mathrm{left},\mathrm{right}\}$ be the two possible movement directions of the TM's head. Then each element of the transition set $\delta \subseteq \set Q \times \Sigma \times \Sigma \times \set D \times \set Q$ is a quintuple of the form $(q,s,s',d,q')$, which means that if the Turing machine reads a symbol $s$ under its head while in state $q$, it overwrites the symbol by $s'$, moves the head in direction $d \in \set D$ and transitions to state $q'$. At the beginning of the computation, the TM's initial state is $q_0$ and the tape is initialized to all $0$s, except for a finite block of consecutive cells containing the input. The machine \emph{halts} once its internal state is $q_f$, for which there is no forward transition.
\end{definition}

As we aim to implement TMs using quantum mechanics, we need them to be deterministic and reversible. The following is based on definition~10 from~\cite{Mor08}.

\begin{definition}[Deterministic and reversible Turing machine]\label{def:DetRevTM}
  Consider a Turing machine $(\set Q,\Sigma,\delta)$, and let $(q_1,s_1,s'_1,d_1,q'_1)$ and $(q_2,s_2,s'_2,d_2,q'_2)$ be any two distinct quintuples in $\delta$. This TM is
  \begin{itemize}
    \item \emph{deterministic} if $(q_1 = q_2) \, \Longrightarrow \, (s_1 \neq s_2)$,
    \item \emph{reversible} if $(q'_1 = q'_2) \, \Longrightarrow \, (s_1' \neq s_2') \land (d_1 = d_2)$.
  \end{itemize}
\end{definition}

The first condition of \cref{def:DetRevTM} rules out the possibility that $q_1 = q_2$ and $s_1 = s_2$, meaning that the current TM's state and tape symbol should unambiguously determine the rest of the transition. Similarly, the second condition rules out the possibility that $q'_1 = q'_2$ and $s'_1 = s'_2$, as well as the possibility that $q'_1 = q'_2$ and $d_1 \neq d_2$, meaning that the reverse transition also is uniquely determined by the current state and tape symbol, and that the direction of the TM's head movement in reverse is uniquely determined by the current state.

%\begin{definition}[Unidirectional Turing machine]\label{def:UniTM}
%A Turing machine $(\set Q,\Sigma,\delta)$ is \emph{unidirectional} if each state can be entered only from on direction, i.e.~if $(q_1,s_1,s'_1,d_1,q')$ and $(q_2,s_2,s'_2,d_2,q')$ are two quintuples then $d_1 = d_2$.
%\end{definition}

For a deterministic TM, one can regard $\delta$ as a partial function, namely $\delta: \set Q \times \Sigma \to \Sigma \times \set D \times \set Q$, since all combinations of internal state $q$ and tape symbol $s$ have at most one forward transition. For a reversible TM, $\delta$ is injective since all combinations of internal state $q'$ and tape symbol $s'$ have at most one backwards transition (whenever such transition exists, it uniquely determines the head movement direction $d$ backwards). In fact, according to \cref{def:DetRevTM}, each state of a reversible TM can be entered only from one direction (this property is referred to as \emph{unidirectionality} in~\cite{BV97}). In other words, it is sufficient to know only the TM's current state (as opposed to both the state and the tape symbol) to answer the question ``From which direction did the TM's head arrive?''.

Due to unidirectionality, it is often natural to restrict the range of $\delta$ to $\set Q \times \Sigma$. In fact, the transition function $\delta$ of a deterministic reversible Turing machine can be replaced by a permutation matrix on $\set Q \times \Sigma$ without affecting the TM's behaviour. For our convenience, we state this observation more formally (see also cor.~B.2 and thm.~4.2 in~\cite{BV97}).

\begin{lemma}\label{thm:permutation}
  For any deterministic reversible Turing machine $(\set Q,\Sigma,\delta)$, the partial transition function $\delta$ can be replaced by a pair $(\op T_\delta, d)$, where $\op T_\delta$ is a permutation matrix on $\set Q \times \Sigma$ and $d: \set Q \to \set D$ is a function that determines, for each internal state $q \in \set Q$, the direction from which the TM's head arrived in $q$. If we update the TM's internal state and the current tape symbol according to $\op T_\delta$, and then move the TM's head in the direction \emph{opposite} to $d(q')$, where $q'$ is the updated state, the behaviour is identical to the original transition function $\delta$.
\end{lemma}

\begin{proof}
  The function $d$ is readily obtained because of unidirectionality. A blueprint of $\op T_\delta$ is obtained by restricting the range of $\delta$ to $\set Q \times \Sigma$ and describing $\delta$'s action on the elements of this set by a binary matrix. Since the TM is deterministic and reversible, this matrix contains at most one entry $1$ in each row and column, so it can be easily extended to a permutation matrix.
\end{proof}

From now on we will consider only deterministic Turing machines and implicitly assume that they are reversible---this is justified by the following result due to Bennett~\cite{Ben73} (see \cite{Per13,Mor08} for more background on reversible computation).

\begin{theorem}[Bennett~\cite{Ben73}]\label{rem:reversible-tm}
  Any deterministic TM can be made reversible with at most polynomial overhead in terms of space and time.
\end{theorem}

\subsection{Quantum Complexity Classes}\label{sec:classes}
In this section, we formally define the quantum complexity classes \BQP, \BQEXP, \QMA and \QMAEXP in terms of the circuit model, and refer reader to~\cite{BV97,Wat12,VW16} for more details on quantum computational complexity.

%http://www.cs.berkeley.edu/~vazirani/f04quantum/notes/lec4.pdf

In what follows, we fix some finite universal set of 2-qubit quantum gates, such as $\{\hadamard,\allowbreak\cnot, R(\pi/4)\}$---see~\cite[ch.~4.5]{Nielsen2010}. We first define a uniform family of quantum circuits over this gate set.

\begin{definition}[Uniform family of quantum circuits]\label{def:circuits}
  Let $f: \sset N \to \sset N$ be a function and $(C_n)_{n\in\sset N}$ be a family of quantum circuits where each $C_n$
  \begin{itemize}
    \item acts on $n$ qubits and has a distinct output qubit,
    \item requires at most $f(n)$ additional ancilla qubits initialized in $\ket{0}$,
    \item contains at most $f(n)$ gates from our universal set.
      %nonessential comment: strictly speaking all these can be different functions, but we just take the biggest one.
  \end{itemize}
  We say that $(C_n)_{n\in\sset N}$ is \emph{$f(n)$-uniform} if there exists a TM that on input $1^n$ produces an explicit description of $C_n$ in less than $f(n)$ steps.
\end{definition}
%nonessential comment: Observe that the input requirement $1^n$ for the circuit is chosen arbitrarily, but has to be uniform for all circuits.

Let $\Sigma$ be a finite set (\emph{alphabet}), and let $\Sigma^n$ and $\Sigma^* := \bigcup_{n \geq 0} \Sigma^n$ denote the sets of all length-$n$ and all finite-length strings over $\Sigma$, respectively. A \emph{promise problem} over alphabet $\Sigma$ is a pair $\PromP = (\PromP_\yes, \PromP_\no)$ such that $\PromP_\yes \cap \PromP_\no = \varnothing$, where $\PromP_\yes, \PromP_\no \subseteq \Sigma^*$ are the sets of input strings corresponding to \yes and \no instances, respectively. We will sometimes write $l \in \PromP$ meaning that $l \in \PromP_\yes \cup \PromP_\no$.

\begin{definition}[Complexity class \BQ]\label{def:bqp}
  A promise problem $\PromP=(\PromP_\yes,\PromP_\no)$ is in \BQ, \emph{bounded-error quantum $f(n)$-time}, if there exists an $f(n)$-uniform family of quantum circuits $(C_n)_{n\in\sset N}$ such that
  \[
    \Pr(C_n(s)=\yes)\ge \frac23 \mt{for $s\in\PromP_\yes$\quad and}
    \Pr(C_n(s)=\yes)\le \frac13 \mt{for $s\in\PromP_\no$,}
  \]
  where $C_n(s)$ denotes the random variable obtained by executing $C_n$ on input $s \in \PromP$ of size $\abs{s} = n$ and measuring the output qubit (the encoding of $s$ as well as the measurement are performed in the computational basis).
\end{definition}

We introduced the notation \BQ to emphasize the fact that the definitions of classes \BQP and \BQEXP are essentially the same up to the bounding function:
\[
  \BQP:=\bigcup_{k\in\sset N}\BQ[n^k] \mt{and} \BQEXP:=\bigcup_{k\in\sset N}\BQ[\exp(n^k)].
\]
Trivially, $\BQP \subseteq \BQEXP$ since a longer runtime can only help.

It is well-known (see~\cite[Prop.~3]{Wat12}) that for \BQP the probabilities of $2/3$ and $1/3$ in \cref{def:bqp} can be exponentially amplified while still remaining in the same complexity class. The same argument works for \BQEXP as well, since we only need a polynomial number of repetitions to achieve the desired amplification.

\begin{fact}[Error-reduction for \BQP and \BQEXP]\label{rem:bqerror}
  For any polynomial $p$, we can assume that $\Pr(C_n(s)=\yes)\ge1-2^{-p(n)}$ for $s \in \PromP_\yes$ and $\Pr(C_n(s)=\yes)\le 2^{-p(n)}$ for $s \in \PromP_\no$ in the definitions of \BQP and \BQEXP.
\end{fact}

Intuitively, \QM is the class of promise problems for which the \yes/\no answers can be verified by a \BQ verifier.

\begin{definition}[Complexity class \QM]\label{def:qma}
  A promise problem $\PromP=(\PromP_\yes,\PromP_\no)$ is in \QM, \emph{$f(n)$-time quantum Merlin-Arthur}, if there exists an $f(n)$-uniform family of verifier quantum circuits\footnote{Here we use a slight variation of \cref{def:circuits}: we also allow for at most $f(n)$ extra input qubits to store the witness state $\rho$ (this is in addition to the $n$ original input qubits and $f(n)$ ancillary qubits that are initialized in $\ket{0}$).} $(C_n)_{n\in\sset N}$ such that
  \begin{itemize}
    \item if $s \in \PromP_\yes$, $\exists$ a witness state $\rho$ on at most $f(n)$ qubits such that $C_n(s,\rho) = \yes$ with probability at least $2/3$. This condition is known as \emph{completeness}.
    \item if $s \in \PromP_\no$, $\forall$ witness states $\rho$ on at most $f(n)$ qubits $C_n(s,\rho) = \yes$ with probability at most $1/3$. This condition is called \emph{soundness}.
  \end{itemize}
\end{definition}

Observe that the witness size is implicitly constrained by the size of the quantum circuit family, cf.~\cref{def:circuits}, e.g.~for \BQP verifiers the witness is poly-sized while for \BQEXP verifiers it can be exp-sized. As before, we define
\[
  \QMA:=\bigcup_{k\in\sset N}\QM[n^k] \mt{and} \QMAEXP:=\bigcup_{k\in\sset N}\QM[\exp(n^k)].
\]
In particular, note that $\QMA \subseteq \QMAEXP$ since a \QMA verifier can be easily promoted to a \QMAEXP verifier. Indeed, while a \QMAEXP verifier gets an exponential-size witness and can run for an exponential amount of time, it does not have to (it can instead discard all witness qubits, except for a polynomial number, and verify them in polynomial time).

Soundness and completeness probabilities for \QMA and \QMAEXP can also be amplified: see theorem~10 in~\cite{Wat12}, section~3.2 of~\cite{VW16}, or lemma~14.1 in~\cite{Kitaev2002} (these techniques were originally devised for \QMA, but they can be easily adapted also for \QMAEXP).

\subsection{Geometrically \texorpdfstring{$k$-Local}{k-Local} Hamiltonians}\label{sec:klocal}

In this section we introduce basic notions relating to local Hamiltonians and formally state the \TILH problem that will play central role. For more background on Hamiltonian complexity, see~\cite{GHLS14,KSV02,Wat12}.

\begin{definition}
  An \emph{$n$-qudit Hamiltonian} is a Hermitian operator $\op H = \op H\ct$ acting on a multipartite Hilbert space $(\sset C^d)^{\otimes n}$ consisting of $n$ systems (qudits), each of \emph{local dimension}~$d$.
\end{definition}

We will label the individual systems by elements of $\set S :=\{1,\ldots,n\}$.  Whenever we talk of a \emph{subset} of systems $\set A \subseteq \set S$, we mean an \emph{ordered tuple} of distinct elements of $\set S$.  If $\op h$ is a $k$-qudit Hamiltonian for some $k \leq n$ and $\set A \subseteq \set S$ is a subset of $\abs{\set A} = k$ systems, we write $\op h_\set A$ to denote the $n$-qudit Hamiltonian that acts as $\op h$ on qudits $\set A$ and trivially (i.e.~as $\1$) on the remaining qudits $\set S \setminus \set A$. We also write $\set A + i \subseteq \set S$ to denote $\set A$ shifted by $i \in \sset N$ positions.

\begin{definition}
  Let $\op H$ be an $n$-qudit Hamiltonian. Then
  \begin{itemize}
    \item $\op H$ is \emph{$k$-local} if
      $\op H=\sum_i \op{h}(i)_{\set A_i}$ with $|\set A_i|\le k\ \forall i$;
    \item $\op H$ is $k$-local and \emph{1D} if each $\set A_i \subseteq \{1,\dotsc,k\} + t_i$ for some shift $t_i$;
    \item $\op H$ is \emph{translationally-invariant} if $\op H = \sum_i \op h_{\set A+i}$ for some $\set A \subseteq \set S$ where $\op h$ is fixed.
  \end{itemize}
  In particular, $\op H$ is a 1D translationally-invariant $k$-local Hamiltonian if $\op H = \sum_i \op h_{\{1,\dotsc,k\}+i}$ for some fixed $k$-qudit Hamiltonian $\op h$.
\end{definition}

Our central problem of interest is deciding the ground energy of 1D translationally-invariant $k$-local Hamiltonians of local dimension $d$. For brevity, we will refer to this as the \TILH{} problem.

\begin{definition}[{\TILH[(k,d)]}]\label{def:khamiltonian}\leavevmode
  Let $\op H = \sum_i \op h_{\{1,\dotsc,k\}+i}$ be a 1D translationally-invariant $k$-local Hamiltonian on a qudit chain of length $n$, where each qudit has local dimension $d$ and $\op h$ is some fixed $k$-qudit Hamiltonian.
  \begin{description}
    \item[Input.] The chain length $n$ and the matrix entries of $\op h$, as well as two real numbers $\alpha$ and $\beta$, all up to $\poly\log n$ bits of precision.
    \item[Promise.] The operator norm of each local term is bounded, $\|\op h\|\le1$, and either $\lmin(\op H)\le\alpha$ or $\lmin(\op H)\ge\beta$, where $\lmin(\op H)$ denotes the smallest eigenvalue of $\op H$ and $\beta - \alpha \geq 1/p(n)$ for some fixed polynomial $p(n)$.
    \item[Output.] \yes if $\lmin(\op H)\le\alpha$, else \no.
  \end{description}
\end{definition}

We emphasize that the input in \cref{def:khamiltonian} is just the description of the $k$-local term $\op h$ and the chain length $n$, not the entire (exponentially-sized) Hamiltonian $\op H$. An equivalent variant of the definition relaxes the norm bound to $\|\op h\|\le\poly n$ and gives a promise that either $\lmin(\op H)\le\alpha$ or $\lmin(\op H)\ge\beta$ for some fixed constants $\beta > \alpha$. We can always rescale the overall Hamiltonian by a polynomial factor to switch between the two definitions.

\begin{theorem}[Kitaev~\cite{KSV02}]\label{rem:inqmaexp}
  \TILH[(k,d)] is in \QMAEXP.
\end{theorem}

\begin{proof}
  This does not trivially follow from the inclusion $\QMA \subseteq \QMAEXP$ since the input size for \TILH is just $\poly\log n$. However, Kitaev's \QMA verifier for the standard \textsc{Local Hamiltonian} problem runs in time $\poly n$, which is not polynomial in the input size for \TILH.
  However, the exponential-time verifier of \QMAEXP offsets the logarithmically small input size, so the same random sampling argument as presented for \QMA in e.g.\ \cite[prop~14.2]{KSV02} goes through.
\end{proof}

\subsection{\QMA versus \QMAEXP}\label{sec:qmaqmaexp}
In this section, we clarify why \QMAEXP is the natural class when considering the \textsc{Local Hamiltonian} problem with translationally-invariant interactions on a system of size $n$.  When specifying a $k$-local Hamiltonian $\op H=\sum_{i\in I}\op h_i$, for some set of interactions $I$ with $|I|=\poly n$, we have to specify each term $\op h_i$ individually.  Since the locality $k$ and the local dimension $d$ are constant, the total input size in \cref{def:local-ham-intro} is thus $l=\poly n$ bits.  In contrast, specifying a \emph{translationally invariant} Hamiltonian $\op H$ requires only a logarithmic number of bits: since all local terms $\op h_i$ are identical and do not vary with the system size $n$, the only part of the input that varies with $n$ and can thus be used to encode different instances of the problem is the system size $n$ itself.

A fact which we will discuss in great detail in \cref{sec:qts} is that the gap of a Hamiltonian encoding computation as a superposition of basis states---a so-called \emph{history state construction}---scales inversely polynomially in the runtime, i.e.~$1/\poly(f(l))$ for an input of size $l$ and an $f(l)$-time computation. Contrasting this with the $1/\poly n$ gap required by \cref{def:local-ham-intro,def:khamiltonian} independently---inverse polynomially in the system size, \emph{not} the input size---we conclude the following core differences between \QMA and \QMAEXP in the context of the \textsc{Local Hamiltonian} problem (recall that $n$ denotes the length of the spin chain and $l$ denotes the total size of the input).

\paragraph{\QMA.}
A \BQP verifier has $\poly l$ runtime on an input of size $l$, so the gap of the Hamiltonian that encodes the verifier scales as $1/\poly l$.  This agrees with $1/\poly n$ in \cref{def:local-ham-intro} since $l$ and $n$ are $\poly$-related.  \QMA is thus the natural class for the \textsc{Local Hamiltonian} problem.

\paragraph{\QMAEXP.}
The \BQEXP verifier can run for $\exp\poly l$ steps in the input size $l$. The gap therefore scales as $1/\exp\poly l$, which agrees with $1/\poly n$ in \cref{def:khamiltonian} since $l=\poly\log n$. \QMAEXP is thus the natural class for \TILH.

\paragraph{} % just to get some extra space
One fact we have glossed over is that even though each instance of \TILH is translationally invariant, we could still vary the local interaction for each system size $n$. As an example, assume that the Hamiltonian $\op H$ is specified by a single local term,

\[
  \op H:=\sum_{i=1}^{N-1}\op h_{i,i+1}
  \mt{where}
  \op h=\begin{pmatrix}
    1 & 0 \\ 0 & \alpha(m)
  \end{pmatrix}
  \mt{with}
  \alpha(m)=\underbrace{3.1415926\ldots42}_{m\ \text{digits of}\ \pi}.
\]
Then the bit complexity of this input is $O(m)$, and the overall input size---i.e.~the possible information specifiable using the two parameters, the system size $n$ and a varying parameter $m$, is thus $O(m+\log n)$.  In order not to overspecify a \textsc{Local Hamiltonian} or \TILH problem, in each case we have to require both bit precision and size of the input parameter to be of the same order (within polynomial factors). We conclude with the following remark.

\begin{remark}\label{rem:qma-qmaexp-information}
  It is natural to allow $\poly n$ precision of the entries in the local terms of the Hamiltonian when working with \QMA, whereas for \QMAEXP local terms need to be precision-limited by $\poly\log n$.
\end{remark}

However, we want to emphasize that \emph{we will only make use of uniformly scaling local interaction terms}, as in~\cite{Gottesman2009}: this in particular allows us to use coupling constants that scale polynomially in $n$.
We also want to note that the polynomially-closing promise gap of history state constructions might not be the end of the story; at this point in time it is not known whether quantum computation can be encoded into the ground state of a local Hamiltonian for which the promise gap e.g.\ scales in a sub-linear fashion in the number of computational steps (cf.~\cite{Bausch2016a} for an extended discussion).

In~\cite{Cubitt2015}, the authors use a phase-estimation algorithm to extract $O(n)$ bits of information from a fixed Hamiltonian term. However, in their construction, the speed at which the gap closes is irrelevant, as long as it remains nonzero in the gapped phase.

With a $\poly(n)$-bounded computation and a $1/\poly(n)$ gap, however, it is not clear how to do this computation in a translationally-invariant manner. For phase estimation of $m$ bits, one requires gates of precision $O(\exp(-m))$, cf.~\cite{Nielsen2010}---the algorithm depends on being able to perform a unitary $\op U$ an exponential number of times, i.e.~$\op U,\op U^2,\op U^4,\ldots,\op U^{2^{m-1}}$. Without having direct access to all powers of this gate---which we do not, if we require bounded local dimension and locality---we need to approximate them in some way: using the \knownth{Solovay-Kitaev theorem} with the required exponential precision $O(\exp(-m))$ results in a circuit of size $\poly\log 1/(\exp(-m))=\poly m$, which limits the amount of information we can extract to $m=O(\log n)$.

It is clear that this is a problem of bootstrapping. For \TILH, we only have $O(\log n)$ information available to start the computation with, and again it is not known whether there exists a more direct way of extracting a phase without having to go through the \knownth{Solovay-Kitaev theorem}, which only gives a sufficient upper bound to approximate the phase estimation algorithm.

\subsection{Laplacian Matrix and Algebraic Connectivity of Graphs}\label{sec:graphs}

In this section we revise general notation and basic results from graph theory. For more background, consult the standard references \cite{Trudeau1993,Diestel2010} and \cite{Godsil2001} on graph theory and algebraic graph theory, respectively.

\begin{definition}\label{def:graphs}
  An \emph{undirected simple graph} $G=(\set V,\set E)$ consists of a set of vertices $\set V$ and a set of edges $\set E$, each edge being an unordered pair of distinct elements of $\set V$ (in particular, there are no self-loops and no multiple edges). If the number of vertices is $n = |\set V|$ and we label them as $\set V = \{v_1, \dotsc, v_n\}$, then the \emph{adjacency matrix} of $G$ is $\op A(G):=(a_{ij})_{1\le i,j\le n}$ where
  \[
    a_{ij}:=
    \begin{cases}
      1&\{v_i,v_j\} \in \set E, \\
      0&\text{otherwise.}
    \end{cases}
  \]
  We further define the \emph{degree matrix} $\op D(G):=\diag((\deg v_i)_{1\le i\le n})$ where $\deg v_i:=\sum_{j=1}^n a_{ij}$.
\end{definition}

We will usually omit the qualifiers ``undirected'' and ``simple'' in the rest of this paper. We proceed to introduce basic notions and facts from algebraic graph theory~\cite{Godsil2001}.

\begin{definition}[Laplacian matrix]\label{def:Laplacian}
  The \emph{Laplacian matrix} of a graph $G$ is defined as $\Delta(G):=\op D(G)-\op A(G)$.
\end{definition}

Since $\op A(G)$ and $\Delta(G)$ are linear operators on $\sset C^n$ where $n = \abs{\set V}$, it will often be convenient to label the basis vectors of this space by $\ket{v}$ where $v \in \set V$ and denote the space itself by $\sset C^{\set V}$.

\begin{definition}\label{def:lmin}
  We write $\lmin(\op M)$ to denote the smallest eigenvalue of Hermitian operator $\op M$. If $\op M \geq 0$ then $\lminpos{\op M}$ denotes the smallest \emph{non-zero} eigenvalue of $\op M$.
\end{definition}

\begin{claim}\label{lem:lapg}
  For any graph $G$, $\Delta(G)$ is real symmetric. In fact, $\Delta(G)$ is positive semi-definite with smallest eigenvalue $\lmin(\Delta(G))=0$ and corresponding eigenvector $(1,\ldots,1)$.
\end{claim}

\begin{proof}
  By construction, $\op A(G)$ and $\op D(G)$ are real symmetric and so is $\Delta(G)$. The second claim follows by observing that $\Delta(G)$ is symmetric and diagonally dominant. Alternatively, $\Delta(G)$ can be expressed as a sum of positive semi-definite matrices:
  \begin{equation}
    \Delta(G) = \sum_{\{a,b\} \in \set E} (\ket a - \ket b) (\bra a - \bra b)\comma
    \label{eq:Lap}
  \end{equation}
  where each term is a principal submatrix of the form
  \[
  \left(\begin{array}{rr}1&-1\\-1&1\end{array}\right)
  \]
  and encodes the Laplacian of a single edge.
  The last statement follows from the fact that the row sums of $\Delta(G)$ are zero.
\end{proof}

\begin{definition}[Algebraic connectivity]\label{def:acon}
  The second smallest eigenvalue of the Laplacian $\Delta(G)$ is denoted with $a(G)$ and called the \emph{algebraic connectivity} of graph $G$. The corresponding eigenvector is known as the \emph{Fiedler vector}.
\end{definition}

\begin{claim}[Fiedler~\cite{Fiedler}]\label{lem:a}
  For any graph $G$, $a(G)>0$ if and only if $G$ is connected.
\end{claim}

\begin{lemma}\label{lem:laplace-gs}
  If $G = G_1 \sqcup \dotsb \sqcup G_m$ is a disjoint union of connected components $G_i$ then $\Delta(G)$ has eigenvalue $0$ with multiplicity exactly $m$ and the next smallest eigenvalue is $\lminpos{\Delta(G)}=\min_ia(G_i)$. Furthermore, $\{\ket{\Phi_1}, \dotsc, \ket{\Phi_m}\}$ with
  \[
    \ket{\Phi_i}:=\frac{1}{\sqrt{|\set V_i|}}\sum_{v \in \set V_i}\ket{v}
  \]
  is an orthonormal basis of the $0$-eigenspace (ground space) of $\Delta$.
\end{lemma}

\begin{proof}
  Note that $\Delta = \Delta_1 \oplus \dotsb \oplus \Delta_m$ where $\Delta_i$ is the Laplacian of $G_i$. Recall from \cref{lem:lapg} that $\Delta_i \geq 0$ and $\lmin(\Delta_i) = 0$, hence the $m$ smallest eigenvalues of $\Delta$ are equal to $0$. Since each $G_i$ is connected, $a(G_i) > 0$ for every $i$ by \cref{lem:a}. Hence the multiplicity of eigenvalue $0$ must be $m$ and the $(m+1)$-st smallest eigenvalue of $\Delta$ is positive and equal to $a(G_i)$ for some $i$. Finally, recall from \cref{lem:lapg} that the uniform superposition over all vertices $\set V_i$ of $G_i$ is a $0$-eigenvector of $\Delta_i$, thus $\Delta \ket{\Phi_i} = 0$ for each $i \in \{1, \dotsc, m\}$. There are no further vectors in the ground space of $\Delta$ since $\Delta$ has eigenvalue $0$ with multiplicity $m$.
\end{proof}

\begin{corollary}\label{cor:ground}
  If $\Delta$ is the Laplacian of graph $G = (\set V,\set E)$ and $\set U\subseteq\set V$ is some connected component of $G$, then $\ket{\Phi_{\set U}} := \sum_{v \in \set U} \ket{v} / \sqrt{|\set U|}$ is a $0$-eigenvector of $\Delta$.
  In fact, any 0-eigenvector of $\Delta$ is a linear combination of such vectors.
\end{corollary}

\begin{claim}[Fiedler~\cite{Fiedler}]\label{ex:linegraph}
  Let $G_L$ be the path graph on $L$ vertices:
  \[
    G_L:=\
    \begin{tikzpicture}[
        every node/.style={draw,circle,minimum size=3pt,inner sep=0pt},
        font=\small,
        every label/.style={draw=none,rectangle},
        text depth=3pt,
        baseline=-1mm
      ]
      \draw (0,0) node (1) [label=$1$] {};
      \draw (1,0) node (2) [label=$2$] {};
      \draw (2,0) node (3) [label=$3$] {};
      \draw (3.5,0) node (4) [label=$L-1$] {};
      \draw (4.5,0) node (5) [label=$L$] {};
      \draw (1)--(2)--(3);
      \draw (4)--(5);
      \draw[dashed] (3)--(4);
    \end{tikzpicture}
    \point
  \]
  Then $a(G_L) = 2 \of{1 - \cos(\pi/L)} \sim \pi^2/L^2$. In particular, $a(G_L) = \Theta(1/L^2)$.
\end{claim}

\begin{corollary}\label{cor:a(G)}
  Let $G$ be a connected graph with $L$ vertices. Then $a(G)=\Omega(1/L^2)$.
\end{corollary}

\begin{proof}
  The algebraic connectivity is non-decreasing under adding edges~\cite[corollary~3.2]{Fiedler}, so for any connected graph on $L$ vertices it is lower-bounded by that of a path graph on $L$ vertices, which is given by \cref{ex:linegraph}.
\end{proof}

\subsection{Kitaev's Geometrical Lemma for Graphs}\label{sec:kitaev-graph}

We will need \knownth{Kitaev's geometrical lemma} (see Lemma~14.4 in~\cite{KSV02}) whose proof is reproduced below.

\begin{lemma}[\cite{KSV02}, p.~147]\label{lem:kitaev}
  Using notation from \cref{def:lmin}, assume $\op A,\op B\ge0$ are such that $\lminpos{\op A} \ge \mu$ and $\lminpos{\op B} \ge \mu$, and the null spaces of $\op A$ and $\op B$ have no vector in common other than $0$, i.e.~$\ker\op A\cap\ker\op B=\{0\}$.
  Then $\lmin(\op A+\op B)\ge 2\mu\sin^2\tfrac{\theta}{2}$, where $\theta$ is the angle between subspaces $\ker\op A$ and $\ker\op B$, i.e.
  \[
    \cos \theta
    := \max_{\substack{\ket{\alpha} \in \ker \op A \\ \ket{\beta} \in \ker \op B}}
    \abs{\braket{\alpha}{\beta}}
  \]
  where $\ket{\alpha}$ and $\ket{\beta}$ are unit vectors.
\end{lemma}

\begin{proof}
  We define $\Pi_\op A$ to be the projector onto $\ker\op A$, and analogously for $\Pi_\op B$. It follows from $\lminpos{\op A} \geq \mu$ that $\op A\ge\mu(\1-\Pi_\op A)$ and similarly for $\op B$. It is hence enough to show that $(\1-\Pi_\op A)+(\1-\Pi_\op B)\ge(2\sin^2\tfrac{\theta}{2})\1$, which is equivalent to $(1+\cos\theta) \1 \geq \Pi_\op A+\Pi_\op B$. In other words, we want to show that every eigenvalue $\lambda$ of $\Pi_\op A+\Pi_\op B$ satisfies
  \begin{equation}\label{eq:kitaev}
    1 + \cos \theta \geq \lambda\point
  \end{equation}

  Let $\ket\psi$ be a normalized eigenvector of $\Pi_\op A+\Pi_\op B$ with eigenvalue $\lambda \geq 0$. Since \cref{eq:kitaev} holds trivially for $\lambda = 0$, we can assume $\lambda > 0$. Since $\Pi_\op A$ projects onto $\ker \op A$, we can find a unit vector $\ket{\psi_\op A} \in \ker \op A$ such that $\Pi_\op A \ket{\psi} = a \ket{\psi_\op A}$ for some $a \in \sset C$; we can adjust the global phase of $\ket{\psi_\op A}$ to guarantee that $a \geq 0$. Similarly, $\Pi_\op B \ket{\psi} = b \ket{\psi_\op B}$ for some unit vector $\ket{\psi_\op B} \in \ker \op B$ and $b \geq 0$. Since $\Pi_\op A$ and $\Pi_\op B$ are projectors, $\bra{\psi} \Pi_\op A \ket{\psi} = \bra{\psi} \Pi_{\op A}^\dagger \Pi_{\op A}^{} \ket{\psi} = a^2 \braket{\psi_\op A}{\psi_\op A} = a^2$ and $\bra{\psi} \Pi_\op B \ket{\psi} = b^2$. From $\lambda \ket{\psi} = (\Pi_\op A + \Pi_\op B) \ket{\psi}$ we get by linearity that
  \[
    \lambda
    = \bra{\psi} (\Pi_\op A+\Pi_\op B) \ket{\psi}
    = a^2+b^2\point
  \]
  Furthermore,
  \begin{align*}
    \lambda^2&=\bra{\psi}(\Pi_\op A+\Pi_\op B)^2\ket{\psi}=a^2+b^2+2ab\Re \braket{\psi_\op A}{\psi_\op B}\\
             &\le\lambda+2ab\abs{\braket{\psi_\op A}{\psi_\op B}}
    \le\lambda+(a^2+b^2)\abs{\braket{\psi_\op A}{\psi_\op B}}
    =\lambda(1+\abs{\braket{\psi_\op A}{\psi_\op B}})\\
    &\le\lambda\of[\Big]{1+
    \max_{\substack{\ket{\alpha} \in \ker \op A \\ \ket{\beta} \in \ker \op B}}
    \abs{\braket{\alpha}{\beta}}
  }
  =\lambda(1+\cos\theta)\comma
\end{align*}
and hence $\lambda\le1+\cos\theta$, which proves \cref{eq:kitaev}.
\end{proof}

We want to use Kitaev's geometrical lemma to lower bound the smallest eigenvalue of a graph Laplacian when certain vertices are penalized. To be more specific, for a graph $G=(\set V,\set E)$ and a set of vertices $\set P\subsetneq\set V$, we write a \emph{penalizing matrix}
\[
  \op P(G,\set P) := \sum_{v \in \set P} \proj{v}.
\]
A priori, it is not clear at all what the spectrum of the \emph{penalized Laplacian} $\Delta+\op P$ is, however we can obtain a lower bound on the smallest eigenvalue.

\begin{lemma}[Kitaev's geometrical lemma for graphs]\label{lem:kitaev-graphs}
  Let $G=(\set V,\set E)$ be a connected graph. Pick a non-empty subset of \emph{penalized} vertices $\set P\subsetneq\set V$ and write the penalized Laplacian as $\Delta_\set P(G):=\Delta(G)+\op P(G,\set P)$. Then $\lmin(\Delta_\set P)=\Omega(1/|\set V|^3)$.
\end{lemma}

\begin{proof}
  Let us first verify that $\Delta$ and $\op P$ satisfy the prerequisites of \cref{lem:kitaev}. Since $G$ is connected, $\lminpos{\Delta} = a(G) > 0$ by \cref{lem:a}. Moreover, $\ker \Delta$ is spanned by the all-ones vector $\ket{\Phi_{\set V}} := \sum_{v \in \set V} \ket{v} / \sqrt{\abs{\set V}}$ according to \cref{cor:ground}. Clearly, $\lminpos{\op P} = 1$ and $\ket{\Phi_{\set V}} \notin \ker \op P$ since $\set P \neq \varnothing$, so $\ker \Delta \cap \ker \op P = \{0\}$. We can take the constant in \cref{lem:kitaev} to be $\mu := \min\{a(G),1\} = \Omega(1/|\set V|^2)$, where we used the lower bound $a(G) = \Omega(1/|\set V|^2)$ from \cref{cor:a(G)} on the algebraic connectivity of $G$. It remains to compute the angle $\theta$ between $\ker \Delta = \spn \{ \ket{\Phi_{\set V}} \}$ and $\ker \op P = \spn \{ \ket{v} : v \notin \set P \}$. We have:
  \[
    \cos \theta
    = \bra{\Phi_{\set V}} \left( \frac{1}{\sqrt{\abs{\set V} - \abs{\set P}}} \sum_{v \notin \set P} \ket{v} \right)
    = \frac{\abs{\set V} - \abs{\set P}}{\sqrt{\abs{\set V} (\abs{\set V} - \abs{\set P})}}
    = \sqrt{1 - \frac{\abs{\set P}}{\abs{\set V}}}
  \]
  and hence
  \[
    2 \sin^2 \tfrac{\theta}{2}
    = 1 - \cos \theta
    = 1 - \sqrt{1 - \frac{\abs{\set P}}{\abs{\set V}}}
    \geq \frac{1}{2} \frac{\abs{\set P}}{\abs{\set V}}
    \geq \frac{1}{2\abs{\set V}}.
  \]
  We conclude by \cref{lem:kitaev} that $\lmin(\Delta_\set P)\ge 2\mu\sin^2\tfrac{\theta}{2}=\Omega(1/|\set V|^3)$.
\end{proof}

\section{Quantum Ring Machine}\label{sec:qrm}
\subsection{Definition}
\begin{figure}[t]
  \centering
  \includegraphics[width=6cm]{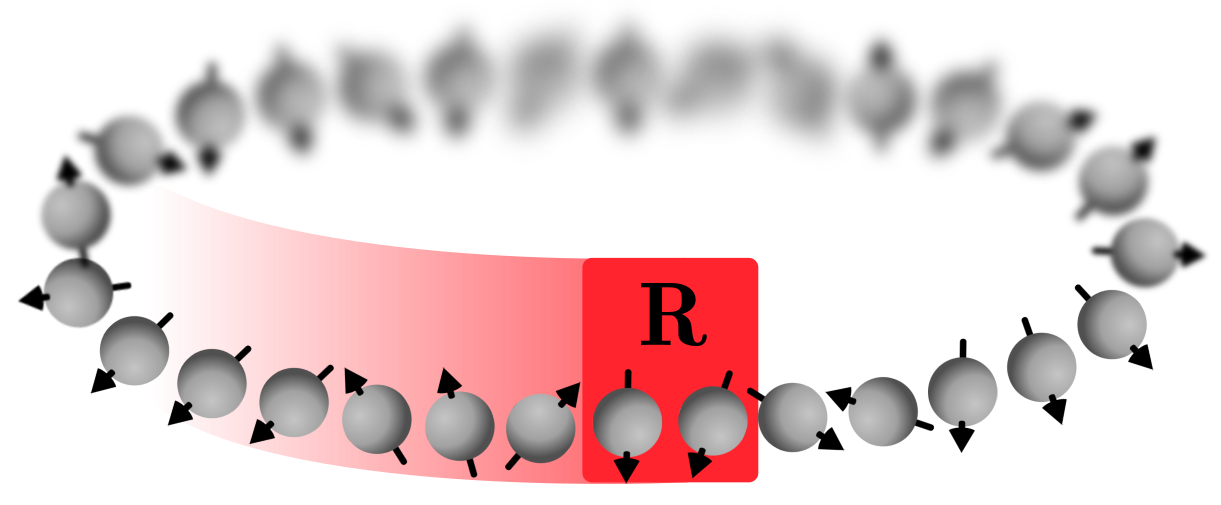}
  \caption{Quantum ring machine (QRM). Starting from a ring of qudits $\hs=(\sset C^d)^{\otimes n}$ in an initial configuration $\qi \in \hs$, a unitary $\op R \in \U(\sset C^d \otimes \sset C^d)$ is applied to pairs of adjacent qudits until one of them is completely in some halting subspace $\hsf\subseteq\sset C^d$.}
  \label{fig:qrm}
\end{figure}

We define a new computational model, a \emph{quantum ring machine} (QRM), and show that it is poly-time equivalent to a uniform class of quantum circuits. Recall that any uniform class of circuits---such as poly-time circuits or exponential-time circuits---inherits its uniformity condition from the corresponding class of classical Turing machines producing these circuit families. To prove that QRMs are quantum-universal, we will encode the given Turing machine into a specific instance of a QRM whose inner workings correspond to those of the original Turing machine, but with an additional quantum tape. In other words, the local Hilbert space of our QRM will be partitioned into two parts: a classical part, storing individual cells of the TM's tape and the internal state of the TM, and a quantum part, storing one qubit per cell. However, since a general QRM does not need to have this specific internal structure, we first give an abstract definition.

\begin{definition}\label{def:qrm-abstract}
  A \emph{quantum ring machine} (QRM) is a tuple $(\op R, n, \qi, \hsf)$, where
  \begin{itemize}
    \item $\op R \in \U(\sset C^d \otimes \sset C^d)$ is a unitary operator on a pair of qudits, each of dimension~$d$,
    \item $n\in\sset N$ is the total number of qudits on the ring,
    \item $\qi \in \hs$ is the \emph{initial state} where $\hs := (\sset C^d)^{\otimes n}$ denotes the joint Hilbert space,
    \item $\hsf \subseteq \sset C^d$ is the \emph{halting subspace} of each qudit.
  \end{itemize}
  Starting from a ring of $n$ qudits initialized in $\qi$, the operation $\op R$ is applied cyclically to adjacent pairs of qudits---see \cref{fig:qrm}---until some qudit indicates halting: its reduced density matrix has support completely within  the halting subspace $\hsf$; up until that point, the probability of finding any qudit within $\hsf$ is zero%
  \footnote{In particular this means that if, after every application of $\op R$, the corresponding two qudits are measured, each with respect to $\hsf$ and its orthogonal complement $\hsf^\perp$, then the probability of finding the reduced state in $\hsf$ should always be either zero or one, with the latter case indicating halting.}.
\end{definition}

\Cref{fig:qrm-circuit} visualizes a QRM as a quantum circuit.
Because the ring is cyclic, we can arbitrarily mark a starting position on the ring. Starting at this position, part of the initial state $\qi$ contains the input while the rest will be used as a workspace. The \emph{input size} is thus upper bounded by the ring size.

In the following definition, we consider a slight extension of QRMs from \cref{def:qrm-abstract} where $\qi$ is replaced by a family of input states $\hsi[I]$ for some index set $I$.

\begin{figure}[t]
  \centering
  \def\W{8}
  \def\w{1}
  \def\h{0.8}
  \def\gx{0.4}
  \def\gy{0.65}
  \begin{tikzpicture}[
      box/.style = {fill=white,draw = black,rounded corners=.5,line width=.8pt},
      snake/.style = {},
      snake2/.style = {decorate, decoration = {zigzag, amplitude = 4pt, segment length = 25pt}}
    ]

    \clip [snake] (-0.2,0.5*\h) rectangle (\W+0.2,-3.5*\h);
    \clip [snake2] (.2,2*\h) rectangle (\W-.2,-5*\h);

    \newcommand{\gate}[2]{
      \path[box] (#1*\w-\gx,-#2*\h-0.5*\h-\gy) rectangle (#1*\w+\gx,-#2*\h-0.5*\h+\gy);
      \node at (#1*\w,-#2*\h-0.5*\h) {$\op R$};
    }

    \foreach \i in {0,...,4} {
      \draw[line width=.8pt,double] (0,-\i*\h-.1) -- (\W,-\i*\h-.1);
      \draw[line width=.8pt] (0,-\i*\h+.1) -- (\W,-\i*\h+.1);
    }

    \gate{0}{3}
    \gate{0}{-1}
    \gate{1}{0}
    \gate{2}{1}
    \gate{3}{2}
    \gate{4}{3}
    \gate{4}{-1}
    \gate{5}{0}
    \gate{6}{1}
    \gate{7}{2}
    \gate{8}{3}
    \gate{8}{-1}
  \end{tikzpicture}
  \caption{Circuit diagram of a QRM with a ring of size $4$. The double lines indicate classical wires that are used to store the TM's internal states and tape, as well as a flag indicating either the TM's halting or the direction of its next head movement (see the proof of \cref{prop:bqpqrm} for more details). The internal details of the QRM's unitary operation $\op R$ are shown in \cref{fig:qrm-head}.}
  \label{fig:qrm-circuit}
\end{figure}

\begin{definition}\label{def:qrmterm}
  A QRM \emph{terminates} on $\hsi[I]$ if it halts, in finitely many steps, on any initial state $\qi[x]$ for $x\in I$.
  Let $(M_n)_n$ be a family of QRMs where $M_n$ has a ring of size $n$. This family is \emph{poly-time terminating} if there exists a polynomial $p$ such that $M_n$ terminates in $p(n)$ steps on all states $\qi[x]$; similarly, it is \emph{exponential-time terminating} if there exists an exponential function $f(n) = O(\exp(c n))$, for some $c>0$, such that $M_n$ terminates in $f(n)$ steps on all states $\qi[x]$.
\end{definition}

\subsection{Universality}

\begin{lemma}\label{prop:bqpqrm}
  Let $\PromP = (\PromP_\yes, \PromP_\no)$ be a promise problem in \BQP. Then there exists a polynomial $p$ and a poly-time terminating family of QRMs $(M_n)_{n}$,
  $$M_n=(\op R, n, \hsi[I_n], \hsf),$$
  with the following properties:
  \begin{enumerate}
  \item All $M_n$ share the same unitary $\op R$ and the same terminating subspace $\hsf$. The ring size of $M_n$ is $n$.
  \item The input states $\hsi[I_n]$ of each $M_n$ consist of \emph{trivial}\footnote{One must be able to produce $\qi[x]$ from $x\in\PromP$ with a constant-depth quantum circuit (in particular, one cannot cheat by allowing the input to contain the answer to the problem), e.g.\ see \cref{eq:psi_in}. This is similar to the types of input encodings one would allow for a poly-time classical TM.} encodings of instances $I_n := \{x \in \PromP : p(|x|) = n\}$, so that the whole computation fits on a ring of size $n$.
  \item If $x\in\PromP_\yes$, the reduced density matrix of the cell that signals halting satisfies an extra constraint: if measured, it collapses to an accepting subspace $\hs_\mathrm{acc} \subseteq \hs_\mathrm{halt}$ with probability $\ge 2/3$.
  \item If $x\in\PromP_\no$, it collapses to $\hs_\mathrm{acc} \subseteq \hs_\mathrm{halt}$ with probability $\le 1/3$.
  \end{enumerate}
\end{lemma}

\newcommand{\LL}{\leftarrow}
\newcommand{\RR}{\rightarrow}
\newcommand{\BB}{-}
\newcommand{\HH}{h}

\begin{proof}
  Our goal is to construct a QRM for simulating a classical Turing machine (TM) that produces a description of a uniform quantum circuit. In addition to computing the circuit's description, the QRM also executes it one gate at a time. More formally, the QRM simulates a deterministic and reversible TM (see \cref{def:DetRevTM}) augmented with the following quantum features:
  \begin{itemize}
    \item in addition to the classical data, each cell of the TM's tape stores one qubit,
    \item a special subset of the TM's states is associated with a universal set of two-qubit quantum gates; whenever the TM enters one of these states, the corresponding gate is applied on the two adjacent qubits that are stored in the pair of cells between which the TM's head just moved.
  \end{itemize}
  It is straightforward to verify that such quantum-enhanced TM is equivalent to a uniform family of quantum circuits.

  Let us now describe the simulation procedure more formally. We write the complex linear span of a finite set $\set S$ as $\sset C^\set S:=\map{span} \{\ket s\in\sset C^{|\set S|}:s\in\set S\}$ and refer to $(\sset C^\set S)^{\otimes n}$ as a \emph{ring} of size $n \in \sset N$, where each copy of $\sset C^\set S$ represents one \emph{cell} of the ring. Each cell further consists of three registers: a \emph{quantum bit} (labelled by $\{0,1\}$), a \emph{classical data} register (labelled by elements of some finite set $\Gamma$), and a \emph{flag} register (labelled by another set $\set F$). The standard basis of each ring cell is thus labelled by triples of the form
  $$\set S := \{0,1\} \times \Gamma \times \set F.$$

  Using the notation from \cref{def:DetRevTM}, let $(\set Q,\Sigma,\delta)$ be the deterministic TM we want to simulate (it is reversible without loss of generality, see \cref{rem:reversible-tm}). The first register of the QRM stores the quantum state obtained by executing the quantum circuit produced by the TM. The second register $\Gamma$ stores the TM's internal state and tape, so $\Gamma := \set Q \times \Sigma$ where $\set Q$ is the set of internal states and $\Sigma$ is the TM's alphabet. The flag symbols $\set F$ in the third register are used to mark the location of the TM's head.

  The flag register's alphabet is given by
  $$\set F := \{\LL,\RR,\BB,\HH\}$$
  and is used as follows. At any time, exactly one cell on the ring contains an \emph{active} flag (either ``$\LL$'', ``$\RR$'', or ``$\HH$'') while the rest are padded with ``$\BB$''. The TM's internal state is always stored in the active cell. Unless the TM has halted (indicated by flag ``$\HH$''), the active flag shows in which direction (``$\LL$'' for \emph{left} and ``$\RR$'' for \emph{right}) the TM's head must be moved before the simulation of the next step can begin. Every time the TM's head moves or its internal state changes, the QRM updates the flag registers and the description of the TM's internal state accordingly. Whenever the TM enters one of the special ``quantum'' states, the QRM applies the corresponding two-qubit unitary.

  Recall from \cref{def:qrm-abstract} that QRM operates by cyclically applying a fixed unitary $\op R$ on pairs of consecutive cells along the ring (see \cref{fig:qrm-circuit}). Most of the time $\op R$ acts trivially, since a non-trivial action is triggered only when either of the two active flags ``$\LL$'' or ``$\RR$'' is encountered. Note that $\op R$ acts on \emph{two} adjacent cells, one of them marked by the active flag and the other indicated by the direction of the flag's arrow:
  \begin{align*}
    \dotsb \BB \underbrace{\BB\!\LL}_{\op R} \BB \dotsb &&
    \dotsb \BB \underbrace{\RR\!\BB}_{\op R} \BB \dotsb
  \end{align*}
  It is crucial that $\op R$ is two-local for the following two reasons. First, updating the active location requires changing two symbols (e.g.~when the TM's head moves left, we need to replace ``$\BB\!\LL$'' by ``$\LL\!\BB$'' or ``$\RR\!\BB$'', depending on the direction the head will move next). For applying a two-qubit gate, we clearly also need a two-local interaction (we use the same convention as above to determine on which two qubits the gate is applied).

  \newcommand\vartextvisiblespace[1][.5em]{%
    \makebox[#1]{%
      \kern.07em
      \vrule height.3ex
      \hrulefill
      \vrule height.3ex
      \kern.07em
    }
  }
  \newcommand{\dummy}{\perp} % dummy state
  \newcommand{\blank}{\vartextvisiblespace[.7em]} % blank symbol

  Recall from \cref{thm:permutation} that, instead of quintuples $\delta$, we can work with a permutation matrix $\op T_\delta$ on $\Gamma$ and a function $d: \set Q \to \{\mathrm{left}, \mathrm{right}\}$ telling us where the TM's head came from. For convenience, we include a special \emph{dummy} state ``$\dummy$'' in $\set Q$ and a designated \emph{blank} symbol ``$\blank$'' in $\set \Sigma$: the dummy state is stored in all cells (except the active cell which stores the actual state of the TM) while the blank symbol is used to initialize the TM's tape. We accordingly extend $\op T_\delta$ so that it acts trivially on $\ket{\dummy,\sigma}$ for any $\sigma \in \Sigma$, and we define $d(\dummy) := \BB$ so that dummy states do not trigger any action in our simulation.

  We take the ring size to be $n=p(|x|)$ for an instance $x\in\PromP$, since the TM can access at most that many tape cells. We require that the ring starts out in a well-formed state, i.e.~for some binary representation $x=x_1x_2\cdots x_l$ and $l=|x|$, a state of the form
  \begin{equation}
    \qi[x] :=
    \bigotimes_{j=1}^{n-l-1}\of{ \ket 0 \otimes \ket{\dummy, \blank} \otimes \ket{\BB} } \otimes
    \bigotimes_{i=1}^l
    \of{ \ket{x_i} \otimes \ket{\dummy, \blank} \otimes \ket{\BB} } \otimes
    \of{ \ket 0 \otimes \ket{q_0, \blank} \otimes \ket{\RR} },
    \label{eq:psi_in}
  \end{equation}
  i.e.\ where all cells but the last are initialized as follows: the TM is in the dummy state ``$\dummy$'', the TM's tape is initialized to a designated blank symbol ``$\blank$'', and the flag is set to ``$\BB$''. The last cell contains the TM's initial state $q_0$ and the ``$\RR$'' flag. The input $x \in \PromP$ is written on the qubit part of the tape, i.e.~the first register of each cell.

  We can now describe in more detail the steps involved in our simulation, and how to perform them reversibly (see \cref{fig:qrm-head,fig:qrm-perm} for more details):
  \begin{enumerate}
    \item If the active cell has the halting flag ``$\HH$'', the TM has halted so nothing happens.
    \item If the active cell has one of the other two flags ``$\LL$'' or ``$\RR$'':
      \begin{enumerate}
        \item The $\set Q$ part of the $\Gamma$ registers of the active cell and its neighbour---indicated by the flag---are exchanged, thus simulating the movement of the TM's head.
        \item The flag register of the active cell is uncomputed using the function $d$.
        \item Description of the TM's internal state and the current tape symbol is updated using $\op T_\delta$.
        \item Based on the updated internal state, a new flag register is computed using $d$ (it belongs to the same cell where the TM's new state is stored, and it indicates in which direction the TM's head will move before the next iteration begins).
        \item If the TM is in one of the special states indicating a quantum gate, the corresponding unitary is applied on the two data registers.
      \end{enumerate}
  \end{enumerate}

  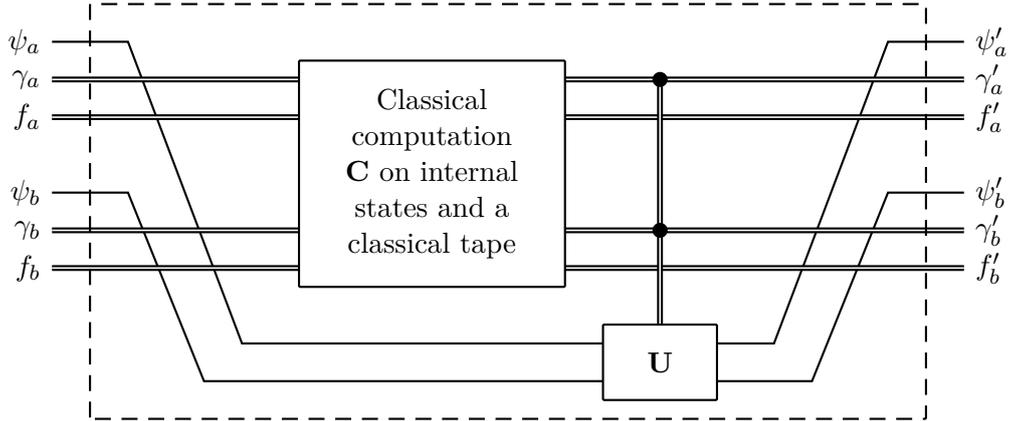
\begin{figure}
    \centering
    \begin{tikzpicture}[
        draw=black,line width=.8pt,
        box/.style = {rounded corners=.5},
        label/.style = {align=right,text width=.7cm},
        labelr/.style = {align=left,text width=.7cm}
      ]

      \newcommand{\gate}[5]{
        \draw[fill=white,box] (#1-.25,#2-.25) rectangle (#3+.25,#4+.25);
        \node[align=center] at (0.5*#3+0.5*#1,0.5*#4+0.5*#2) {#5};
      }

      % Quantum wires
      \draw (0,5) -- (1,5) -- (2.5,1.0) -- ( 9.5,1.0) -- (11,5) -- (12,5);
      \draw (0,3) -- (1,3) -- (2.0,0.5) -- (10.0,0.5) -- (11,3) -- (12,3);

      % Classical wires
      \draw[double] (0,4.5) -- (12,4.5);
      \draw[double] (0,4.0) -- (12,4.0);
      \draw[double] (0,2.5) -- (12,2.5);
      \draw[double] (0,2.0) -- (12,2.0);

      % Classical gate
      \gate{3.5}{2.0}{6.5}{4.5}{Classical\\computation\\$\op C$ on internal\\states and a\\classical tape};

      % Controlled-U
      \draw[double] (8,1) -- (8,4.5);
      \fill (8,2.5) circle [radius=0.1];
      \fill (8,4.5) circle [radius=0.1];
      \gate{7.5}{0.5}{8.5}{1.0}{$\op U$}

      % Bounding box
      \draw[dash pattern=on 6pt off 4pt,box] (0.5,0.0) rectangle (11.5,5.5);

      \node[label] at (-.5,5.0) {$\psi_a$};
      \node[label] at (-.5,4.5) {$\gamma_a$};
      \node[label] at (-.5,4.0) {$f_a$};

      \node[labelr] at (12.5,5.0) {$\psi_a'$};
      \node[labelr] at (12.5,4.5) {$\gamma_a'$};
      \node[labelr] at (12.5,4.0) {$f_a'$};

      \node[label] at (-.5,3.0) {$\psi_b$};
      \node[label] at (-.5,2.5) {$\gamma_b$};
      \node[label] at (-.5,2.0) {$f_b$};

      \node[labelr] at (12.5,3.0) {$\psi_b'$};
      \node[labelr] at (12.5,2.5) {$\gamma_b'$};
      \node[labelr] at (12.5,2.0) {$f_b'$};
    \end{tikzpicture}
    \caption{Circuit diagram for implementing the QRM's head unitary $\op R$ (double wires are classical while single wires are quantum). All computation is classical except for a single classically-controlled quantum gate $\op U$ that can be triggered by either of the two $\Gamma$ registers. A classical circuit for implementing $\op C$ is shown in \cref{fig:qrm-perm}.}
    \label{fig:qrm-head}
  \end{figure}

  We now describe the unitary operator $\op R$ that acts on two adjacent QRM's cells:
  \begin{itemize}
    \item For $a = (\psi_a, \gamma_a, f_a) \in \set S$, write the corresponding basis state as $\ket a:=\ket{\psi_a, \gamma_a, f_a}\in\sset C^{\set S}$ where $\psi_a \in \{0,1\}$, $\gamma_a \in \Gamma$, $f_a \in \set F$, and analogously for $\ket b$.
      Then $\ket a\otimes\ket b\in\sset C^{\set S\times\set S}$ is also a basis state and we require, up to reordering the registers (see \cref{fig:qrm-head}), that
      \[
        \op R(\ket a\otimes\ket b)=\ket{\Psi}\otimes\ket{\gamma_a',f_a'}\otimes\ket{\gamma_b',f_b'}
      \]
      for some $\gamma_a',\gamma_b'\in\Gamma$ and $f_a',f_b' \in \set F$, i.e.~$\op R$ acts classically on each register except for the quantum data registers $\ket{\psi_a}$ and $\ket{\psi_b}$ (in particular, we allow $\ket{\Psi} \in \sset C^2 \otimes \sset C^2$ to be entangled).
    \item Using the same notation, if $f_a \neq {\RR}$ and $f_b \neq {\LL}$, we further demand $\ket{\gamma_a',f_a'}=\ket{\gamma_a,f_a}$, $\ket{\gamma_b',f_b'}=\ket{\gamma_b,f_b}$, and $\ket{\Psi}=\ket{\psi_a}\otimes\ket{\psi_b}$, i.e.~if neither $f_a$ nor $f_b$ signal ``apply head here'', $\op R$ acts as the identity operator on all registers.
    \item The active flag always moves in the direction indicated by the arrow. If $f_a = {\RR}$ then $f_a' = \BB$ and $f_b' \neq \BB$, meaning that the head has moved right. Similarly, if $f_b = {\LL}$ then $f_b' = \BB$ and $f_a' \neq \BB$, meaning that the head has moved left. In each case there are three possible transitions---they indicate whether the TM has halted or in which direction its head has to move next:
      $$\begin{array}{ccc}
        f_a f_b && f_a f_b \\
        \RR \BB && \BB \LL \\[5pt]
        \Downarrow && \Downarrow \\[7pt]
        f_a' f_b' && f_a' f_b' \\
        \BB \RR && \RR \BB \\
        \BB \LL && \LL \BB \\
        \BB \;\,\HH\;\, && \:\HH\;\, \BB
      \end{array}$$
  \end{itemize}
  \Cref{fig:qrm-head} shows how $\op R$ acts on two adjacent cells. For each cell, the halting subspace $\hsf$ is spanned by all standard basis vectors with the last register in the halting state $\ket{\HH}$:
  \[
    \hsf := \sset C^2 \otimes \sset C^\Gamma \otimes \ket{\HH}.
  \]
  \Cref{fig:qrm-perm} provides details on how to implement $\op C$ reversibly.

  \begin{figure}
    \centering
    \begin{tikzpicture}[
        draw=black,line width=.8pt,
        box/.style = {rounded corners=.5},
        label/.style = {align=right,text width=.9cm},
        labelr/.style = {align=left,text width=.9cm}
      ]

      \newcommand{\gate}[4]{
        \draw[fill=white,box] (#1-#4,#2-#4) rectangle (#1+#4,#2+#4);
        \node[align=center] at (#1,#2) {#3};
      }

      \def\r{0.10}
      \def\R{0.22}

      \newcommand{\CNOT}[4]{
        \draw (#1,#2) -- (#1,#3-#4*\R);
        \fill (#1,#2) circle (\r);
        \draw (#1,#3) circle (\R);
      }

      % Control-SWAP line
      \draw (3,0) -- (3,3.0);

      % Classical wires
      \draw[double] (0,3.2)                       -- (2.5,3.2) -- (3.5,2.8) -- (4,2.8) -- (5,1.2) -- (9,1.2);
      \draw[double] (0,2.8) -- (1,2.8) -- (2,2.0)                           -- (4,2.0) -- (5,2.8) -- (9,2.8);
      \draw[double] (0,2.0) -- (1,2.0) -- (2,1.2)                           -- (4,1.2) -- (5,2.0) -- (9,2.0);
      \draw[double] (0,1.2) -- (1,1.2) -- (2,2.8) -- (2.5,2.8) -- (3.5,3.2)                       -- (9,3.2);
      \draw[double] (0,0.8)                                                                       -- (9,0.8);
      \draw[double] (0,0.0)                                                                       -- (9,0.0);

      % Permutation gates
      \gate{7}{1}{$\op T_\delta$}{0.42}
      \gate{7}{3}{$\op T_\delta$}{0.42}

      % CNOTs
      \CNOT{5.5}{1.2}{2.0}{-1}
      \CNOT{6.0}{3.2}{0.0}{1}
      \CNOT{8.0}{1.2}{0.0}{1}
      \CNOT{8.0}{3.2}{2.0}{1}

      % Control-SWAP controls
      \fill (3,0.0) circle (\r);
      \fill (3,1.2) circle (\r);

      % Bounding box
      \draw[dash pattern=on 6pt off 4pt,box] (0.5,-0.5) rectangle (8.5,3.7);

      % Wire labels
      \node[label]  at (-1.4,2.95) {$\gamma_a=$};
      \node[label]  at (-1.0,3.0) {$\bigg\{$};
        \node[label]  at (-.6,3.20) {$q_a$};
        \node[label]  at (-.6,2.75) {$\sigma_a$};
        \node[label]  at (-.6,2.00) {$f_a$};

        \node[labelr] at (10.4,3.0) {$=\gamma_a'$};
      \node[labelr] at (10.0,3.0) {$\bigg\}$};
      \node[labelr] at (9.6,3.25) {$q_a'$};
      \node[labelr] at (9.6,2.80) {$\sigma_a'$};
      \node[labelr] at (9.6,2.00) {$f_a'$};

      \node[label]  at (-1.4,0.95) {$\gamma_b=$};
      \node[label]  at (-1.0,1.0) {$\bigg\{$};
        \node[label]  at (-.6,1.20) {$q_b$};
        \node[label]  at (-.6,0.75) {$\sigma_b$};
        \node[label]  at (-.6,0.00) {$f_b$};

        \node[labelr] at (10.4,1.0) {$=\gamma_b'$};
      \node[labelr] at (10.0,1.0) {$\bigg\}$};
      \node[labelr] at (9.6,1.25) {$q_b'$};
      \node[labelr] at (9.6,0.80) {$\sigma_b'$};
      \node[labelr] at (9.6,0.00) {$f_b'$};
    \end{tikzpicture}
    \caption{Circuit diagram for implementing the classical permutation $\op C$ in \cref{fig:qrm-head} (all wires are classical and all gates are reversible). Conditioned on the flag registers being either ``$\RR\BB$'' or ``$\BB\LL$'', the controlled-controlled-\swap gate exchanges the internal state registers of the two cells. The \cnot gates in the first layer are conditioned on the value of $d(q)$, for state $q$, and they uncompute the flag register of the opposite cell (the cell where the TM's head came from). The permutation $\op T_\delta$ acts on $\Gamma$ registers of both cells to update the TM's internal state and the current tape symbol. Recall that $\op T_\delta$ acts trivially if the state is dummy (at most one of the cells is in a non-dummy state). The final layer of \cnot gates again condition on $d(q')$, where $q'$ is the new state, and update the flag registers to indicate where the TM's head will move next. These flags will be uncomputed by the next iteration.}
    \label{fig:qrm-perm}
  \end{figure}
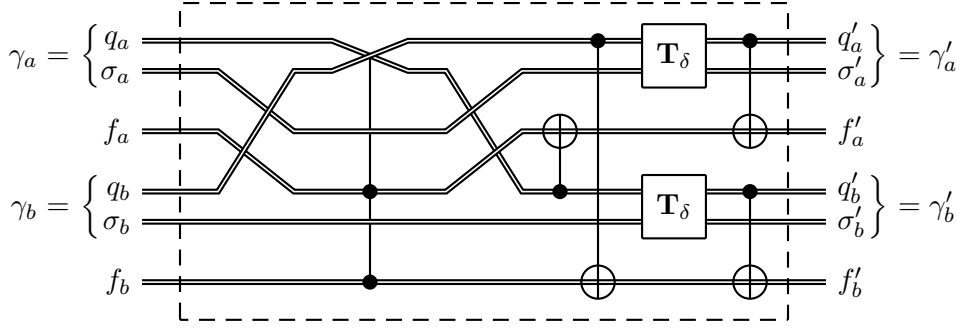

  We construct the desired family of QRMs $(M_n)_n$ in the special form described above. It is straightforward to verify that this ring machine executes the circuit written out by the TM, and the runtime overhead of $M_n$ as compared to the circuit is at most quadratic.
\end{proof}

\begin{corollary}\label{cor:bqexp-qrm}
  Using an exponential-time terminating family of QRMs, \cref{prop:bqpqrm} holds for \BQEXP as well.
\end{corollary}

\section{Unitary Labelled Graphs}
\subsection{Definitions}

The following definition introduces graphs whose vertices are labelled by Hilbert spaces and whose edges are labelled by unitaries between these spaces.
\begin{definition}\label{def:ulg}
  Given an undirected graph $G=(\set S,\set E)$ without self-loops, a \emph{unitary labelled graph} (ULG) is a triple $(G,(\hs_v)_{v\in\set S},g)$ where
  \begin{itemize}
    \item $(\hs_v)_{v\in\set S}$ is a family of Hilbert spaces, one space $\hs_v$ for each vertex $v\in\set S$,
    \item $g$ is a function that assigns to each directed\footnote{While $G$ is an undirected graph, we need to arbitrarily direct its edges so that we can discern between labels $\op U$ and $\op U^\dagger$ assigned to edges $ab$ and $ba$, respectively.} edge $ab \in \set E$ some unitary operator $g(ab) \in \U(\hs_a)$ so that $g(ab) = g(ba)\ct$ (this requires that $\hs_a \cong \hs_b$ whenever $ab \in \set E$).
  \end{itemize}
\end{definition}

To facilitate notation, we will write an edge and its associated unitary jointly as $(a\leftrightarrow b,\op U)$ and call it a \emph{rule}.
By definition, the rule $(a\leftrightarrow b,\op U)$ is equivalent to the rule $(b\leftrightarrow a,\op U\ct)$.
With this notation, it is convenient to specify a unitary labelled graph by $G = (\set S, \set R)$ where $\set R:=\{(a\leftrightarrow b,g(ab)):ab\in\set E\}$ is the corresponding set of rules.

\begin{example}\label{ex:ulg}
  Let $\set S:=\{1,2,3,4,5\}$, $\hs=\sset C^2$ for all vertices, and consider the following set of rules:
  \[
    \set R:=\{(1\leftrightarrow 2,\op U),(2\leftrightarrow 3,\op V),(3\leftrightarrow 4,\op W),(4\leftrightarrow 1,\1),(4\leftrightarrow 5,\1)\}\point
  \]
  The underlying graph for this example is shown in \cref{fig:ulg}.
\end{example}
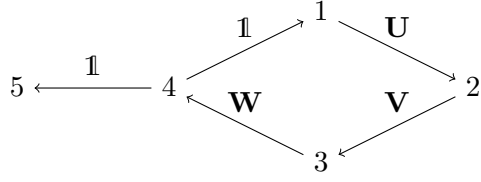
\begin{figure}
  \centering
  \begin{tikzpicture}[node distance = 14mm, text height = 1.5ex, text depth = .25ex]
    \node (5) at (0,0) {5};
    \node (4) at (2,0) {4};
    \node (1) at (4,1) {1};
    \node (3) at (4,-1) {3};
    \node (2) at (6,0) {2};
    \draw[->] (1) -- node[above] {$\op U$} (2);
    \draw[->] (2) -- node[above] {$\op V$} (3);
    \draw[->] (3) -- node[above] {$\op W$} (4);
    \draw[->] (4) -- node[above] {$\1$} (1);
    \draw[->] (4) -- node[above] {$\1$} (5);
  \end{tikzpicture}
  \caption{Unitary labelled graph from \cref{ex:ulg}. Observe that we mark the direction for the unitaries with an arrow, despite working with undirected graphs.}
  \label{fig:ulg}
\end{figure}

\begin{definition}\label{def:ulg-simple}
  Let $G = (\set S, \set R)$ be a ULG. If the product of unitaries along any directed path connecting $a$ and $b$ is equal, and this property holds for all $a,b\in\set S$, we call the ULG \emph{simple}. Equivalently, for a ULG to be simple, the product of unitaries along any directed cycle should be $\1$.
\end{definition}
The ULG in \cref{ex:ulg} is simple if and only if $\op W\op V\op U=\1$.

The following definition assigns a Hamiltonian to each ULG. This Hamiltonian extends the notion of a graph Laplacian, see \cref{def:Laplacian}, to ULGs (while this might not be immediately obvious from the definition, it will be made more clear in \cref{lem:ulgsimpleclassical} below).
\begin{definition}\label{def:ulglaplacian}
  Let $G = (\set S, \set R)$ be a connected ULG, \hs denote the Hilbert space attached to each of its vertices, $n:=\dim\hs$ be the dimension of \hs, and let $\{\ket{e_i}\}_{i=1}^n$ be some orthonormal basis of \hs. The \emph{Hamiltonian associated to $G$} is the following Hermitian operator on $\sset C^\set S\otimes\hs$:
  \begin{equation}\label{eq:H(G)}
    \op H(G):=\sum_{(a \leftrightarrow b, \op U)\in\set R}\;\sum_{i=1}^{n}(\ket{a}\otimes\ket{e_i}-\ket{b}\otimes\op U\ket{e_i})(\bra{a}\otimes\bra{e_i}-\bra{b}\otimes\bra{e_i}\op U\ct)\point
  \end{equation}
\end{definition}

This is reminiscent of \cref{eq:Lap} for $\Delta(G)$, the Laplacian of graph $G$.
Furthermore, it also explains why we excluded self-loops in \cref{def:ulg}: just as they have no effect on the graph Laplacian, they also impose no changes in the associated Hamiltonian of a simple ULG---the only possible self-loop unitary for such ULG is $\1$, making the corresponding term in \cref{eq:H(G)} vanish.

\begin{proposition}\label{prop:ulglaplacian}
  The Hamiltonian $\op H=\op H(G)$ of a connected UGL $G = (\set S, \set R)$, see \cref{def:ulglaplacian}, is invariant under replacing any rule $(a\leftrightarrow b,\op U) \in \set R$ with the corresponding inverse rule $(b\leftrightarrow a,\op U^\dagger)$. Moreover, the matrix entries of $\op H$ do not depend on the choice of the basis $\{\ket{e_i}\}_{i=1}^n$.
\end{proposition}

\begin{proof}
  Since $G$ is connected, the Hilbert spaces attached to all its vertices are isomorphic. Observe further that
  \begin{equation}\label{eq:HIU}
    \op H \equiv\sum_{(a \leftrightarrow b, \op U)\in\set R}(\proj{a}\otimes\1_n+\proj{b}\otimes\1_n-\ketbra{a}{b}\otimes\op U^\dagger-\ketbra{b}{a}\otimes\op U)\comma
  \end{equation}
  hence the claim follows.
\end{proof}

Note that an alternative way of writing $\op H$ is as follows:
  \[
    \op H \equiv\sum_{(a \leftrightarrow b, \op U)\in\set R} \sum_{i}
    (\ket a\otimes\ket{i}-\ket b\otimes\op U\ket{i})(\bra a\otimes\bra{i}-\bra b\otimes\bra{i} \op U^\dagger)\comma
  \]
which emphasizes the fact that each term is positive semi-definite.

One can extend the notion of an associated Hamiltonian to a non-connected ULG as well by taking a direct sum of the Hamiltonians for each component of the graph (equivalently, one can assume that the Hilbert spaces associated to different components of the graph are mutually orthogonal and take the new Hilbert space to be their direct sum). Either way, such extension yields a block-diagonal associated Hamiltonian.

\subsection{Semi-Classical Unitary Labelled Graphs}

A ULG is semi-classical if its Hamiltonian is equal to a graph Laplacian (see \cref{def:Laplacian}), after a unitary change of basis.
\begin{definition}
  A ULG $G$ is \emph{semi-classical} if its associated Hamiltonian can be expressed as $\op H = \op W (\Delta \otimes \1_n) \op W^\dagger$, where $\op W$ is some unitary operator, $\Delta$ is the Laplacian of $G$, and $\1_n$ acts on the $n$-dimensional Hilbert space attached to each vertex of $G$.
\end{definition}
This definition can be easily extended also to non-connected ULGs as well.

The following lemma is important for analysing the spectrum of any Hamiltonian coming from a simple ULG. It reduces the problem to analysing instead the spectrum of the corresponding graph Laplacian.
\begin{lemma}\label{lem:ulgsimpleclassical}
  Any simple ULG is semi-classical.
\end{lemma}
\begin{proof}
  Denote the UGL by $G = (\set S, \set R)$ where $\set S$ and $\set R$ are the sets of vertices and rules, respectively. If $G$ has disjoint components, $\op H$ is block-diagonal and we can deal with each block separately, hence we can assume without loss of generality that $G$ is connected and all its vertices have isomorphic attached Hilbert spaces \hs.

\begin{figure}
  \centering
  \begin{tikzpicture}
    \node (a) at (0,1) {$a$};
    \node (b) at (3,0) {$b$};
    \node (c) at (3,2) {$c$};
    \node (d) at (6,0) {$d$};
    \node (e) at (6,2) {$e$};
    \node (f) at (9,1) {$f$};
    \draw (3,1) ellipse [x radius = 0.6, y radius = 1.3];
    \node at (3.0,-0.6) {$\set S_a$};   \draw[dashed] (2.3,3.2) -- +(0,-4.0);
    \node at (1.5, 2.9) {$\op H_a$};    \draw[dashed] (3.7,3.2) -- +(0,-4.0);
    \node at (3.0, 2.9) {$\op H_\mathrm{n}$};    \draw[dashed] (5.3,3.2) -- +(0,-4.0);
    \node at (4.5, 2.9) {$\op H_\mathrm{nn}$};
    \node at (6.0, 2.9) {$\op F$};
    \draw[->] (a) -- node[above] {$\op U_b$} (b);
    \draw[->] (a) -- node[above] {$\op U_c$} (c);
    \draw[->] (b) -- node[right] {$\op U$} (c);
    \draw[->] (b) -- (d);
    \draw[->] (e) -- (d);
    \draw[->] (c) -- (e);
    \draw[->] (d) -- (f);
    \draw[->] (e) -- (f);
  \end{tikzpicture}
  \caption{Terms of $\op H$, see \cref{eq:HHHR}, grouped according to how far the corresponding edges are from the chosen vertex $a$. For this ULG to be simple, the labels of edges forming the triangle $abc$ must satisfy $\op U_c\ct \op U \op U_b = \1$. We have omitted the labels of all other edges.}
  \label{fig:triangle}
\end{figure}
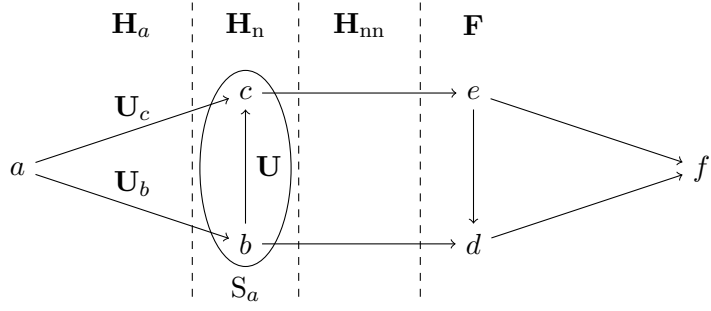

  Pick an arbitrary vertex $a \in \set S$ and denote its set of neighbours by $\set S_a$. Using \cref{prop:ulglaplacian}, rewrite $\set R$ in a form where $a$ only has outgoing edges. Following \cref{eq:HIU}, define the term that encodes rule $(a\leftrightarrow b,\op U_b) \in \set R$ as follows:
  \begin{equation}\label{eq:hab}
    \op h_{ab}:=
    \proj{a}\otimes\1+\proj{b}\otimes\1-\ketbra{a}{b}\otimes\op U_b^\dagger-\ketbra{b}{a}\otimes\op U_b\comma
  \end{equation}
  where the subscript of $\op U_b$ identifies the vertex with incoming edge. Then the terms of $\op H$ can be grouped as follows (see \cref{fig:triangle}):
  \begin{align}
    \op H&=\sum_{\substack{(a\leftrightarrow b,\op U_b)\\b\in\set S_a}}\op h_{ab}+
    \sum_{\substack{(b\leftrightarrow c,\op U_c)\\b,c\in\set S_a}}\op h_{bc}+
    \sum_{\substack{(b\leftrightarrow d,\op U_d)\\b\in\set S_a \land d\notin \set S_a \cup \{a\}}}\op h_{bd}+\op F \nonumber\\
    &=:\op H_a+\op H_\mathrm{n}+\op H_\mathrm{nn}+\op F\comma \label{eq:HHHR}
  \end{align}
  where $\op F$ denotes the rest of the terms and all sums range over $\set R$, with some restrictions on the endpoints of the edges.
  Our strategy now is to apply a sequence of unitary transformations to bring the Hamiltonian $\op H$ to the desired form, a few terms at a time.

  First, for the given vertex $a \in \set S$, define the following unitary:
  \[
    \op W_a:=\prod_{\substack{(a\leftrightarrow b,\op U_b)\\b\in\set S_a}}\bigl(\proj b\otimes\op U_b+(\1-\proj b)\otimes\1\bigr)\point
  \]
  Observe that all terms in the product commute. Moreover, $\op W_a^\dagger\op F\op W_a=\op F$ and, by \cref{eq:hab},
  \[
    \op W_a^\dagger\op h_{ab}\op W_a=(\proj a+\proj b-\ketbra ab-\ketbra ba)\otimes\1
  \]
  for all $b \in \set S_a$, so $\op W_a$ takes care of all terms of $\op H_a$ simultaneously.

  For the terms in $\op H_\mathrm{n}$, pick any edge $(b\leftrightarrow c, \op U) \in \set R$, with $b,c\in\set S_a$, and note from \cref{eq:hab} that
  \[
    \op W_a^\dagger\op h_{bc}\op W_a
    = \proj{b} \otimes \1 + \proj{c} \otimes \1
    - \ketbra{b}{c} \otimes \op U_b\ct \op U\ct \op U_c
    - \ketbra{c}{b} \otimes \op U_c\ct \op U \op U_b\point
  \]
  However, since $abc$ is a cycle and the ULG is simple, the product of unitaries along the cycle must be $\1$, i.e.~$\op U_c^\dagger\op U\op U_b=\1$ (see \cref{fig:triangle}), so the formula simplifies to
  \begin{align*}
    \op W_a^\dagger\op h_{bc}\op W_a
   &= (\proj{b} + \proj{c} - \ketbra{b}{c} - \ketbra{c}{b}) \otimes \1 \\
   &= (\ket{b} - \ket{c})(\bra{b} - \bra{c}) \otimes \1\point
  \end{align*}
By a similar argument, we can show that the rules in $\op H_\mathrm{nn}$ change their unitary by a factor of $\op U_b$ when outgoing, or $\op U_b^\dagger$ when incoming, respectively. We are left with a ULG with a new set of rules, namely, every edge that is either attached to $a$ or between two different neighbours of $a$ is trivial, i.e.\ the edge unitary is the identity operator $\1$.

  We will apply the same procedure to different vertices until all edges become trivial. More specifically, we consider an arbitrary sequence of subsets $\set S_1 \subset \set S_2 \subset \dotsb \subset \set S_m \subset \set S$, starting from any vertex $\{a_1\} = \set S_1$ and ending with the set of all vertices $\set S$, such that each subsequent $\set S_{k+1}$ can be obtained from $\set S_k$ by including all neighbours of some vertex $a_k \in \set S_k$. The overall unitary is then $\op W = \prod_{k=1}^m \op W_{a_k}$, where the product is over the sequence of vertices $a_1, a_2, \dotsc, a_m \in \set S$. Each successive unitary $\op W_{a_k}$ is obtained from the current set of rules $\set R_k$, where $\set R_1 = \set R$ is the original set while all rules in the final set are trivial. Our goal is to show that, at every step $k$, we can guarantee that each rule in $\set R_k$ has a trivial unitary whenever both endpoints of the corresponding edge are in $\set S_k$.

  We proceed by induction. Since there are no edges between vertices in $\set S_1 = \{a_1\}$, the induction basis is trivial. Assuming the inductive hypothesis holds for $k$, we take $a_k \in \set S_k$ and apply the unitary $\op W_{a_k}$ that acts non-trivially to all neighbours of $a_k$ (recall that $\set S_{k+1}$ is formed by $\set S_k$ together with the neighbours of $a_k$). As discussed above, $\op W_{a_k}$ trivializes all edges between $a_k$ and any of its neighbours. By simplicity of the ULG, it trivializes also all edges between any two different neighbours of $a_k$. Moreover, it does not affect any edges within $\set S_k$ (they are trivial already by the inductive assumption). In other words, all edges between vertices in $\set S_{k+1}$ are trivial, thus completing the induction.

Since all edge unitaries have now been transformed to the trivial unitary $\1$, all terms in $\op W^\dagger\op H\op W$ are of the form
\begin{equation}
  (\proj{a} + \proj{b} - \ketbra{a}{b} - \ketbra{b}{a}) \otimes \1
  = (\ket{a} - \ket{b})(\bra{a} - \bra{b}) \otimes \1 \comma
\end{equation}
for some $a,b \in \set S$. Comparing this to \cref{eq:Lap}, the overall Hamiltonian is in fact equivalent to the Laplacian $\Delta$ of $G$, i.e.~$\op W^\dagger\op H\op W=\Delta\otimes\1$.
\end{proof}

\begin{lemma}\label{lem:qevolution}
  Let $G$ be a simple ULG with vertices $\set S$ and rules $\set R$. Write $G=G_1\oplus\dotsb\oplus G_N$, $G_i=(\set V_i,\set E_i)$ for the associated Laplacian of the induced classical graph $(\set S,\set R')$ with $\set R':=\{a\leftrightarrow b:(a\leftrightarrow b,\op U)\in\set R\}$, and pick an arbitrary $v_i\in\set V_i\ \forall i$. Let further $n_i:=\dim\hs_i$ and choose a basis $\{\ket{e_{i,j}}\}_j$ of $\hs_i$ for all connected graph components $G_i$. Then the ground space $\ker\op H$ is spanned by the set $\{\ket{\Psi_{i,0}},\ldots,\ket{\Psi_{i,n_i}}\}_{i=1}^N$, where
  \[
    \ket{\Psi_{i,j}}:=\frac1{\sqrt{|\set V_i|}}\sum_{s\in \set V_i}\ket{s}\otimes\ket{q_s}
    \mt{and}
    \ket{q_s}=\begin{cases}
      \ket{e_{i,j}} & \text{if\ } s=v_i,\\
      \op U\ket{q_r} & \text{if\ } (r\leftrightarrow s,\op U)\in\set E_i.
    \end{cases}
  \]
  Furthermore, the $\ket{\Psi_{i,j}}$ form a basis of $\ker\op H$.
\end{lemma}
\begin{proof}
  Because the ULG is simple, by \cref{lem:ulgsimpleclassical}, there exists a unitary $\op W$ and a classical Laplacian $\Delta$ such that $\op W^\dagger\op H\op W=\Delta\otimes\1_n$, and hence $\ker\op H=\ker\Delta\otimes\1_n$. By \cref{lem:laplace-gs}, the ground space of $\Delta\otimes\1_n$ has a basis given by
  \[
    \ket{\Phi_{i,j}}:=\frac1{\sqrt{|\set V_i|}}\sum_{s\in \set V_i}\ket{s}\otimes\ket{e_i},\quad i=1,\ldots,N\mt{and}j=1,\ldots,n\point
  \]
  Observe that
  \[
    \op W\ket{\Phi_{i,j}}=\frac1{\sqrt{|\set V_i|}}\sum_{s\in\set V_i}\prod_{a\in\set V_i}\op W_a\ket s\otimes\ket{e_i}=\ket{\Psi_{i,j}}\comma
  \]
  which can be easily verified.
\end{proof}

\subsection{Kitaev's Geometrical Lemma for Unitary Labelled Graphs}
Analogous to \cref{sec:kitaev-graph}, we extend the notion of penalizing vertices to ULGs. First we state an immediate corollary from \cref{lem:kitaev-graphs}.
\begin{corollary}
  Take a connected simple ULG with Hilbert space $\hs$ for all vertices $s\in\set S$. Pick a non-empty subset of vertices $\set P\subsetneq\set S$ and write the \emph{penalized associated Hamiltonian} $\op H_\set P(G):=\op H(G)+\op P(G,\set P)\otimes\1_{\dim\hs}$. Then $\lmin(\op H_{\set P}(G))=\Omega(1/|\set S|^3)$.
\end{corollary}
\begin{proof}
  $G$ is simple and $\op H(G)$ has the same spectrum as $\Delta(G)$, up to multiplicity. Now use \cref{lem:kitaev-graphs} on $\Delta(G)+\op P(G,\set P)$.
\end{proof}

A more interesting case is when one does not want to penalize the entire Hilbert space attached to a vertex, but only a subspace. This is captured in the following lemma.
\begin{lemma}\label{lem:kitaev-ulgs}
  Take a connected simple ULG with Hilbert space $\hs$ for all vertices $s\in\set S$. Pick a non-empty subset of vertices $\set P\subsetneq\set S$ and a set of projectors $\Pi=\{\Pi_p\}_{p\in\set P}$ on $\hs$. For $\op H_\set P(G,\Pi):=\op H(G)+\sum_{p\in\set P}\proj p\otimes\Pi_p$, we have
  $\lmin(\op H_\set P(G,\Pi))\ge\mu\Omega(1/|\set S|^3)$, where $\mu=1-\max\{|\lambda_{\max}(\Pi_i^c\op U_{ij}\Pi_j^c)|:p_i,p_j\in\set P,i\neq j\}$ and $\op U_{ij}$ is the product of unitaries of a path connecting vertices $p_i$ and $p_j$.
\end{lemma}
\begin{proof}
  First note that the $\op U_{ij}$ are well-defined, since the ULG is simple and connected. Construct $\op W$ such that the root of the spanning tree is one of the penalized vertices, namely $r$ with projector $\Pi_r$. Then
  \begin{align*}
    \op W^\dagger\op H_P(G,\Pi)\op W&=\Delta(G)\otimes\1+\proj r\otimes\Pi_r+\sum_{p\neq r}\proj p\otimes\op R_p\Pi_p\op R_p^\dagger\\
                                    &=:\Delta\otimes\1+\op A
  \end{align*}
  where the $\op R_p$ are the product of unitaries connecting vertex $p$ with the root $r$. Following the notation of \cref{lem:kitaev-graphs}, we want to calculate the angle between the kernels of the Laplacian and the penalty terms. We write $\Pi_{\op A}:=\1\otimes\1-\sum_p\proj p\Pi_p$ for the projector onto the kernel of the penalty terms. Then
  \[
    \op W^\dagger\Pi_{\op A}\op W=\proj r\otimes\Pi_r^c+\sum_{p\neq r}\proj p\otimes \op R_p\Pi_p^c\op R_p^\dagger+\sum_{v\not\in\set P}\proj v\otimes\1\point
  \]
  Noting that the kernel of $\Delta(G)\otimes\1$ is spanned by $\{\ket{\Psi_\set V}\otimes\ket\phi:\ket\phi\in\hs\}$, we get
  \begin{align*}
    \cos\theta&=\max_{\ket\phi}\bra{\Phi_\set V}\bra\phi\op W^\dagger\Pi_\op A\op W\ket{\Phi_\set V}\ket\phi=\frac1{|\set V|}\max_{\ket\phi}\sum_{\ket v,\ket{v'}}\bra v\bra\phi\op W^\dagger\Pi_\op A\op W\ket{\Phi_\set V}\ket{v'}\ket\phi\\
               &=\frac{|\set V|-|\set P|}{|\set V|}+\frac1{|\set V|}\max_{\ket\phi}\bra\phi\Pi_r^c+\sum_{p\neq r}\op R_p\Pi_p^c\op R_p^\dagger\ket\phi\\
                &\le\frac{|\set V|-|\set P|}{|\set V|}
    +\frac{|P|-1}{|\set V|}(1+\cos\vartheta)\le
    1-\frac{1}{|\set V|}(1-\cos\vartheta)\comma
  \end{align*}
  where we used \cref{eq:kitaev} in \cref{lem:kitaev} in the last line with a bound on the angle between subspaces $\cos\vartheta=\max_{p\neq r}\sphericalangle(\supp\Pi_r^c,\supp \op R_p\Pi_p^c\op R_p^\dagger)$ and $|\set P|\ge2$. We can bound this further by
  \begin{align*}
    \cos^2\vartheta&=\max_{p\neq r}\max_{\substack{\ket\xi\in\supp\Pi_r^c\\ \ket\eta\in\supp \op R_p\Pi_p^c\op R_p^\dagger}}|\braket\xi\eta|^2=
    \max_{p\neq r}\max_{\ket\xi,\ket\eta}|\bra\xi\Pi_r^c\op R_p\Pi_p^c\op R_p^\dagger\ket\eta|^2\\
    &\equiv\max_{p\neq r}\max_{\ket\xi,\ket\eta}|\bra\xi\Pi_r^c\op R_p\Pi_p^c\ket\eta|^2\le\max_{p\neq r}|\lambda_\text{max}(\Pi_r^c\op R_p\Pi_p^c)|^2=:\lambda_\text{max}^2\point
  \end{align*}
  The rest follows \cref{lem:kitaev-graphs}:
  \[
    2\sin^2\frac\theta2=1-\cos\theta
    \ge1-1+\frac{1}{|\set V|}(1-\lambda_\textbf{max})=\frac{1-\lambda_\text{max}}{|\set V|}\comma
  \]
  and the claim follows.
\end{proof}

\section{Quantum Thue Systems}\label{sec:qts}
\subsection{Thue Systems}
Let us briefly recall the idea behind classical Thue systems, also known as string rewriting systems.
\begin{definition}\label{def:ts}
  A \emph{Thue system}---TS for short---is a tuple $(\Sigma,\set R)$ of a finite alphabet $\Sigma$ and a finite symmetric binary relation $R\subset\Sigma^*\times\Sigma^*$, where $\Sigma^*:=\bigcup_{i=0}^\infty\Sigma^i$ denotes the set of all \emph{strings} over the alphabet $\Sigma$.
\end{definition}
The binary relation $R$ is usually written as a set of rewrite rules $a\leftrightarrow b$ and they are naturally extended to other strings in $\Sigma^*$: if $a\leftrightarrow b\in \set R$, then $c\leftrightarrow d$ in $\set R$ if there exist $u,v\in\Sigma^*$ such that $c=uav$ and $d=ubv$. Thue systems with this extension---denoted by $\set R^*$---are a special case of abstract reduction systems and well-studied as computational models---see~\cite{Thomas2010}.

Thue systems are multiway systems, i.e.~starting from an initial string $a\in\Sigma^*$, exactly one substring is replaced at a time---in particular, this means that there might be branching points when the substitution is not unique.
For our purpose, it is enough to consider length-preserving substitutions, i.e.~any space $\Sigma^i$ should be invariant under $\set R$; we denote the length of any string $s$ with $|s|$.
In this case, there exists a natural representation of the Thue system over strings of length $N$ as a finite, undirected and not necessarily connected graph.
\begin{definition}\label{def:ts-graph}
  Let $(\Sigma,\set R)$ be a TS and $N\in\sset N$. The \emph{associated graph} $G=(\set V,\set E)$ for strings of length $N$ has vertices $\set V:=\set \Sigma^N$ and edges $\set E:=\{(a,b):a\leftrightarrow b\in\set R^*\}$. The \emph{Laplacian} of the TS is defined as the discrete Laplacian of the associated graph $G$. For brevity we just write $G=(\Sigma,\set R)$.
\end{definition}

\begin{example}\label{ex:ts}
  Take the alphabet $\Sigma:=\{a,b,c\}$ and $\set R:=\{c\leftrightarrow b,ab\leftrightarrow cc\}$. For example, starting from string $aab$, we can obtain a chain
  $aab \leftrightarrow acc \leftrightarrow acb \leftrightarrow abb \leftrightarrow ccb \leftrightarrow bcb \leftrightarrow bbb$.
  The entire graph for strings of length 3 for this example is shown in \cref{fig:ts-example}.
\end{example}

\begin{figure}
  \centering
  \begin{tikzpicture}[node distance = 14mm, text height = 1.5ex, text depth = .25ex, <->]
    \node (aac) at (0,0) {$aac$};
    \node (aab)[      right of = aac] {$aab$};  \draw (aac) -- (aab);
    \node (acc)[      right of = aab] {$acc$};  \draw (aab) -- (acc);
    \node (acb)[above right of = acc] {$acb$};  \draw (acc) -- (acb);
    \node (abc)[below right of = acc] {$abc$};  \draw (acc) -- (abc);
    \node (abb)[below right of = acb] {$abb$};  \draw (acb) -- (abb);  \draw (abc) -- (abb);
    \node (ccc)[below right of = abc] {$ccc$};  \draw (abc) -- (ccc);
    \node (ccb)[above right of = ccc] {$ccb$};  \draw (ccc) -- (ccb);  \draw (abb) -- (ccb);
    \node (cbc)[      right of = ccc] {$cbc$};  \draw (ccc) -- (cbc);
    \node (bcc)[below right of = ccc] {$bcc$};  \draw (ccc) -- (bcc);
    \node (bcb)[      right of = cbc] {$bcb$};  \draw (ccb) -- (bcb);  \draw (bcc) -- (bcb);
    \node (bbb)[      right of = bcb] {$bbb$};  \draw (bcb) -- (bbb);
    \node (cbb)[above  left of = bbb] {$cbb$};  \draw (cbc) -- (cbb);  \draw (cbb) -- (bbb);  \draw (ccb) -- (cbb);
    \node (bbc)[below  left of = bbb] {$bbc$};  \draw (cbc) -- (bbc);  \draw (bbc) -- (bbb);  \draw (bcc) -- (bbc);
    \node (cab)[below  left of = ccc] {$cab$};  \draw (cab) -- (ccc);
    \node (bab)[below right of = cab] {$bab$};  \draw (cab) -- (bab);  \draw (bab) -- (bcc);
    \node (bac)[below  left of = bab] {$bac$};  \draw (bac) -- (bab);
    \node (cac)[below  left of = cab] {$cac$};  \draw (cac) -- (cab);  \draw (cac) -- (bac);
    \node (aca) at (-2,-2) {$aca$};
    \node (aba)[right of = aca] {$aba$};  \draw (aca) -- (aba);
    \node (cca)[right of = aba] {$cca$};  \draw (aba) -- (cca);
    \node (cba)[above right of = cca] {$cba$};  \draw (cca) -- (cba);
    \node (bca)[below right of = cca] {$bca$};  \draw (cca) -- (bca);
    \node (bba)[below right of = cba] {$bba$};  \draw (cba) -- (bba);  \draw (bca) -- (bba);
    \node (caa) at (-1,-4) {$caa$};
    \node (baa)[right of = caa] {$baa$};  \draw (caa) -- (baa);
    \node (aaa) at ( 7,-4) {$aaa$};
  \end{tikzpicture}
  \caption{Undirected graph associated to strings of length 3 for the Thue system in \cref{ex:ts}.}
  \label{fig:ts-example}
\end{figure}
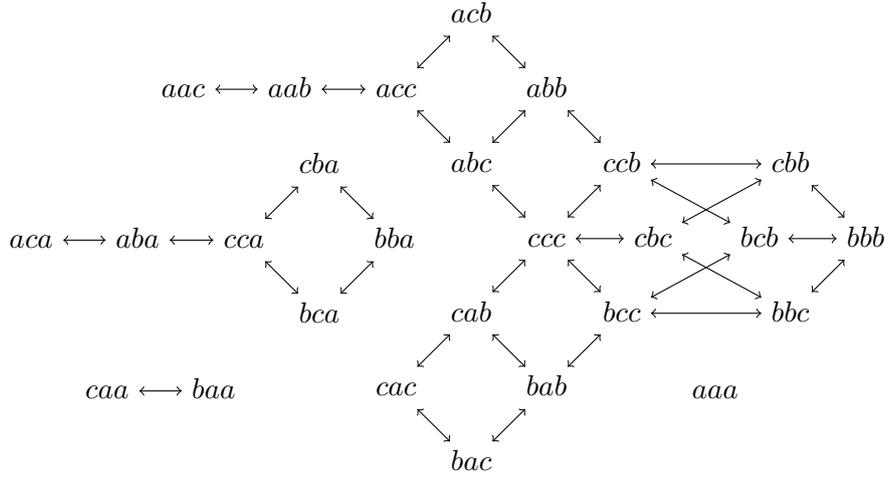

\begin{definition}\label{def:evolution}
  Let $(\Sigma,\set R)$ be a TS. We call a nonempty subset $\set U\subseteq\Sigma^*$ a \emph{valid evolution} if $\set U$ is closed under the transition rules $\set R^*$. We call $\set U$ \emph{irreducible} if there exists no valid evolution $\set U'\subsetneq\set U$.
\end{definition}
\begin{example}
  In \cref{ex:ts}, the set $\{caa, baa, aaa\}$ is a valid evolution, but it is not irreducible. The only irreducible evolutions for strings of length 3 in this example are the sets formed from the connected components.
\end{example}
It is immediate to see this one-to-one correspondence between connected associated graph components and irreducible evolutions.

We want to introduce a sense of locality to TS relations.
\begin{definition}\label{def:ts-locality}
  A TS $(\Sigma,\set R)$ is $k$-local, where $k:=\min\{i:\exists r\in\set R:r\in\Sigma^i\times\Sigma^i\}$.
\end{definition}
Observe how this definition is well-defined, as we required $\set R$ to be finite, cf. \cref{def:ts}.

\subsection{Quantum Thue Systems and their Hamiltonian}
We begin by generalizing the notion of Thue systems to the case where our alphabet has special quantum symbols, with rewriting rules being unitary operators between them. To work with these two alphabets, consider $\Sigma=\Sigma_{cl}\sqcup\Sigma_q$ as the union of two disjoint---classical and quantum---alphabets. For a string $s\in\Sigma^*$, write $\qn{s}$ for the number of letters from $\Sigma_q$ in $s$. This allows the following definition.
\begin{definition}\label{def:qts}
  A \emph{quantum Thue system} (QTS) is a quadruple $(\Sigma,\set R,\{\op U_r\}_{r\in\set R},\hs)$ of a bipartite alphabet $\Sigma=\Sigma_q\sqcup\Sigma_{cl}$, a relation $\set R$, a unitary operator $\op U_r$ for each rule $r\in\set R$ and a finite-dimensional Hilbert space \hs with the following properties:
  \begin{itemize}
    \item $(\Sigma,\set R)$ is a TS,
    \item $\qn{\cdot}$ is invariant under any rule $r\in\set R$,
    \item $\op U_r\in\U(\hs^{\otimes\qn{r}})$ for all $r\in\set R$.
  \end{itemize}
\end{definition}
The invariance of $\qn{\cdot}$ under a rule $r=s_1\leftrightarrow s_2$ allows to abbreviate $\qn{r}:=\qn{s_1}=\qn{s_2}$, which indicates the number of quantum letters the rule $r$ acts on.

We can again use the QTS to form sequences of strings: starting from a string $s\in\Sigma^*$, apply rules consecutively as for TSs. In addition, to each string $s\in\Sigma^*$, we attach a Hilbert space $\hs_s:=\hs^{\otimes\qn{s}}$: starting from some vector $v\in\hs_s$, each time a rule $r$ is applied to a substring, the corresponding unitary acts on the subspace wherever the rule matches, acting as identity everywhere else.

Analogous to \cref{fig:ts-example}, we can build a graph for strings of length $L$ for any QTS---where each edge is labelled by the acting unitary. The following lemma should therefore not come as a surprise.
\begin{lemma}\label{lem:qts-hamiltonian}
  Any $k$-local QTS $(\Sigma,\set R,\{\op U_r\},\hs)$ restricted to strings of a certain length $N\ge k$ is also a ULG. Furthermore, the associated Hamiltonian for strings of length $N$ is isomorphic to a geometrically $k$-local and translationally-invariant Hamiltonian on a chain $(\sset C^\Sigma\otimes\hs)^{\otimes N}$ with the same spectrum.
\end{lemma}
\begin{proof}
  We explicitly define the ULG $(\set S,\set R')$ for strings of length $N\ge k$. The vertex set $\set S:=\Sigma^N$ is straightforward. For every $r\in\set R$ denoted $s_1\leftrightarrow s_2$, define the ULG edges $(us_1v\leftrightarrow us_2v,\1_\hs^{\otimes|u|}\otimes\op U_r\otimes\1_\hs^{\otimes|v|})$ for any $u,v\in\Sigma^*$---potentially extending $\op U_r$ to $\hs^{\otimes N-k}$ acting trivially on classical substrings---such that $us_1v\in\Sigma^N$. It is straightforward to verify that this defines a valid ULG.

  The second claim follows from the canonical isomorphism between the two Hilbert spaces $(\sset C^\Sigma)^{\otimes N}\otimes\hs^{\otimes N}\isomap(\sset C^\Sigma\otimes\hs)^{\otimes N}$, i.e.~a simple rearrangement. Conjugating the associated Hamiltonian of the ULG with this isomorphism proves the second claim.
\end{proof}
As QTSs are also ULGs, we will be---without always specifying the string length restriction explicitly---using ULG terminology for QTSs, e.g.~associated Hamiltonians, irreducible evolutions or speak of QTSs being simple.

\begin{lemma}\label{lem:qts-hamiltonian-compressed}
  Let the setup be as in \cref{lem:qts-hamiltonian}. Then the isomorphism extends to a Hamiltonian on the chain $(\sset C^{\Sigma_{cl}}\oplus(\sset C^{\Sigma_q}\otimes\hs))^{\otimes N}$ with the same spectrum up to multiplicities.
\end{lemma}
\begin{proof}
  Any rule on a classical substring acts identically on \hs, hence on the set of Hamiltonians with this property a conjugation of the isomorphic Hamiltonian in \cref{lem:qts-hamiltonian} $\op H$ with the projector $(\sset C^\Sigma\otimes\hs)^{\otimes N}\longhookrightarrow(\sset C^{\Sigma_{cl}}\oplus(\sset C^{\Sigma_q}\otimes\hs))^{\otimes N}$ is an isomorphism. The projector preserves the spectrum, up to multiplicities.
\end{proof}
In the following, we often gloss over the fact and simply assume that the associated Hamiltonian of a QTS is local in the sense of \cref{lem:qts-hamiltonian-compressed}.
Observe however that a ULG induced from a QTS is not necessarily simple, and it is easy to find a counterexample.

\subsection{Quantum Thue Systems as a Computational Model}
To use a QTS for computation, we need to mark some strings that have special meaning, e.g.~are input or output of the computation.

\begin{definition}\label{def:marker}
  For a QTS with alphabet $\Sigma$ and Hilbert space \hs, a \emph{marker} is any tuple $(s,\pi)$ where $s\in\Sigma^*$ and $\pi$ is a projector on some subspace of $\hs^{\otimes\qn{s}}$. The set of markers on strings of length $k$---called $k$-local markers---is denoted $\mathcal M^{(k)}$, and $\mathcal M:=\bigcup_k\mathcal M^{(k)}$.
\end{definition}
% markers and projectors
\newcommand{\sinp}{s_{\mathrm{inp}}} \newcommand{\Pinp}{\Pi_{\mathrm{inp}}}
\newcommand{\sout}{s_{\mathrm{out}}} \newcommand{\Pout}{\Pi_{\mathrm{out}}}
That is, we can specify a string $s$ and a configuration of the quantum part of this string as a specific state in the computation. It is useful to think of using one marker $\sinp$ to mark a string as the start of the computation, and a second one $\sout$ to mark the end; the quantum parts of the markers---$\Pinp$ and $\Pout$---then define the valid input and output of the computation.

\begin{definition}\label{def:qts-decides}
  Let $(\Sigma,\set R,\{\op U_r\},\hs)$ be a QTS and $\PromP=(\PromP_\yes,\PromP_\no)$ a promise problem. We introduce an \emph{encoding function} $\enc:\PromP \to \Sigma^*$, input and output markers $(\sinp,\Pinp),(\sout,\Pout)\in\mathcal M^{(n)}$ for some $n\in\sset N$. Then the QTS
  \begin{itemize}
    \item \emph{rejects} an instance $l\in\PromP$ if there exists a chain of rules in $\set R$ connecting $\enc(l)$ with two strings containing $\sinp$ and $\sout$, respectively, and $\bra\psi\Pinp+\op U^\dagger\Pout\op U\ket\psi\ge\epsilon$ for all $\ket\psi\in\hs$---here $\op U=\op U(l)$ stands for the product of unitaries along this chain, and $\Pinp$ and $\Pout$ are extended trivially to the entire chain in case $\qn{\sinp}>\qn{\enc(l)}$ or $\qn{\sout}>\qn{\enc(l)}$.
    \item \emph{accepts} $l$ if there exists a $\ket\psi\in\hs$ such that $\bra\psi\Pinp+\op U^\dagger\Pout\op U\ket\psi\le\epsilon/2$.
    \item \emph{decides} $\PromP$ if for all $l\in\PromP$, $l$ is accepted if $l\in\PromP_\yes$, and rejected for $l\in\PromP_\no$.
  \end{itemize}
\end{definition}
The rejection and acceptance threshold $\epsilon$ will depend on the class of promise problems that we want to decide. In particular, we want to allow this threshold to scale with the problem instance size, i.e.~$\epsilon=\epsilon(|l|)$, and thus indirectly with the time that a computation can take, as specified in the following definition.

\begin{definition}\label{def:computation-time}
  Let $Q$ be a QTS that decides $\PromP$. For an instance $l\in\PromP$, the \emph{history state} is defined as the irreducible evolution of the ULG containing $\enc(l)$.
\end{definition}
For a QTS with unambiguous transition rules---i.e.~where the history state is a line---the size of the history state simply corresponds to the runtime of the underlying computation.

We now want to describe a simple example for a QTS which can decide the following simple promise problem.
\begin{definition}[\textsc{Even Natural Number}]\leavevmode
  \begin{description}
    \item[Instance.] Natural number $n\in\sset N$.
    \item[Output.] \yes if $n$ even, otherwise \no.
  \end{description}
\end{definition}
\begin{example}
  Let the alphabet $\Sigma:=\{-,\star,\|\}$, where $\star$ is the only quantum symbol with Hilbert space $\sset C^2$. We define $\sout=\star\|$ and $\Pout=\proj1$. Let further
  \[
    \enc:\sset N\to\Sigma^* \mt{where} \enc(n):=\star\underbrace{--\ldots-}_\text{$n$ times}\|.
  \]
  $\sinp=\enc(l)$, $\sout=\star\|$, and $\Pinp=\Pout=\proj1$. We have a single rule $(\star-\leftrightarrow-\star,\op R)$ where $\op R$ is a rotation by $\pi/2$, i.e.~$\op R:=-\ketbra10+\ketbra01$. Then this QTS decides \textsc{Even Natural Number}.
\end{example}
\begin{proof}
  The proof is straightforward. Starting on the encoded input $\enc(n)$, the TS generates a sequence
  \[
    \star--\ldots-\| \quad\longmapsto\quad -\star-\ldots-\| \quad\longmapsto\quad \cdots \quad\longmapsto\quad --\ldots-\star\|,
  \]
  so there always exists a chain of rules that connects $\enc(n)$ with a string containing $\sout$. The decision is thus made by the content of the quantum part: for $n$ applications of the rule, starting from a vector $\ket v\in\sset C^2$, we apply $\op R$ $n$ times. Now take any state $\ket\psi\in\sset C^2$ and write $\ket\psi=a\ket0+b\ket1$. Then
  \begin{align*}
    \bra 0\Big(\Pinp&+(\op R^\dagger)^n\Pout\op R^n\Big)\ket 0\\
                    &=|\bra 0\op R^n\ket 1|^2
    =|\bra 0
    \op R^{n \bmod 2}
%    \left(\begin{cases}
%        \op R & \text{if $n$ odd,} \\
%        \1_2  & \text{otherwise}.
%    \end{cases}\right)
    \ket 1|^2
    =\begin{cases}
      0 & \text{if $n$ is odd,} \\
      1 & \text{otherwise,}
    \end{cases}
    \\
    \shortintertext{and}
    \bra 1\Big(\Pinp&+(\op R^\dagger)^n\Pout\op R^n\Big)\ket 1\\
                    &=1+|\bra 1\op R^n\ket 1|^2
    \ge1.
    \\
    \shortintertext{Therefore, if $n$ is even,}
    \bra\psi\Big(\Pinp&+(\op R^\dagger)^n\Pout\op R^n\Big)\ket\psi
    \ge|a|^2+|b|^2=1
  \end{align*}
  and $\ket0$ is an accepting state for odd $n$. The claim follows.
\end{proof}

\section{Hardness Results}\label{sec:hardness}
% Some useful notation
\newcommand{\J}{j} % map from L to natural numbers
\newcommand{\g}{g} % polynomial for QTS block size

\subsection{A Special Kind of Quantum Thue System}
We have seen that QTS can be used to answer simple problems. On the other hand, a more interesting question is whether there exists a universal QTS which can run any computation of a certain class of promise problems \textsf{C}, i.e.~is complete for \textsf{C}. Of particular interest in this setting is the question about scaling of the defining parameters for such a QTS: how big is the alphabet, what is the locality and how does the string length of the vertices in the history state scale, i.e.~for a promise problem $\PromP\in\textsf{C}$, does there exist a function $f$ such that for $l\in\PromP$, $|\enc(l)|=O(f(|l|))$? And what about the size of the history state?

For the complexity class \BQEXP, we have the following lemma.
\newcommand{\hist}[1][l]{\ensuremath{\mathcal M_#1}\xspace}
\newcommand{\mmevo}{\ensuremath{\mathcal M}\xspace}
\newcommand{\allowed}{\ensuremath{\mathfrak A}\xspace}
\begin{lemma}\label{lem:qts-special}
  For any \BQEXP promise problem $\PromP$, there exists a 2-local QTS $(\Sigma,\set R,\{\op U_r\},\sset C^2)$ which decides $\PromP$, and has the following uniform properties:
  \begin{enumerate}[label=\upshape\textbf{W\arabic*}]
    \item \label{itm:W-preserves} The alphabet has special characters $\set H\subset\Sigma$---heads---and $\set B\subset\Sigma$---boundaries, and a set of \emph{allowed pairs} $\allowed\subset\Sigma\times\Sigma$. All transition rules preserve any symbols from $\set B$ and the number of symbols in $\set H$ (denoted $|s|_h$ for a string $s$).
    \item \label{itm:W-decides} Let $\J:\PromP\longrightarrow\sset N$ be a map with $\J(l)=O(\exp\poly |l|)$, where $|l|$ denotes the size of instance $l\in\PromP$. The QTS decides instance $l$ on strings of this length, i.e.~$|\enc(l)|=\J(l)$. Both input and output penalty are 2-local markers containing precisely one head symbol, and $\sinp,\sout\in\allowed$.
    \item \label{itm:W-hist} For any $l\in\PromP$, the history state \hist is simple. All strings $s\in\hist$ are of the form $s\in\set B\times(\Sigma\setminus\set B)^*\times\set B$ (bracketed), and have one head $|s|_h=1$. Furthermore, all length-2 substrings of $s$ are in $\mathfrak A$, and the size $|\hist|=\poly(\J(l))$.
    \item \label{itm:W-others} For all other irreducible evolutions $\mmevo\neq\hist$, at least one of the following is true:
      \begin{itemize}
        \item $|s|_h=0\ \forall s\in\mmevo$,
        \item \mmevo is not bracketed (i.e.\ with a boundary symbol on the left and right ends),
        \item \mmevo can be broken up into $O(g(l))$-sized connected parts---where $g(l)=\poly(\J(l))$---each of which containing at least one string containing an invalid character tuple not in \allowed.
      \end{itemize}
  \end{enumerate}
\end{lemma}

\subsection{Hardness Result}
\begin{theorem}\label{th:hardness}
  $(2,42)$-\textsc{Hamiltonian} is \QMAEXP-hard.
\end{theorem}
\begin{proof}
  Following \cref{def:khamiltonian,def:qma}, we need to show that there exists a 1D translationally-invariant $2$-local Hamiltonian $\op H$ on $\hs=(\sset C^{42})^{\otimes n}$ with $O(1)$ local terms, such that either (a) $\lmin(\op H)\le\alpha$ or (b) $\lmin(\op H)\ge\beta$ with a polynomial promise gap $\beta-\alpha=\Omega(1/\poly n)$, and deciding between (a) and (b) is at least as hard as some \QMAEXP-hard promise problem. The proof will be a simple combination of our previously-collected results.

  \begin{enumerate}
    \item Let $\PromP$ be a promise problem in \QMAEXP. By \cref{def:qma}, the verification of $\PromP$ is a \BQEXP problem. By \cref{rem:bqerror}, we can assume without loss of generality that the accept and reject probabilities in \cref{def:qts-decides} are $1-\epsilon$ and $\epsilon$, respectively, where $\epsilon=1/3^{\poly |l|}$ to be specified below, where $|l|$ denotes the length of the problem input.
    \item By \cref{cor:qts-qma-exists}, we can thus create a QTS with properties as in \cref{lem:qts-special} that verifies $\PromP$: more specifically, for an instance $l\in\PromP$ and by \cref{itm:W-decides}, we know that this QTS verifies $l$ on strings of length $\J(l)$.
    \item By \cref{lem:qts-hamiltonian}, the QTS restricted to strings of length $\J(l)$ is also a ULG. Denote the Hamiltonian associated to this ULG by $\op H_l$, block-diagonal in the irreducible evolutions.
  \end{enumerate}

  \newcommand{\bonusH}{\op B_\text{heads}}
  \newcommand{\penaltyBd}{\op P_\text{boundaries}}

  With $\Gamma$ denoting the alphabet from \cref{def:alphabet} and $\J(l)$ denoting the number of systems, we define a Hamiltonian on the Hilbert space $(\sset C^\Gamma)^{\otimes \J(l)}$ as follows:
  \begin{equation}\label{eq:overall-hamiltonian}
    \op H:=\op H_l+\bonusH+p(l)(\penaltyBd+\op P)+\op P_\text{in/out},
  \end{equation}
  where
  \begin{itemize}
    \item $\penaltyBd$ penalizes any non-bracketed string (i.e.\ strings without a boundary symbol on at least one end)---cf. \cref{itm:W-hist},
    \item $\bonusH$ acting on a string $\ket s \in (\sset C^{\Gamma})^{\otimes \J(l)}$ gives a bonus of $|s|_h$, according to how many head symbols there are in $s$,
    \item $\op P$ penalizes any character tuple not in $\mathfrak A$,
    \item $p(l)$ is a function used to scale the penalties, which will be specified later, but---keeping \cref{rem:qma-qmaexp-information} in mind---must not exceed $p(l)=\poly j(l)$.
  \end{itemize}

  Penalizing non-bracketed strings follows an idea by~\cite{Gottesman2009}. With $\penaltyBd$, we give a 1-local bonus of size $1$ to brackets appearing anywhere, but a penalty of $1/2$ to them appearing next to any other symbol; since no transition rule ever moves the boundaries, this gives a uniform energy shift to all strings with brackets. The unique highest-bonus string will have a bracket appearing at the start and end with a bonus of size 1.

  The encoding and output penalties $\Pinp=\Pout=\proj1$ are translationally-invariantly extended to the entire chain, i.e.~on Hilbert space $(\hs_{cl}\oplus\hs_q)^{\otimes \J(l)}$, we act with the 2-local projector
  \[
    \op P_\text{in}:=\sum_{i=1}^{\J(l)-1}(\proj{\sinp}\oplus\Pinp)_{i,i+1},
  \]
  and analogously for $\op P_\text{out}$.

  \paragraph{Completeness.}
  \newcommand{\hstate}{\Psi_l}
  Assume $l$ is a \yes-instance, and denote the history state as an eigenvector of $\op H_l$ with $\ket\hstate$, which by \cref{itm:W-hist} is also an eigenstate of $\bonusH$, $\penaltyBd$ and $\op P$. A direct calculation yields
  \begin{align*}
    \bra\hstate\op H\ket\hstate&=
    \overbrace{\bra\hstate\op H_l\ket\hstate}^{=0}+
    \overbrace{\bra\hstate\bonusH\ket\hstate}^{=-1}+
    p(l)(
    \overbrace{\bra\hstate\penaltyBd\ket\hstate}^{=-1}+
    \overbrace{\bra\hstate\op P\ket\hstate}^{=0}
    )\\
    &\hphantom{=}+
    \bra\hstate\op P_\text{in/out}\ket\hstate
    \\
    &=-2+\bra\hstate\op P_\text{in/out}\ket\hstate.
  \end{align*}
  By \cref{itm:W-decides}, we further know that at least one vertex in $\ket\hstate$ has the in- and output substrings $\sinp$, $\sout$, and because $|\sinp|_h=|\sout|_h=1$, there is at most one such substring match for every vertex. As an upper bound, we can thus assume that the penalty applies exactly once in every vertex---i.e.~$\bra\hstate\op P_\text{in/out}\ket\hstate\le\epsilon$, and conclude $\bra\hstate\op H\ket\hstate\le-2+\epsilon$.
  % note: under slightly stricter conditions, i.e. that the input output penalties apply in only a fraction of the vertices (which is the case for the Wheelbarrow), we get a scaling \epsilon / T here; this would introduce another scaling factor. I personally think we should leave this out, since the scaling difference between yes and now is polynomial---doesn't matter if it's a T^2 or T^3 difference.

  \paragraph{Soundness.}
  Assume $l$ is a \no-instance. We need to lower-bound the lowest energy eigenvalue of $\op H$, and since we know that $\op H$ is block-diagonal in the irreducible evolutions, we can bound each block separately---the history state block given in \cref{itm:W-hist} and any other irreducible evolution block characterized by \cref{itm:W-others}. Without loss of generality we can therefore assume that $\ket\psi$ is completely supported on a single block of $\op H$ (but not necessarily an eigenvector).

  Take any $\ket\psi$ with support constrained to the history state block. As in the completeness part, a direct calculation allows the estimate
  \begin{align*}
    \bra\psi\op H\ket\psi&=
    \bra\psi\op H_l\ket\psi+
    \overbrace{\bra\psi\bonusH\ket\psi}^{=-1}+
    p(l)(
    \overbrace{\bra\psi\penaltyBd\ket\psi}^{\ge-1}+
    \overbrace{\bra\psi\op P\ket\psi}^{\ge0}
    )\\
    &\hphantom=+
    \bra\psi\op P_\text{in/out}\ket\psi\\
    &\ge-2+\bra\psi\op H_l+\op P_\text{in/out}\ket\psi.
  \end{align*}
  We can now apply \cref{lem:kitaev-ulgs} to the last expression. By \cref{itm:W-decides,def:qts-decides}, we obtain a bound
  $\bra\psi\op H_l+\op P_\text{in/out}\ket\psi \ge (1-\epsilon)/|\hist|^3$. Observe how this lower bound scales $\propto1/|\hist|^3$, whereas for \yes-instances the upper bound scales constant $\propto\epsilon$.
  % note: this is where the scaling would be \epsilon / T.
  Since we want the lower bound for \no-instances---$\beta$---and the upper bound for \yes-instances---$\alpha$---to be separated by at least some $\beta-\alpha=\Omega(1/\poly\J(l))$, cf.~\cref{def:khamiltonian}, we need to amplify the accepting probability to $\epsilon=O(1/|\hist|^4)=O(1/\poly\J(l)^4)$. Observe that this does not exceed the allowed amplification, which is only limited to $O(1/3^{\poly|l|})$.

  We proceed to show lower bounds for all other minimum valid evolutions, following \cref{itm:W-others}.
  Assume we are in a block with 0 heads, which is well-defined by \cref{itm:W-preserves}. The bonus term $\bonusH$ vanishes on this subspace while all other operators in \cref{eq:overall-hamiltonian} are positive semi-definite, so we obtain a lower bound of $\bra\psi\op H\ket\psi\ge0$ for any state solely supported there.

  Analogously, non-bracketed blocks can be bounded by a direct calculation, as $\penaltyBd$ penalizes all vertices equally: any non-bracketed state $\ket\psi$ for a block with $h$ heads satisfies $\bra\psi\op H\ket\psi\ge-h+p(l)$. It thus suffices to set $p(l)\ge \J(l)$, as the number of possible heads on a string is limited by its length, i.e.~$h\le\J(l)$.

  The last blocks remaining are the ones with $\g(\J(l))$-sized connected parts with invalid tuples, where $g(n)=\poly n$ as defined in \cref{itm:W-others}. First observe that this part of the ULG is not necessarily simple, so we remove the transitions which allow non-trivial loops without breaking the graph up into multiple parts. We \emph{then} split this graph into $\g(\J(l))$-sized connected components by temporarily removing further edges from it, which yields a Hamiltonian for a sparser graph $\op H'$. Since adding any edges back in corresponds to adding a positive semi-definite matrix to $\op H'$, it suffices to lower-bound the spectrum of $\op H'$ on this subspace. Note that we do not remove vertices or change any penalties, so in particular all the diagonal operators in \cref{eq:overall-hamiltonian} remain untouched.

  Hence assume $\ket\psi$ has support in one of the slices of size upper-bounded by $\g(\J(l))$ with $h$ heads, such that at least one vertex picks up a penalty from $\op P$. Again applying \cref{lem:kitaev-ulgs}, we obtain a bound $\bra\psi\op H\ket\psi\ge\bra\psi\op H'\ket\psi\ge-h-1+p(l)\times\Omega(1/\g(\J(l))^3)$. We therefore have to scale $p$ to e.g.~$p(l)\ge \g(\J(l))^5$, which is still allowed by \cref{rem:qma-qmaexp-information} (namely, $p$ is polynomial in $l$). This concludes the proof.
\end{proof}

What remains to be shown is the existence of a QTS as in \cref{lem:qts-special}. The next section will provide an explicit construction, finalizing the proof of our main result. This construction is meant as a proof-of-concept---the model we present can be modified in numerous ways and is likely not optimal. It does, however, make heavy use of our newly-developed methods such as branching, thus reducing the local dimension of the underlying Hamiltonian to $42$, as compared to the hitherto best result by \cite{Gottesman2009} which is larger by at least several orders of magnitude.

\section{Turing's Wheelbarrow}\label{sec:wheelbarrow}
\subsection{Introduction}\label{sec:wheelbarrow-intro}
Turing's Wheelbarrow is our constructive proof of a QTS with properties as mentioned in \cref{lem:qts-special}.
The QTS will be optimized for local dimension and locality---every transition rule will be 2-local and act on strings from an alphabet $\Gamma$ with 48 characters.
We describe the QTS by explicitly writing out all transition rules of the QTS and then prove the properties from \cref{lem:qts-special}. Finally, in \cref{sec:Reduction} we reduce its local dimension down to 42.

The conceptual idea of the Wheelbarrow QTS is the following. To build a QTS which can decide a promise problem $\PromP\in\BQEXP$, we first prefix the original circuit $C_l$ deciding an instance $l\in\PromP$ by another circuit which verifies that a number of ancillas necessary for $C_l$ are correctly initialized to $\ket0$. On some extra ancillas, we write out the problem instance $l$, and also leave an unconstrained section of qubits available for $C_l$. This witness section, problem instance and the leftover ancillas are then fed into $C_l$, and the output wire contains $\ket{\text{out}}=\cos((p_a+p_\text{out})/3)\ket 0+\sin((p_a+p_\text{out})/3)\ket 1$ for the amplitudes $p_a$---all ancillas being 0---and $p_\text{out}$---the circuit output of $C_l$ on the ancillas and problem instance. This overall circuit, denoted $C'_l$, is shown in \cref{fig:augmented-verifier}.

\label[author=Maris]{What is meant by "language can be written out"?}
It is clear that this augmented circuit family $(C'_l)_{l\in\PromP}$ is in the same uniformity class as the original circuit family $(C_l)_{l\in\PromP}$, and we can thus define these circuits with output $\ket{\text{out}}$ to be a separate \BQEXP problem $\PromP'$. By \cref{prop:bqpqrm} and its proof, this new promise problem can be decided by a family of \BQEXP QRMs with the special property that the head motion and all internal QRM states are classical---cf.~\cref{fig:qrm-head}.

Using the \knownth{Solovay-Kitaev theorem}~\cite[appx.~3]{Nielsen2010}, the head unitary of such a QRM can in turn be efficiently rewritten as a circuit $R$ using the following small set of gates.
\begin{remark}\label{rem:gateset}
  \toffoli, \swap and a classically-controlled quantum-universal unitary together with at least one classical and quantum ancilla is universal for quantum computation and exactly universal for classical computation.
\end{remark}
In particular, S-K tells us that since the head circuit $R=R(l)$ depends on the problem instance $l$---as it needs to write the instance out---and the size of this circuit is $|R(l)|=\poly|l|$. The Wheelbarrow QTS which we construct will then be able to execute this head cyclically on a ring of qubits, where the execution is halted once the QRM terminates: as the QRM motion is deterministic, the runtime will be $\exp\poly |l|$-bounded, as required for a \BQEXP computation.

The first step is to bootstrap the QRM head $\op U$. Starting from an initially empty string, we use a number of rules to translate the string length $N$ into a circuit description of $\op U$ on the left side of the string. This section will have size $\approx\log_6N$, as we need 6 instruction symbols---a classically-controlled unitary \${:U}, \toffoli \${:T}, ancilla-checking symbol \${:A}, swap \${:S}, left-shift symbol \$r and halt \${:H}. The remaining right side of the string will act as classical and quantum tape that the computation runs on.
\Cref{fig:circuit-to-desc} outlines how a circuit can be translated into such a 6-ary circuit description.

\begin{figure}
  \centering
  \begin{tikzpicture}[
      draw=black,line width=.8pt,
      box/.style = {rounded corners=.5}
    ]
    \newcommand{\cgate}[3]{
      \draw[] (#1,#2) -- (#1,#2+1);
      \fill (#1,#2+1) circle [radius=2pt];
      \draw[fill=white,box] (#1-.4,#2-.4) rectangle (#1+.4,#2+.4);
      \node[align=center] at (#1,#2) {#3};
    }
    \draw[] (0,0) -- (12,0);
    \foreach\i in {0,1,2,1+0,1+1,1+2,2+0,2+1,2+2}
    {\draw[] (0,1+\i/4) -- (12,1+\i/4);}

    \cgate{2}{0}{$\op R_{\frac13}$}
    \cgate{11}{0}{$\op R_{\frac13}$}
    \draw[fill=white,box] (1,.75) rectangle (3,2.75);
    \draw[fill=white,box] (3.5,1.75) rectangle (5,2.75);
    \draw[fill=white,box] (7.5,.75) rectangle (9.5,3.75);

    \node[align=center] at (2,1.75) {ancillas\\ all $0$?};
    \node[align=center] at (4.25,2.25) {write $l$};
    \node[align=center] at (8.5,2.25) {quantum\\ verifier\\ circuit};
    \node[align=center] at (10.25,.8) {out};
    \node[align=center] at (6.25,2.75) {instance $l$};
    \node[align=right] at (-.4, 0) {$\ket 0$};
    \node[align=right,rotate=90] at (-.4,1.75) {ancillas};
    \node[align=center,rotate=90] at (-.4-.4,3.25) {uncon-\\ strained\\ witness};
    \node[align=left] at (12.5, 0) {$\ket{\text{out}}$};
  \end{tikzpicture}
  \caption{An augmented quantum verifier circuit. The circuit uses one ancilla $\ket 0$ to verify that as many ancillas as necessary for the computation are set to $0$, rotating the single guaranteed $\ket 0$ ancilla by $\pi/3$ if this is not the case. On some ancillas, the problem instance $l$ is written out. Another rotation by $\pi/3$ is applied depending on the output of the verifier circuit. The overall output state then takes the form $\ket{\text{out}}=\cos((p_a+p_\text{out})/3)\ket 0+\sin((p_a+p_\text{out})/3)\ket 1$.}
  \label{fig:augmented-verifier}
\end{figure}
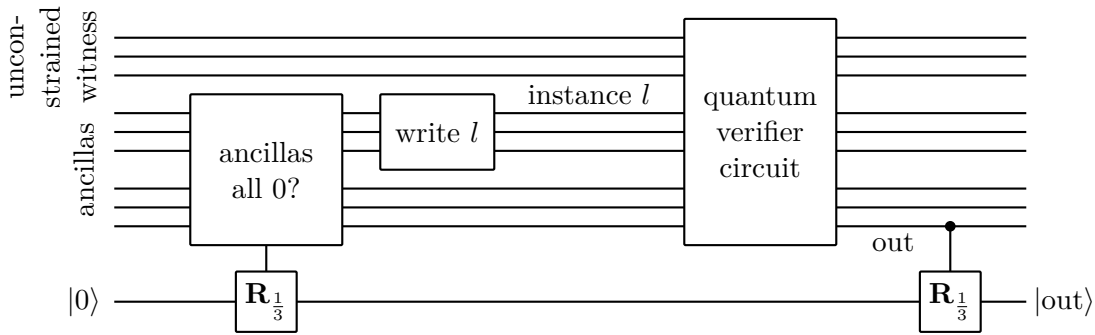

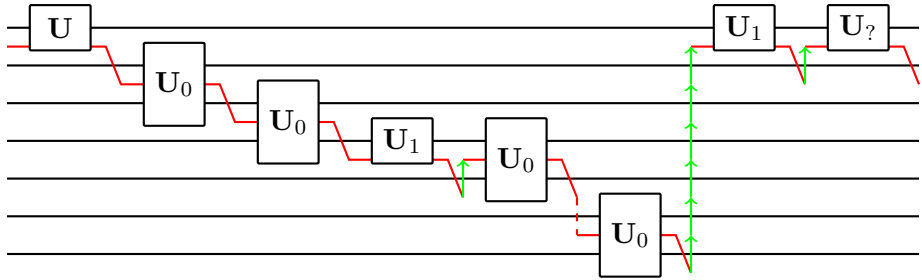
\begin{figure}
  \centering
  \begin{tikzpicture}[
      draw=black,line width=.8pt,
      box/.style = {rounded corners=.5},
      down/.style = {red},
      up/.style = {->,green}
    ]
    \newcommand{\gate}[4]{
      \draw[fill=white,box] (#1*1.5-.4,#2/2-.3) rectangle (#1*1.5+.4,#2/2+#3/2+.3);
      \node[align=center] at (#1*1.5,#2/2+#3/4) {#4};
    }
    \foreach\i in {0,...,6}
    {\draw[] (-0.7,\i/2) -- (11.3,\i/2);}

    \draw[down] (-.7,2.75) -- (.6,2.75) -- (.8,2.25) -- (2.1,2.25) -- (2.3,1.75) -- (3.6,1.75) -- (3.8,1.25) -- (5.1,1.25) -- (5.3,.75);
    \draw[down] (5.3,1.25) -- (6.6,1.25) -- (6.8,.75);
    \draw[down,dashed] (6.8,.75) -- (6.8,.25);
    \draw[down] (6.8,.25) -- (8.1,.25) -- (8.3,-.25);
    \draw[down] (8.3,2.75) -- (9.6,2.75) -- (9.8,2.25);
    \draw[down] (9.8,2.75) -- (11.1,2.75) -- (11.3,2.25);

    \draw[up] (5.3,.75) -- (5.3,1.25);
    \draw[up] (8.3,-.25) -- (8.3,.25);
    \draw[up] (8.3,.25) -- (8.3,.75);
    \draw[up] (8.3,.75) -- (8.3,1.25);
    \draw[up] (8.3,1.25) -- (8.3,1.75);
    \draw[up] (8.3,1.75) -- (8.3,2.25);
    \draw[up] (8.3,2.25) -- (8.3,2.75);
    \draw[up] (9.8,2.25) -- (9.8,2.75);

    \gate{0}{6}{0}{$\op U$}
    \gate{1}{4}{1}{$\op U_0$}
    \gate{2}{3}{1}{$\op U_0$}
    \gate{3}{3}{0}{$\op U_1$}
    \gate{4}{2}{1}{$\op U_0$}
    \gate{5}{0}{1}{$\op U_0$}
    \gate{6}{6}{0}{$\op U_1$}
    \gate{7}{6}{0}{$\op U_?$}
  \end{tikzpicture}
  \caption{Example on how to translate a sample circuit section into a program description, where $\op U_0$, $\op U_1$ and $\op U_?$ can stand for any unitary gate. Starting from the top left, the description here is \${:1:0:0:1r:0:0r:0:0rrrrrr:1r:?}. The dashed line stands for a normal identity \${:0r:0}, as $\cnot^2=\1$. The rhombus \$r is a special symbol that shifts the current gate position up by one; as in each successive step the position moves down by one by default, it suffices to only have this one special shift symbol. The last gate is a special identity to be used to initialize ancillas and penalize a section of the output.}
  \label{fig:circuit-to-desc}
\end{figure}

The QRM is then executed: for every round, a program bit is taken from the left side of the string, moved towards the tape and then applied to the leftmost two data qubits. The leftmost data qubit is then picked up and carried to the right, where it is deposited. The revert action is similar, only that the rightmost data qubit is picked up and moved to the left of the tape. \Cref{fig:wheelbarrow-gates} illustrates both operations. The execution runs until the underlying ring machine terminates, which can be determined using a special halt operation \${:H} which only proceeds if the tape data is not in a halting configuration.

This also explains the choice of \emph{Turing's Wheelbarrow} as name for this QTS: qubits and program symbols are moving across the tape in two cyclic motions, mimicking a busy worker carrying and depositing information in a wheelbarrow.

\begin{figure}
  \centering
  \subfigure{
    \includegraphics[width=0.47\linewidth]{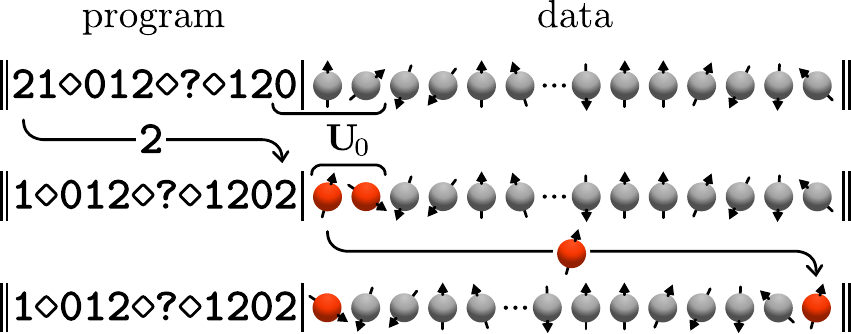}
  }
  \subfigure{
    \includegraphics[width=0.47\linewidth]{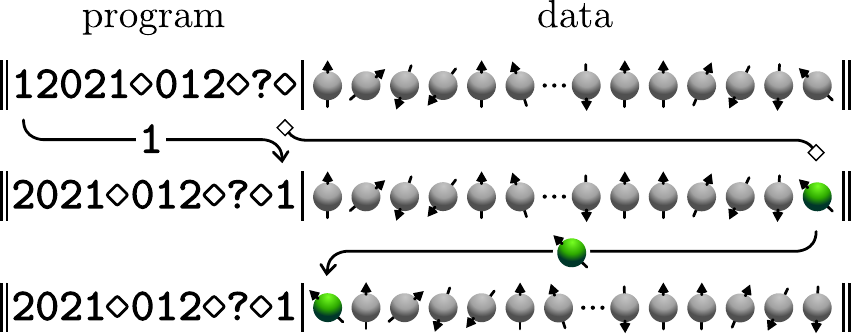}
  }
  \caption{The two actions that can be performed by the wheelbarrow construction. On the left, we apply a gate $\op U_0$ corresponding to the rightmost program bit \${:0}. The ring of qubits is then rotated by one, which is the default downwards shift as mentioned in \cref{fig:circuit-to-desc}. On the right, the special action of the \$r symbol is depicted: it signalises the rightmost qubit to move back to the left side. After either action, the program string is rotated by one.}
  \label{fig:wheelbarrow-gates}
\end{figure}

\subsection{Notation}
For convenience, we define a special notation to describe the construction of Turing's Wheelbarrow. We begin by introducing the alphabet and tape.
\begin{definition}\label{def:alphabet}
  Let $\Gamma := \Gamma_{cl} \sqcup \Gamma_q$ denote the alphabet consisting of 48 symbols where
  \begin{align*}
    \Gamma_{cl} := \Bigl\{ \;
      &  \$|,\$?,\$!,\$I, \quad
      \$g,\$G,\$b,\$B, \quad
      \$+,\$X,\$Y,\$Z,\$C,\$D, \\
      &  \${:U},\${:T},\${:A},\${:S},\${:H},\$r, \quad
      \${@U},\${@T},\${@A},\${@S},\${@H},\${@r}, \quad
      \$o,\${OU},\${OT},\${OA},\${OS},\${OH},\$R,\$w, \\
      &  \${:!},\${*0},\${*1},\${;1},\${"1}, \quad
      \${:0},\${:1}, \quad
      \${@0},\${@1}
    \; \Bigr\}, \\
    \Gamma_q := \Bigl\{ \;
      &\${-:a},\${-@a},\${-;a},\${-"a},\${-*a}
    \; \Bigr\}.
  \end{align*}
  These two sets correspond to the classical and quantum symbols, respectively, and are of size $\abs{\Gamma_{cl}} = 43$ and $\abs{\Gamma_q} = 5$. The set of head characters is
  \begin{align*}
    \set H:=\Gamma\setminus\{ \$|, \$?, \${:0}, \${:1}, \${:U}, \${:T}, \${:A}, \${:S}, \${:H}, \$r, \$o, \$w, \$R, \${-:a} \},
  \end{align*}
  and the boundary characters are $\set B:=\{\$|,\$!,\$?\}$.
\end{definition}
The number of alphabet characters can be further reduced to 39, 3 of which are quantum, see \cref{cor:qts-qma-exists}. For reasons of clarity, we use a slightly larger alphabet in this construction.

We will generally use the letters $x,y,z$ as placeholders for program symbols---denoting any of the symbols \${:U}, \${:T}, \${:A}, \${:S}, \${:H}, \$r as \${:x}, or alternatively \${@U}, \${@T}, \${@A}, \${@S}, \${@H}, \${@r} as \${@x}, which is always clear from the context. The symbol \${:U} encodes a classically-controlled unitary, \${:T} a \toffoli, \${:A} an ancilla, \${:S} a \swap, \${:H} a halt and \$r a special tape revert symbol.

We now introduce the notation for transition rules.
\begin{definition}\label{def:rules}
  We write a transition rule $xy \leftrightarrow zw$ of a quantum Thue system as
  \[
    \begin{rules}
      \${>:x>:y}\\
      \${<:z<:w}
    \end{rules}.
  \]
  The blue shading is used to indicate the location on the tape where the transition rule is applied.
  Note that, by construction, transition rules are symmetric, i.e.\
  \[
    \begin{rules}
      \${>:x>:y}\\
      \${<:z<:w}
    \end{rules}
    \mt{is equivalent to}
    \begin{rules}
      \${>:z>:w}\\
      \${<:x<:y}
    \end{rules}.
  \]
  If the first rule is associated with a non-trivial unitary $\op U$, the inverse rule is associated with the adjoint $\op U^\dagger$.
\end{definition}

As in \cref{def:ts}, we never need to write out the values of the qubits anywhere. In fact, the only place where the associated Hilbert space comes in is when we want to apply a quantum gate to the qubits (see \cref{sec:computation}). As an example, consider the action of swapping two neighbouring qubits. The Thue system itself does not notice this, e.g. we would have a transition with an explicit comment on the Hilbert space unitary, i.e.
\[
  \begin{rules}
    \${>-@a>-:a}\\
    \${<-:a<-@a}
  \end{rules}
  \quad\text{where the associated Hilbert spaces are swapped with $\op U=\swap$.}
\]
To emphasize that the subspaces are in fact swapped, we generally use the letters $a, b, c, d$ to label different quantum subsystems. \emph{This is only to facilitate notation!} In principle, we could stick to the letter $a$ and write out the swap action for every transition rule where this is relevant. But because we believe it is easier to read and most of the non-trivial unitaries that we use are swaps, we simply write
\[
  \begin{rules}
    \${>-@a>-:b}\\
    \${<-:b<-@a}
  \end{rules}
\]
which is self-explanatory.

\subsection{Transition Rules and History State}\label{sec:wheelbarrow-rules}
\begin{figure}
  \centering
  \centerline{%\documentclass[]{scrartcl}
%
%\input{../symphony}
%
%\usepackage{tikz}
%\usetikzlibrary{calc,arrows,patterns,decorations.pathmorphing,backgrounds,fadings,fit}
%
%\input{../qss}
%
%%\standaloneenv{rules}

%%%%%%%%%%%%%%%%%%%%%%%%%%%%%%%%
%%%%%%%%%%%%%%%%%%%%%%%%%%%%%%%%

%\begin{document}
\resizebox{5.5in}{!}{
\centerline{\begin{tabular}{ccc}
counting phase & \multicolumn{2}{|c}{computation phase} \\ \midrule
\begin{rules}
	\${|>I>:0-:a-:a:c.-:a:c-:a:c-:ab|}\\
	\${|<:0<>I>-:a-:a:c.-:a:c-:a:c-:ab|}\\
	\rvdots\\
	\${|:0-:a-:a:c.-:a:c-:a:c<-:a<>I>b|}\\
	\${|:0-:a-:a:c.-:a:c-:a:c-:a<o<>B>|}\\
	\${|:0-:a-:a:c.-:a:c-:a:c-:a>o<>G<|}\\
	\${|:0-:a-:a:c.-:a:c-:a:c->:a><G<o|}\\
	\${|:0-:a-:a:c.-:a:c-:a>:c><G<-:ao|}\\
	\rvdots\\
	\${|>:0><G<-:a-:a:c.-:a:c-:a:c-:ao|}\\
	\${>|><G<:0-:a-:a:c.-:a:c-:a:c-:ao|}\\
	\${<|<>+>:0-:a-:a:c.-:a:c-:a:c-:ao|}\\
	\${|<:U<>Y>-:a-:a:c.-:a:c-:a:c-:ao|}\\
	\${|>:U<>Z<-:a-:a:c.-:a:c-:a:c-:ao|}\\
	\${>|><C<:U-:a-:a:c.-:a:c-:a:c-:ao|}\\
	\${<|><D>:U-:a-:a:c.-:a:c-:a:c-:ao|}\\
	\${|<:U<>D>-:a-:a:c.-:a:c-:a:c-:ao|}\\
	\${|:U<-@a<>g>-:a:c.-:a:c-:a:c-:ao|}\\
	\${|:U-@a<-:a<>g>:c.-:a:c-:a:c-:ao|}\\
	\rvdots\\
	\${|:U-@a-:a:c.-:a:c-:a<:c><g>-:ao|}\\
	\${|:U-@a-:a:c.-:a:c-:a:c<-:a<>g>o|}\\
	\${|:U->@a>-:a:c.-:a:c-:a:c-:a<o<>g>|}\\
	\${|:U<-:a<>-@a>:c.-:a:c-:a:c-:ao<b<|}\\
	\rvdots\\
	\${|:U-:a:c.-:a:c-:a:c<-:a<>-@a>ob|}\\
	\${|:U-:a:c.-:a:c-:a:c-:a<o<>-@a>b|}\\
	\${|:U-:a:c.-:a:c-:a:c-:ao<-:a<>B>|}\\
	\${|:U-:a:c.-:a:c-:a:c-:ao>-:a><G<|}\\
    \rvdots
\end{rules}
&
\begin{rules}
    \${>|>:H:U:S:S:T.ro:0:1-:a-:a:c.-:a>:c><G<|}\\
    \${<?<:H:U:S:S:T.ro:0:1-:a-:a:c.>-:a><G<:c|}\\
    \rvdots\\
    \${?>:H<>G<:U:S:S:T.ro:0:1-:a-:a:c.-:a:c|}\\
    \${>?<>G<:H:U:S:S:T.ro:0:1-:a-:a:c.-:a:c|}\\
    \${<!<>g>:H:U:S:S:T.ro:0:1-:a-:a:c.-:a:c|}\\
    \${!<:H<>g>:U:S:S:T.ro:0:1-:a-:a:c.-:a:c|}\\
    \${>!>:H<:U<>g>:S:S:T.ro:0:1-:a-:a:c.-:a:c|}\\
    \${<?<>@H>:U<:S<>g>:S:T.ro:0:1-:a-:a:c.-:a:c|}\\
    \${?<:U<>@H>:S<:S<>g>:T.ro:0:1-:a-:a:c.-:a:c|}\\
    \rvdots\\
    \${?:U:S:S:T.<r<>@H>o:0:1-:a-:a:c.-:a<:c<>g>|}\\
    \${?:U:S:S:T.r<:H<>OH>:0:1-:a-:a:c.-:a:c<b<|}\\
    \${?:U:S:S:T.r:H<o<>@0>:1-:a-:a:c.-:a:cb|}\\
    \${?:U:S:S:T.r:Ho<:1<>@0>-:a-:a:c.-:a:cb|}\\
    \rvdots\\
    \${?:U:S:S:T.r:Ho:1-:a-:a:c.-:a<:c<>@0>b|}\\
    \${?:U:S:S:T.r:Ho:1-:a-:a:c.-:a:c<:0<>B>|}\\
    \${?:U:S:S:T.r:Ho:1-:a-:a:c.-:a:c>:0<>G<|}\\
    \${?:U:S:S:T.r:Ho:1-:a-:a:c.-:a>:c<>G<:0|}\\
    \rvdots\\
    \${>?><G<:U:S:S:T.r:Ho:1-:a-:a:c.-:a:c:0|}\\
    \${<!<>g>:U:S:S:T.r:Ho:1-:a-:a:c.-:a:c:0|}\\
    \${!<:U<>g>:S:S:T.r:Ho:1-:a-:a:c.-:a:c:0|}\\
    \${>!>:U<:S<>g>:S:T.r:Ho:1-:a-:a:c.-:a:c:0|}\\
    \${<?<>@U>:S<:S<>g>:T.r:Ho:1-:a-:a:c.-:a:c:0|}\\
    \rvdots\\
    \${?:S:S:T.r<:H<>@U>o:1-:a-:a:c.-:a:c<:0<>g>|}\\
    \${?:S:S:T.r:H<:U<>OU>:1-:a-:a:c.-:a:c:0<b<|}\\
    \${?:S:S:T.r:H:U<o<>:!>-:a-:a:c.-:a:c:0b|}\\
    \${?:S:S:T.r:H:Uo<:1<>-;a>-:a:c.-:a:c:0b|}\\
    \${?:S:S:T.r:H:Uo>:1<>-"a<-:a:c.-:a:c:0b|}\\
    \${?:S:S:T.r:H:Uo<>@1<>-:a-:a:c.-:a:c:0b|}\\
    \${?:S:S:T.r:H:Uo<-:a<>@1>-:a:c.-:a:c:0b|}\\
    \rvdots\\
 \end{rules}
&
\begin{rules}
    \${?>r<>G<:H:U:S:S.:Uo:c-:a-:a:c.-:a:c:1|}\\
    \${>?<>G<r:H:U:S:S.:Uo:c-:a-:a:c.-:a:c:1|}\\
    \${<!<>g>r:H:U:S:S.:Uo:c-:a-:a:c.-:a:c:1|}\\
    \${!<r<>g>:H:U:S:S.:Uo:c-:a-:a:c.-:a:c:1|}\\
    \${>!>r<:H<>g>:U:S:S.:Uo:c-:a-:a:c.-:a:c:1|}\\
    \${<?<>@r>:H<:U<>g>:S:S.:Uo:c-:a-:a:c.-:a:c:1|}\\
    \${?<:H<>@r>:U<:S<>g>:S.:Uo:c-:a-:a:c.-:a:c:1|}\\
    \rvdots\\
    \${?:H:U:S:S.<:U<>@r>o:c>-:a<>g<-:a:c.-:a:c:1|}\\
    \${?:H:U:S:S.:U<r<R>:c<>g<-:a-:a:c.-:a:c:1|}\\
    \${?:H:U:S:S.:Ur>R><g<:c-:a-:a:c.-:a:c:1|}\\
    \${?:H:U:S:S.:Ur<R<>G>:c-:a-:a:c.-:a:c:1|}\\
    \${?:H:U:S:S.:UrR<:c<>G>-:a-:a:c.-:a:c:1|}\\
    \rvdots\\
    \${?:H:U:S:S.:UrR:c-:a-:a:c.-:a:c<:1<>G>|}\\
    \${?:H:U:S:S.:UrR:c-:a-:a:c.-:a:c>:1><B<|}\\
    \${?:H:U:S:S.:UrR:c-:a-:a:c.-:a>:c<>@1<b|}\\
    \${?:H:U:S:S.:UrR:c-:a-:a:c.->:a<>@1<:cb|}\\
    \rvdots\\
    \${?:H:U:S:S.:UrR>:c<>@1<-:a-:a:c.-:a:cb|}\\
    \${?:H:U:S:S.:Ur>R<>@1<:c-:a-:a:c.-:a:c>b>|}\\
    \${?:H:U:S:S.:Ur<w<:1:c-:a-:a:c.-:a>:c><g<|}\\
    \${?:H:U:S:S.:Urw:1:c-:a-:a:c.->:a<>g<:c|}\\
    \rvdots\\
    \${?:H:U:S:S.:Urw>:1><g<:c-:a-:a:c.-:a:c|}\\
    \${?:H:U:S:S.:Ur>w<>g<:1:c-:a-:a:c.-:a:c|}\\
    \${?:H:U:S:S.:Ur<>o<>G:1:c-:a-:a:c.-:a:c|}\\
    \${?:H:U:S:S.:U>r<>G<o:1:c-:a-:a:c.-:a:c|}\\
    \rvdots\\
    \${?>:H<>G<:U:S:S.:Uro:1:c-:a-:a:c.-:a:c|}\\
    \${>?<>G<:H:U:S:S.:Uro:1:c-:a-:a:c.-:a:c|}\\
    \${<!<>g>:H:U:S:S.:Uro:1:c-:a-:a:c.-:a:c|}\\
    \${!<:H<>g>:U:S:S.:Uro:1:c-:a-:a:c.-:a:c|}\\
    \${>!>:H<:U<>g>:S:S.:Uro:1:c-:a-:a:c.-:a:c|}\\
    \${<?<>@H>:U<:S<>g>:S.:Uro:1:c-:a-:a:c.-:a:c|}\\
    \${?<:U<>@H>:S<:S<>g>.:Uro:1:c-:a-:a:c.-:a:c|}\\
    \rvdots\\
    \${?:U:S:S.:U<r<>@H>o:1:c-:a-:a:c.-:a<:c<>g>|}\\
    \${?:U:S:S.:Ur<:H<OH:1:c-:a-:a:c.-:a:c<b<|}
\end{rules}
\end{tabular}}
}
%\end{document}
}
  \caption{Evolution of the history state without branching.}
  \label{fig:wheelbarrow-tape-evolution}
\end{figure}

The following table contains a list of all transition rules, visualized from a starting string of the form \${|I-:a:c.-:a:c:c b|}. The horizontal direction corresponds to space while the time flows from top to bottom. By default, the unitary associated to any rule---if not mentioned otherwise---is the identity. Apart from \swap operations, the only non-trivial unitary appears in the computation step in \cref{sec:computation}.

Observe that there are many possible local ambiguities within the history state, which we analyse in detail in \cref{sec:wheelbarrow-ambiguities}.

\newenvironment{ttabular}{%
  \begin{spacing}{.8}%
    \noindent
    \begin{tabular}{@{}p{\textwidth-4.35cm}R{4cm}@{}}
    }{
    \end{tabular}
  \end{spacing}
}

\subsubsection{Initialization}\label{sec:rules-initialization}
The initialization is done by moving a special symbol, a ``sweeper'' \$I, from one end of the tape to the other side. This ensures that the tape is actually correctly initialized, since any symbol apart from \${-:a}, \${:c} or the ghost \$b would result in a penalized configuration, cf.\ \cref{tab:allowed-character-tuples}.

\begin{ttabular}
  \toprule
  Left hand side has a sweeper \$I, right hand side an inactive ghost \$b, and all middle symbols are qubits \${-:a} or data bits $\${:c} \in \{ \${:0}, \${:1} \}$, which are opaque for the ghosts. We let the sweeper move through all middle symbols. This allows a dynamic ``initialization'' of the tape: if the sweeper \$I bumps into any symbol that is not a qubit, data bit or inactive ghost, we can penalize the configuration, singling out the proper history state.

  Once the sweeper reaches the ghost \${b} at the right boundary, it activates the ghost to \${B} and transitions to the box \$o.
  &
  \begin{rules}
    \$[0]{|>I>-:a:c.-:a:c:c b|}\\
    \$[1]{<-:a<>I>:c}\\
    \$[2]{<:c<I}\\
    \rvdots\\
    \$[6]{:c>I>:cb|}\\
    \$[7]{<:c<>I>b|}\\
    \$[0]{|-:a:c.-:a:c:c<o<B|}
  \end{rules}
  \\ \bottomrule
\end{ttabular}

\subsubsection{Ghost}
The ghost symbols act as general ``carriage return'' symbols: this saves having different return variants for each head symbol used, and is solely a way of saving local dimension.
The ghost can thus be seen as a particle to the right side of any other head symbol, and which diffuses freely on the tape (i.e.\ randomly moves left or right).
Only if the ghost is ``activated''---i.e.\ carries a ``head flag''---can it interact non-trivially with the symbols around it.

Generally, if there is an extra head symbol on the tape, the ghost is inactive (\$g and \$b). The ghost can itself carry the head flag, in which case we call it \emph{active} and denote it by \$G or \$B. The white active ghost can either turn itself into a head symbol on the left hand side, or activate the boundary. On the right boundary, it oscillates between white and black. This construction saves us a lot of symbols, since we only ever need to specify special right-moving heads, whereas the left movement of the head state is done generically by the ghost. We will often gloss over inactive ghost transitions and assume the ghost just ``moves out of the way'' as necessary.

\vspace{.3cm}
\begin{ttabular}
  \toprule
  Ghosts can change color on the right boundary. Since the black ghosts \${b} and \${B} are static, any incoming head from the left can detect when it has reached the boundary as it will encounter a black ghost.
  &
  \begin{rules}
    \${>B>|}\\\${<G<|}
  \end{rules}
  \begin{rules}
    \${>g>|}\\\${<b<|}
  \end{rules}
  \\ \midrule
  White ghosts can move through all static symbols, but \emph{not} through heads.
  &
  \begin{rules}
    \${>G>:*}\\\${<:*<G}
  \end{rules}
  \begin{rules}
    \${>g>:*}\\\${<:*<g}
  \end{rules}
  \\ \bottomrule
\end{ttabular}

\subsubsection{Base-6 Counter}
From a high-level perspective, the base-6 counter and unary counter (next section) work together to translate the tape length into a base-6 big endian number on the left side of the tape.
This base-6 number then encodes the program which we execute afterwards: we count in base $6$ through the sequence \${:T}, \${:U}, \${:A}, \${:S}, \${:H}, \$r---encoding a \toffoli, classically controlled unitary, ancilla, swap, halt or tape revert operation, respectively.

\vspace{.3cm}\noindent
\begin{ttabular}
  \toprule
  Active ghost \$G hits the left boundary and turns into the incrementer \$+.
  &
  \begin{rules}
    \${>|>G:T:U:A:Arr:H.}\\
    \${<|<+}
  \end{rules}
  \\ \midrule
  If the incrementer encounters \$r---the highest-valued digit---it flips it to \${:T}---the lowest-valued digit. This results in an overflow that is carried over to the next digit to the right. If the next digit is again \$r, the same procedure repeats until a different symbol or a qubit is encountered.
  &
  \begin{rules}
    \${|>+>r.r:x.}\\
    \$[1]{<:T<+}\\
    \rvdots\\
    \$[4]{>+>r:x.}\\
    \$[4]{<:T<+:x}
  \end{rules}
  \\ \addlinespace
  If the incrementer \$+ encounters \${:T} or classical zero \${:0}---both of which are treated as lowest-value symbol---it increments it to the next higher-valued symbol \${:U}. To uniquely distinguish to which symbol to decrement when run in reverse, the incrementer \$+ has to transition to the checking symbol \$X or \$Y, verifying that the symbol to its right is another counter or tape symbol, respectively. We never encounter the configuration \${+:1} or \${+-:a}, as this is penalized.

  The incrementation ends with the reverter symbol \$C.
  &
  \begin{rules}
    \${>+>:T:x}\\
    \${<:U<>X>:x}\\
    \${>:U<>Z<:x}\\
    \${<C<:U}
  \end{rules}
  \begin{rules}
    \${>+>:0-:a}\\
    \${<:U<>Y>-:a}\\
    \${>:U<>Z<-:a}\\
    \${<C<:U}
  \end{rules}
  \begin{rules}
    \${>+>:0:c}\\
    \${<:U<>Y>:c}\\
    \${>:U<>Z<:c}\\
    \${<C<:U}
  \end{rules}
  \\ \addlinespace
  If the incrementer encounters \${:U}, \${:A}, \${:S} or \${:H}, it increments the symbol to the next higher one---\${:A}, \${:S}, \${:H} or \$r, respectively---turning into the reverter symbol \$C.
  &
  \begin{rules}
    \${>+>:U}\\
    \${<C<:A}
  \end{rules}
  \begin{rules}
    \${>+>:A}\\
    \${<C<:S}
  \end{rules}
  \begin{rules}
    \${>+>:S}\\
    \${<C<:H}
  \end{rules}
  \begin{rules}
    \${>+>:H}\\
    \${<C<r}
  \end{rules}
  \\ \midrule
  The reverter \$C moves through the lowest-valued digits \${:T} all the way to the left boundary where it turns into the right mover \$D. Note that no digits other than \${:T} are possible to the left of \$C since incrementation proceeds to the next digit only in case of an overflow.
  &
  \begin{rules}
    \${>:T>C}\\
    \${<C<:T}
  \end{rules}
  \begin{rules}
    \${>|>C}\\
    \${<|<D}
  \end{rules}
  \\ \addlinespace
  The right mover \$D proceeds to the right through all digits of the base-6 counter.
  &
  \begin{rules}
    \${>D>:T}\\\${<:T<D}
  \end{rules}
  \begin{rules}
    \${>D>:U}\\\${<:U<D}
  \end{rules}
  \begin{rules}
    \${>D>:A}\\\${<:A<D}
  \end{rules}
  \begin{rules}
    \${>D>:S}\\\${<:S<D}
  \end{rules}
  \begin{rules}
    \${>D>:H}\\\${<:H<D}
  \end{rules}
  \begin{rules}
    \${>D>r}\\\${<r<D}
  \end{rules}
  \\ \addlinespace
  Eventually \$D encounters a qubit \${-:a} or classical bit \${:c}. It turns into an inactive ghost \$g and picks up the qubit \${-@a} or classical bit \${@c}.
  Afterwards the ghost \$g moves out of the way and we proceed to the unary counter.
  &
  \begin{rules}
    \${>D>-:a}\\\${<-@a<g}
  \end{rules}
  \begin{rules}
    \${>D>:c}\\\${<@c<g}
  \end{rules}
  \\ \bottomrule
\end{ttabular}
\vspace{.3cm}\noindent

For the configuration \${Do} there is no forward transition, which means that once we entered the computation phase, this counting does not continue.

\subsubsection{Unary Counter}\label{sec:rules-unary-counter}
The unary counter is necessary so that the base-6 counter knows when to stop.
We use a block symbol \$o to denote the position of the unary counter on the tape, starting from the right and moving to the left at each increment. Whenever this block is moved left once, the base-6 counter has been incremented by one as well.
In this way, once the unary counter has run out of space, we have translated the tape length into a base-6 number on the left side of the tape.

\begin{ttabular}
  \toprule
  A qubit \${-@a} or classical bit \${@c} is carried to the right past all other qubits or classical bits.
  &
  \begin{rules}
    \${>-@a>:c}\\\${<:c<-@a}
  \end{rules}
  \begin{rules}
    \${>-@a>-:b}\\\${<-:b<-@a}
  \end{rules}
  \begin{rules}
    \${>@c>:d}\\\${<:d<@c}
  \end{rules}
  \begin{rules}
    \${>@c>-:a}\\\${<-:a<@c}
  \end{rules}
  \\ \midrule
  The position of \$o indicates the value of the unary counter.
  As the qubit \${-@a} moves through it, the block \$o is pushed one position to the left---this increments the unary counter by one.
  &
  \begin{rules}
    \${>-@a>o}\\\${<o<-@a}
  \end{rules}
  \begin{rules}
    \${>@c>o}\\\${<o<@c}
  \end{rules}
  \\ \midrule
  Once the moving qubit \${-@a} reaches the black ghost \$b at the right boundary, the qubit \${-:a} is dropped and the ghost is activated to \$B.
  &
  \begin{rules}
    \${>@c>b|}\\\${<:c<B|}
  \end{rules}
  \begin{rules}
    \${>-@a>b|}\\\${<-:a<B|}
  \end{rules}
  \\ \bottomrule
\end{ttabular}

\subsubsection{Computation}\label{sec:computation}
The tape now has the form \${|:H:x.:xo:c-:a-:a:c.-:aG|}, i.e.~the counting is complete and by our choice of the chain length, the program description starts with a halt symbol \${:H}. The rest of the program string does not contain any \${:H}'s.

The idea behind the computation is depicted in \cref{fig:wheelbarrow-gates-intro}.
We first take the base 6 symbol from the left end of the program description and move it to the right end (e.g.\ $\${:x}2345$ would become $2345\${:x}$).
This symbol \${:x} can then be picked up by the box \${o}, which becomes activated to \${Ox}.
The active box is now followed by a set of rules which applies this program action to the (qu)bits right next to it.
Afterwards, the leftmost (qu)bit is carried to the right end and the procedure repeats.

The content of the tape symbols is checked on the fly using the ancilla program bit symbol \${OA}.
If it appears next to a qubit, a penalty is given for the qubit marginal being $\proj 1$; for a classical bit, we penalize \${:0} (this is because we do not have a \textsc{Not} gate, but \textsc{Not} can be implemented with \toffoli, which maps $111\mapsto110$).
The implementation details of all the different program bits are explained in the following table.

The computation halts once a halting program bit \${OH} is next to a classical \${:1}.

\vspace{.3cm}\noindent
\begin{ttabular}
  \toprule
  The boundary \$| is flipped to \$? and can only revert to \$| next to a halt symbol \${:H}---this ensures that we can only transition back and forth between counting and computation if the program bits are in their original order. An active ghost \$G can hit this left boundary \$? and activate it to \$!.
  &
  \begin{rules}
    \${>|>:H}\\
    \${<?<:H}
  \end{rules}
  \quad\quad\quad\quad
  \begin{rules}
    \${>?>G}\\
    \${<!<g}
  \end{rules}
  \\ \midrule
  A program bit $\${:x} \in \{ \${:T},\${:U},\${:A},\${:S},\${:H},\$r \}$ is picked up as \${@x} and carried to the right of the program string.
  &
  \begin{rules}
    \${>!>:x}\\
    \${<?<@x}
  \end{rules}
  \quad\quad\quad\quad
  \begin{rules}
    \${>@x>:y}\\\${<:y<@x}
  \end{rules}
  \\ \midrule
  These transition rules apply the \toffoli gate to three classical bits \${:c:c:c}. If either \${:0:c:c} or \${:1:0:c}, the last bit remains unchanged. Only for the configuration \${:1:1:c} we perform a bit flip on the last bit, i.e.~$d=\lnot c$. The first bit is then picked up with a carrier \${@0} or \${@1}.
  &
  \begin{rules}
    \${>@T>o:c:c:c-:a.}\\
    \${<:T<OT:c:c:c}
  \end{rules}
  \begin{rules}
    \${>OT>:0}\\
    \${<o<@0}
  \end{rules}
  \begin{rules}
    \${>OT>:1:c}\\
    \${<o<;1:c}
  \end{rules}
  \begin{rules}
    \${>;1>:0}\\
    \${<@1<:0}
  \end{rules}
  \begin{rules}
    \${>;1>:1:c}\\
    \${<:1<>"1>:c}\\
    \${:1<@1<:d}
  \end{rules}
  \\ \midrule
  These transition rules apply the classically-controlled unitary operation to a pair of qubits \${-:a-:a}, but only for the configuration \${:1-:a-:a}. The control bit is then picked up with a carrier \${@0} or \${@1}. Observe that this is the \emph{only} position where we apply a unitary operation to the quantum symbols.
  &
  \begin{rules}
    \${>@U>o:c-:a-:a:c.}\\
    \${<:U<OU:c-:a-:a}
  \end{rules}
  \begin{rules}
    \${>OU>:0}\\
    \${<o<@0}
  \end{rules}
  \begin{rules}
    \${>OU>:1-:a-:a}\\
    \${<o<>:!>-:a-:a}\\
    \$[1]{<:1<>-;a>-:a}\\
    \$[1]{>:1<>-"a<-:a}\\
    \$[1]{<@1<-:a}
  \end{rules}
  \\ \midrule
  \${OA-:a} acts as identity on the qubit, but is used later on to penalize when the attached Hilbert space is $\ket 1$, giving us the possibility for ancillas.

  \${OA:1} acts as identity, and we will penalize \${OA:0}, giving us a classical ancilla bit. Observe that we choose 1 here, as we can create 0s out of nothing but 1s with the \toffoli gate, but not vice versa.
  &
  \begin{rules}
    \${>@A>o-:a}\\
    \${<:A<>OA>-:a}\\
    \$[1]{<o<-@a}
  \end{rules}
  \begin{rules}
    \${>@A>o:1}\\
    \${<:A<>OA>:1}\\
    \$[1]{<o<@1}
  \end{rules}
  \\ \midrule
  \${OS} implements the \swap gate.
  &
  \begin{rules}
    \${>@S>o:c:d}\\
    \${<:S<>OS>:c:d}\\
    \$[1]{<o<>*c>:d}\\
    \$[2]{<@d<:c}
  \end{rules}
  \begin{rules}
    \${>@S>o:c-:a}\\
    \${<:S<>OS>:c-:a}\\
    \$[1]{<o<>*c>-:a}\\
    \$[2]{<-@a<:c}
  \end{rules}
  \begin{rules}
    \${>@S>o-:a:c}\\
    \${<:S<>OS>-:a:c}\\
    \$[1]{<o<>-*a>:c}\\
    \$[2]{<@c<-:a}
  \end{rules}
  \begin{rules}
    \${>@S>o-:a-:b}\\
    \${<:S<>OS>-:a-:b}\\
    \$[1]{<o<>-*a>-:b}\\
    \$[2]{<-@b<-:a}
  \end{rules}
  \\ \midrule
  \${OH} acts as identity on \${:0}, but has no forward transition for \${OH:1}, i.e.~the operation explicitly halts the computation.
  &
  \begin{rules}
    \${>@H>o:0.}\\
    \${<:H<>OH>:0}\\
    \$[1]{<o<@0}
  \end{rules}
  \\ \midrule
\end{ttabular}

% table broken to next page
\vspace{.3cm}\noindent
\begin{ttabular}
  \midrule
  \$R implements a tape revert, i.e.~moving the current tape position up by one. \$R acts like an activated boundary on the left hand side, i.e.~it blocks ghosts \$g or \$G. An incoming ghost from the right can activate from \$g to \$G, after which it proceeds back to the right end.

  Outlined in \cref{sec:rules-unary-counter} but run backwards, the activated ghost will move through to the right hand side and deactivate at the boundary, while picking up a qubit \${@a}. This qubit will move to the left until it encounters \$R. It drops the qubit and deactivates \$R to \$w. As soon as the inactive ghost \$g encounters this symbol, the ghost is reactivated and the box \$o is restored.
  &
  \begin{rules}
    \${>@r>o}\\
    \${<r<R}
  \end{rules}
  \quad\quad\quad\quad
  \begin{rules}
    \${>R>g}\\
    \${<R<G}
  \end{rules}
  \quad\quad\quad\quad
  \begin{rules}
    \${>-:a>B}\\\${<-@a<b}
  \end{rules}
  \begin{rules}
    \${>:c>B}\\\${<@c<b}
  \end{rules}
  \quad\quad
  \begin{rules}
    \${>R>-@a}\\\${<w<-:a}
  \end{rules}
  \begin{rules}
    \${>R>@c}\\\${<w<:c}
  \end{rules}
  \quad\quad
  \begin{rules}
    \${>w>g}\\\${<o<G}
  \end{rules}
  \\ \bottomrule
\end{ttabular}
\vspace{.3cm}\noindent

\begin{definition}[Turing's Wheelbarrow]
  \emph{Turing's Wheelbarrow} is the Quantum Thue System $(\Gamma,\set R,\{\op U_r\}_{r \in \set R},\sset C^2)$, where $\Gamma$ is given in \cref{def:alphabet} and the relation $\set R$ is defined by the transition rules in \cref{sec:wheelbarrow-rules} (with the conventions on notation from \cref{def:rules}).
\end{definition}

One can verify that Turing's Wheelbarrow, when applied to an initial string of the form \${|I-:a:c:c-:a.:c-:a b|}, where the sequence of \${:c}s and \${-:a} is such that they match the counting and computation phase, first translates the string length into a program description on the left string side, which is then executed cyclically on the tape. We call an initial configuration of this type \emph{valid initial configuration}.

There are, however, ambiguous transitions, which lead to branching in the graph---we discuss all possible branching points for the irreducible evolution containing this initial configuration.

\subsection{Branching in the History State}\label{sec:wheelbarrow-ambiguities}
We make extensive use of branching and ambiguous transitions to compress the number of symbols necessary to implement the Wheelbarrow. Therefore we need to show two things.
\begin{enumerate}
  \item The size of the history state is poly-bounded.
  \item There are no ambiguous transitions which lead to a penalized configuration.
\end{enumerate}

We take \cref{fig:wheelbarrow-tape-evolution} as a point of reference.

\paragraph{Ghosts.}
Whenever there is a ghost on the tape, it can either be active---\$G or \$B, or inactive---\$g or \$b.
\begin{fact}
  Inactive ghosts never change non-head symbols or pass through heads.
\end{fact}
Therefore we will disregard any branching due to inactive ghosts, which happens because we can always move either the head or ghost at each step. This increases the history state size by an at most quadratic factor.

\paragraph{Counting Phase.}
\begin{fact}
  Initialization and 6-ary counter are not ambiguous in either direction.
\end{fact}
We have two ambiguities to analyse. Starting from an intermediary counting stage where the program string starts with \${:H}, we can prematurely transition to the computation phase:
\[
  \begin{rules}
    \${>|>:HG:x.:x-:a:c.:co-:a.:a|}\\
    \${<?<>:H>G:x.:x-:a:c.:co-:a.:a|}\\
    \${>?<>G<:H:x.:x-:a:c.:co-:a.:a|}\\
    \${<!<>g>:H:x.:x-:a:c.:co-:a.:a|}\\
    \${>!<>:H<g:x.:x-:a:c.:co-:a.:a|}\\
    \${<?<@Hg:x.:x-:a:c.:co-:a.:a|}\\
    \rvdots\\
    \${?:x.<:x<@H-:a:c.<:c<go-:a.:a|}
  \end{rules}
\]
As there is no forward transition for \${@H-:a}, this branch is just a leg, increasing the history state by a small constant factor $\le2$. The same argument holds for transitioning to \$? during incrementation, i.e.
\[
  \begin{rules}
    \${>|>:H:x+:x.:x-:a.:co-:a.:a|}\\
    \${<?<:H:x+:x.:x-:a.:co-:a.:a|}
  \end{rules}
\]
Run forward, there is no transition for \${?C}, and run backwards there is none for \${?+}.

Running the carrier \${-@a} or \${@c} backwards for the unary counter at any point before counting is completed leads to another ambiguity, e.g.
\[
  \begin{rules}
    \${.-:a:c:co-:a<-:a<>@1>:c:c.}\\
    \$[2]{-:a:c:co-:a-:a<>:c<>@1:c}\\
    \$[2]{-:a:c:co-:a>-:a<>@1<:c:c}\\
    \$[2]{-:a:c:co>-:a<>@1<-:a:c:c}\\
    \$[2]{-:a:c:co<>@1<>-:a-:a:c:c}\\
    \$[2]{-:a:c:co<:1<>-"a>-:a:c:c}\\
    \rvdots\\
    \$[2]{-:a:c:c>o><:!<-:a-:a:c:c}\\
    \$[2]{-:a:c:c<OU<:1-:a-:a:c:c}
  \end{rules}
\]
There is no backwards transition for \${:cOx} or \${-:aOx} though. If there is no box \$o, the branch dies off even before that. This ambiguity hence increases the history state size by another small constant factor.

\paragraph{Computation Phase.}
A similar argument as in the last section shows that a late transition into the counting phase once we are in the middle of the computation does not proceed, as there is no forward transition for a configuration \${Do}. Furthermore, the same ambiguity running a carrier \${@c} or \${-@a} backwards holds, which we have already discussed.
\begin{fact}
  The application of gates \${OH}, \${OU}, \${OT}, \${OA} and \${OS} does not introduce any branching.
\end{fact}

It remains to analyse the revert command, where we have a branching point for a configuration
\[
  \begin{rules}
    \${.:c-:ao>:c<>G<-:a-:a:c.:c|}\\
    \$[2]{:c-:a>o<>G<:c-:a-:a:c.:c|}\\
    \$[2]{:c-:a<w<>g>:c-:a-:a:c.:c|}\\
    \$[2]{:c-:aw<:c<>g>-:a-:a:c.:c|}\\
    \$[2]{:c-:a>w>:c<-:a<>g>-:a:c.:c|}\\
    \$[2]{:c-:a<R<>@c>-:a<-:a<>g>:c.:c|}\\
    \rvdots
  \end{rules}
\]
Observe, however, that all that could happen is that the tape symbol is carried to the right, where it is dropped next to the boundary ghost \$b. The ghost is activated and moves back to \${RG}, where it deactivates. The branch does not continue further, as there is no transition out of \${-:aR} or \${:cR}. This increases the history state size by some small constant $\le 2$.

We define the set of tuples $\mathfrak A$ as all the possible character pairs that appear in this history state---including all branches---in \cref{tab:allowed-character-tuples}.
\begin{table}
  \centering
  \newcolumntype{k}{>{\centering\arraybackslash}p{5mm}}
  \newcommand{\X}{\ding{51}}
  \newcommand{\Xa}{\X$^\dagger$}
  \newcommand{\Xb}{\X$^\ast$}
  \newcommand{\Xc}{\X$^{\ast\dagger}$}
  \setlength\tabcolsep{0pt}
  \centerline{
    \begin{tabular}{k | kkkkkkkk | kkkk | kk | kkkkkk | kk | kkkkk | kkkk | }
      \toprule
      &                   \multicolumn{10}{c|}{non-heads}                   &                                                                   \multicolumn{21}{c}{heads}                                                                    \\
      \cmidrule{2-32} & \$| & \$? & \${:x} & \$o & \$w & \$R & \${:c} & \${-:a} & \$g & \$b & \$G & \$B & \$! & \$I & \$+ & \$X & \$Y & \$Z & \$C & \$D & \${@x} & \${Ox} & \${:!} & \${*c} & \${;1} & \${"1} & \${@c} & \${-;a} & \${-"a} & \${*a} & \${-@a} \\ \midrule
      \$|             &     &     & \X     &     &     &     & \X     & \X      &     &     & \X  &     &     & \X  & \X  &     &     &     & \X  & \X  &        &        &        &        &        &        &        &         &         &        &  \\
      \$?             &     &     & \X     &     &     &     &        &         &     &     & \X  &     &     &     & \X  &     &     &     & \X  &     & \X     &        &        &        &        &        &        &         &         &        &  \\
      \${:x}          &     &     & \X     & \X  & \X  & \X  & \X     & \X      & \X  &     & \X  &     &     &     & \Xa & \Xa & \Xa & \Xa & \Xa & \X  & \X     & \X     & \X     & \X     & \X     & \X     & \X     & \X      & \X      & \X     & \X      \\
      \$o             & \X  &     &        &     &     &     & \X     & \X      & \X  & \X  & \X  & \X  &     &     &     &     &     &     &     &     &        &        & \X     & \X     & \X     & \X     & \X     & \X      & \X      & \X     & \X      \\
      \$w             &     &     &        &     &     &     & \X     & \X      & \X  &     &     &     &     &     &     &     &     &     &     &     &        &        &        &        &        &        &        &         &         &        &  \\
      \$R             &     &     &        &     &     &     & \X     & \X      & \X  &     & \X  &     &     &     &     &     &     &     &     &     &        &        & \X     & \X     & \X     & \X     & \X     & \X      & \X      & \X     & \X      \\
      \${:c}          & \X  &     &        & \X  & \X  & \X  & \X     & \X      & \X  & \X  & \X  & \X  &     & \X  &     &     &     &     &     &     &        &        & \X     & \X     & \X     & \X     & \X     & \X      & \X      & \X     & \X      \\
      \${-:a}         & \X  &     &        & \X  & \X  & \X  & \X     & \X      & \X  & \X  & \X  & \X  &     & \X  &     &     &     &     &     &     &        &        & \X     & \X     & \X     & \X     & \X     & \X      & \X      & \X     & \X      \\ \midrule
      \$g             & \X  &     & \X     & \X  &     &     & \X     & \X      &     &     &     &     &     &     &     &     &     &     &     &     &        &        &        &        &        &        &        &         &         &        &  \\
      \$b             & \X  &     &        &     &     &     &        &         &     &     &     &     &     &     &     &     &     &     &     &     &        &        &        &        &        &        &        &         &         &        &  \\
      \$G             & \X  &     & \X     & \X  &     &     & \X     & \X      &     &     &     &     &     &     &     &     &     &     &     &     &        &        &        &        &        &        &        &         &         &        &  \\
      \$B             & \X  &     &        &     &     &     &        &         &     &     &     &     &     &     &     &     &     &     &     &     &        &        &        &        &        &        &        &         &         &        &  \\ \midrule
      \$!             &     &     & \X     &     &     &     &        &         & \X  &     &     &     &     &     &     &     &     &     &     &     &        &        &        &        &        &        &        &         &         &        &  \\
      \$I             &     &     &        &     &     &     & \X     & \X      & \X  & \X  &     &     &     &     &     &     &     &     &     &     &        &        &        &        &        &        &        &         &         &        &  \\ \midrule
      \$+             &     &     & \X     &     &     &     & \X     &         &     &     &     &     &     &     &     &     &     &     &     &     &        &        &        &        &        &        &        &         &         &        &  \\
      \$X             &     &     & \X     &     &     &     &        &         &     &     &     &     &     &     &     &     &     &     &     &     &        &        &        &        &        &        &        &         &         &        &  \\
      \$Y             &     &     &        &     &     &     &        & \X      &     &     &     &     &     &     &     &     &     &     &     &     &        &        &        &        &        &        &        &         &         &        &  \\
      \$Z             &     &     & \X     &     &     &     & \X     & \X      &     &     &     &     &     &     &     &     &     &     &     &     &        &        &        &        &        &        &        &         &         &        &  \\
      \$C             &     &     & \X     &     &     &     &        &         &     &     &     &     &     &     &     &     &     &     &     &     &        &        &        &        &        &        &        &         &         &        &  \\
      \$D             &     &     & \X     & \X  &     &     & \X     & \X      &     &     &     &     &     &     &     &     &     &     &     &     &        &        &        &        &        &        &        &         &         &        &  \\ \midrule
      \${@x}          &     &     & \X     & \X  &     &     & \X     & \X      & \X  &     &     &     &     &     &     &     &     &     &     &     &        &        &        &        &        &        &        &         &         &        &  \\
      \${Ox}          &     &     &        &     &     &     & \Xb    & \X      & \X  &     &     &     &     &     &     &     &     &     &     &     &        &        &        &        &        &        &        &         &         &        &  \\ \midrule
      \${:!}          &     &     &        &     &     &     &        & \X      & \X  &     &     &     &     &     &     &     &     &     &     &     &        &        &        &        &        &        &        &         &         &        &  \\
      \${*c}          &     &     &        &     &     &     & \X     & \X      & \X  &     &     &     &     &     &     &     &     &     &     &     &        &        &        &        &        &        &        &         &         &        &  \\
      \${;1}          &     &     &        &     &     &     & \X     &         & \X  &     &     &     &     &     &     &     &     &     &     &     &        &        &        &        &        &        &        &         &         &        &  \\
      \${"1}          &     &     &        &     &     &     & \X     &         & \X  &     &     &     &     &     &     &     &     &     &     &     &        &        &        &        &        &        &        &         &         &        &  \\
      \${@c}          &     &     &        & \X  &     &     & \X     & \X      & \X  & \X  &     &     &     &     &     &     &     &     &     &     &        &        &        &        &        &        &        &         &         &        &  \\ \midrule
      \${-;a}         &     &     &        &     &     &     &        & \X      & \X  &     &     &     &     &     &     &     &     &     &     &     &        &        &        &        &        &        &        &         &         &        &  \\
      \${-"a}         &     &     &        &     &     &     &        & \X      & \X  & \X  &     &     &     &     &     &     &     &     &     &     &        &        &        &        &        &        &        &         &         &        &  \\
      \${*a}          &     &     &        &     &     &     & \X     & \X      & \X  &     &     &     &     &     &     &     &     &     &     &     &        &        &        &        &        &        &        &         &         &        &  \\
      \${-@a}         &     &     &        & \X  &     &     & \X     & \X      & \X  & \X  &     &     &     &     &     &     &     &     &     &     &        &        &        &        &        &        &        &         &         &        &  \\ \midrule
    \end{tabular}
  }
  \caption{
    All possible character tuples occurring in the history state of Turing's Wheelbarrow. The row is the first character, the column the second---e.g.~\${|G} is allowed, whereas \${:x!} is not. \${:c} can be \${:0}, \${:1}, and \${:x} stands for any program bit \${:U}, \${:T}, \${:A}, \${:S}, \${:H}, or \$r. \Xa only allows the combination \${:T+}, \${:UX}, \${:UY}, \${:UZ} and \${:UC}.
    \Xb only allows the combination allowed by the gates, i.e.~\${OT:c}, \${OU:c}, \${OA-:a}, \${OA:1}, \${OS:c}, \${OS-:a} and \${OH:0}.
    Observe how the lower right block is completely empty, as there can only ever be one head on the tape.
  }
  \label{tab:allowed-character-tuples}
\end{table}

This exhaustive analysis of all possible branching points in the history state allows us to conclude the following corollary.
\begin{corollary}\label{cor:history-state-size}
  For strings of length $n$, the size of the irreducible evolution containing a valid initial configuration of the form \${|I-:a:c:c-:a.:c-:a b|}---the history state---is of size $O(n^3)$, and contains no forbidden character pairs.
\end{corollary}

\subsection{Simplicity of Turing's Wheelbarrow}
Let us briefly recall the idea behind simplicity in the context of QTSs. A QTS is called simple, if, for any two strings connected by more than one chain of transitions, the product of unitaries along this chain is identical. Equivalently, we can show that there are no loops in the graph connecting any strings. Regarding the QTS transition rules for Turing's Wheelbarrow, as constructed in the last section, it is easy to see that it will not be simple. However, for our purposes, it suffices to proof the following lemma.
\begin{lemma}\label{lem:wheelbarrow-simple}
  Each bracketed string in Turing's Wheelbarrow with at least one head either belongs to the history state, which is simple, or---by removing edges---can be broken up into $\poly n$-sized valid evolutions with illegal pairs.
\end{lemma}
\begin{proof}
  As no transition rule ever changes the number of heads or position of brackets, the distinction is well-defined. We can analyse each separately.

  \paragraph{One head.}
  We can exclude strings with illegal pairs right away. Furthermore, we can disregard configurations of non-head characters which are just allowed because there is a head symbol or ghost between them, such as \${-:aG:c}, as moving the head either way (which is possible, since there is only one of them) transitions to an illegal pair.

  So, disregarding any head and ghost state on the string for now, the most general non-head-non-boundary string compatible with \cref{tab:allowed-character-tuples} is
  \[
    \underbrace{\${:x}^*}_A \underbrace{\$w^?(\${:co}|\${-:ao}|\${:c}|\${-:a})^+}_B.
  \]
  It is straightforward to see that evolving this configuration backwards will transition to an illegal pair, if either
  \begin{itemize}
    \item $B$ has multiple \$os, or \$w and at least one \$o, or neither \$w or \$o.
    \item $A$ does not match the string length $\log_6 n$. Decrementation can only start if the substring $A$ starts with the halt symbol \${:H}, so it cannot happen that we start decrementing a rotated number, e.g.~\${:U:H:A:A:S:T} instead of \${:H:A:A:S:T:U}, which would translate into different lengths.
  \end{itemize}
  Evolving this string forward then reaches the computation part, and in case the pattern of classical and qubit states in $B$ does not match the one required for the encoded gates in $A$ we again have an illegal pair.

  We are left with the history state, and it suffices to check any transition rule containing a non-trivial unitary attached, which by construction is the computational step only, i.e.
  \[
    \begin{rules}
      \${>-;a>-:a}\\
      \${<-"a<-:a}
    \end{rules}.
  \]
  Following the transitions forward to the next such transition, by construction, the encoded Turing machine evolution is reversible, hence there is no loop as the Turing machine changes the classical content of the tape in section $B$.

  \paragraph{Multiple heads.} None of the heads can pass through each other. As further boundary markers such as \$|, \$? and \$! are immobile and opaque and there exists no transition out of \$b or \$B if not left of a boundary, we can without loss of generality assume that the tape is bracketed by either of \$|, \$!, \$b, \$B, or possibly no opaque symbol if our subsection lies at the tape ends.

  If there are $h\ge2$ head symbols on the tape, a simple argument allows us to slice the graph up into $\poly(\J(l))$-sized parts: first observe that following any of the heads---with potential intermediate transitions---sweeps the entire width of the string. For any configuration of the first $h-1$ heads, the last head will thus necessarily bump into the $h-1$\textsuperscript{st} within $O(\J(l))$ steps. The same argument shows that there can be at most one ghost on the tape.
\end{proof}

\subsection{Special Properties}
\begin{proof}[Proof of \cref{lem:qts-special}]
  We will check the properties of \cref{lem:qts-special} one-by-one.

  The deciding property follows by construction. The projectors $\Pinp=\Pout=\proj1$ are supposed to act on checking the first ancilla and output as seen in \cref{fig:augmented-verifier}, i.e.~they apply to the qubit after the special identity symbol \${O?}.

  \Cref{itm:W-preserves} is readily verified.

  \Cref{itm:W-decides}. The encoding is given by the valid initial configuration
  \[
    \enc(l) := \${|I}\underbrace{\${-:a:c-:a:c:c.-:a}}_{N\ \text{times}}\${b|},
  \]
  where $N$ is a unary encoding of the QRM head circuit executing \cref{fig:augmented-verifier} rewritten as depicted in \cref{fig:circuit-to-desc}, and the sequence of \${-:a}s and \${:c}s is such that they match the counting and computation phase. By construction, we therefore obtain, $N+4=\poly(|l|)=:\J(l)$.

  Also by construction and as outlined in \cref{sec:wheelbarrow-intro}, the program string on the left side of $\enc$ describes the head of a QRM writing out the circuit \cref{fig:augmented-verifier}. This QRM is in the same uniformity class as the original verifier's, and a constant in the size of $l\in\PromP$. We can hence pad it---using identity gates---to get the space and runtime for the QRM right, which can be as large as $\poly\J(l)$, as required for a \BQEXP computation.

  Both input and output markers $\sinp=\sout=\${OA-:a}$ are 2-local, contain one head \${OA}, and $\Pinp=\Pout=\proj1$.

  \Cref{itm:W-hist} We have shown the first claim in \cref{lem:wheelbarrow-simple,cor:history-state-size}. The rest follows by direct verification.

  \Cref{itm:W-others} We can immediately sort out the no-head and not-bracketed cases. The rest follows from \cref{lem:wheelbarrow-simple}.

  This concludes the proof.
\end{proof}

\subsection{Final Dimension Reduction}\label{sec:Reduction}
We want to make a few final remarks, and suggest an immediate optimization of the Wheelbarrow construction.

The distinction between the quantum and classical tape symbols \${:c} and \${-:a} is unnecessary, if we can ensure that there is never a quantum operation on classical symbols and vice versa. This is already proven.

The reason why we can merge these symbols is that while the QTS requires the ULG vertices to comprise only the classical alphabet symbols, we do not need to make this distinction for a ULG---as long as we can ensure that the Hilbert space dimension on each vertex in a connected component is the same. It is also clear that this does not break simplicity in \cref{lem:wheelbarrow-simple}, as we always \emph{know} which tape symbols are classical (the ones appearing next to classical operations, e.g.~\${OT}) and which ones are quantum (e.g.~the one next to \${:!}). This observation allows the following optimization.
\begin{remark}
  The Wheelbarrow construction works exactly the same when merging \${-:a} with \${:c}, \${-@a} with \${@c}, and \${-*a} with \${*c}.
\end{remark}

Once we have merged the symbols, there is another merge possible. We know that \toffoli and some basis-changing unitary $\op U$ are quantum-universal, see e.g.~\cite[ch.~4.5]{Nielsen2010}. This means that we can replace the classically-controlled unitary with such a one-qubit unitary, and apply \toffoli gates to quantum symbols as well. A similar argument as before shows that this does not break \cref{lem:wheelbarrow-simple}, and we phrase the following remark.
%\TODO[author=Johannes]{I don't think we can get rid of SWAP as well, but that might be worth investigating.}
\begin{remark}
  The Wheelbarrow construction works when replacing the controlled unitary with a single-qubit basis-changing unitary, and extending \toffoli to work on classical and quantum tape content. This makes the symbols \${:!}, \${-;a} and \${-"a}  obsolete.
\end{remark}

Including the saved symbols from the last two remarks---\${:c}, \${@c}, \${*c}, \${:!}, \${-;a} and \${-"a}---we conclude with the existence proof of \cref{lem:qts-special}.
\begin{corollary}\label{cor:qts-qma-exists}
  There exists a family of simple QTSs with 2-local rules on an alphabet of size $39$---$3$ of which are quantum with a Hilbert space $\sset C^2$---and all properties given in \cref{lem:qts-special}.
\end{corollary}

\begin{remark}\label{rem:no-scaling}
  It is straightforward to get $O(1)$-interactions, i.e.~removing the scaling polynomial $p(l)$ in \cref{th:hardness} if we can locally distinguish the history state at all times. This is possible e.g. by using distinct non-head symbols on the left and right hand side of the head and penalizing invalid configuration using regular expressions as in~\cite{Gottesman2009}. This would increase our dimension by roughly 15.
\end{remark}

\section{Conclusion}
This work was motivated by the idea of finding a simple, translationally-invariant and physically interesting system, for which the ground state energy problem is \QMAEXP-hard. In~\cite{Gottesman2009}, Gottesman and Irani concluded that their construction is not ``particularly natural'', due to the large local dimension necessary, but that the existence of some very simple \QMAEXP-hard local Hamiltonian problems seems quite possible.

Our results bring us another step closer to this goal: we reprove the hardness result in \cite{Gottesman2009} but with a local dimension of 42, whereas in~\cite{Gottesman2009}---though not explicitly specified---it was several orders of magnitude larger. To prove this result, we develop new tools and computational models which we believe are applicable to a wider range of problems.

At this point it would be interesting to see where the threshold for the translationally-invariant local Hamiltonian problem lies: does there exist a local dimension $d_\text{min}$, for which the problem is in \BQP, or \BQEXP? We have shown that $d_\text{min}<42$, but do not believe this to be a strict bound. We therefore encourage the interested reader to construct their own version of the Wheelbarrow, which might yield an even lower local dimension, and thus tighten our bound.

Furthermore, a lot of work recently has been done to analyse non-translationally-invariant systems, and to classify interactions with locally-varying interaction strengths, e.g.~\cite{Cubitt2013,Piddock2015}. In contrast to our construction, the hardness results in \cite{Piddock2015} resemble more a tiling construction, a subject also addressed in~\cite{Gottesman2009}. It would be an interesting approach to see if these two---fundamentally quite different---results can be combined, or if there exists yet another, completely different, method of encoding computation into the ground state of a local Hamiltonian.

Finally, we want to mention that while the research focus---as outlined in \cref{tab:history}---quickly shifted towards the 1D variant of the problem, from a physical perspective both 2D and 3D versions of this result are still of great interest, and apart from a trivial extension of our result to higher spatial dimensions, both remain open problems.

\section{Acknowledgments}

J.B.~acknowledges support from the German National Academic Foundation and the Engineering and Physical Sciences Research Council (EPSRC grant no.~1600123). T.S.C.~is supported by the Royal Society.  This work was made possible through the support of grant \#48322 from the John Templeton Foundation.  M.O.~acknowledges Leverhulme Trust Early Career Fellowship (ECF-2015-256) and European Union project QALGO (Grant Agreement No.~600700) for financial support.  The opinions expressed in this publication are those of the authors and do not necessarily reflect the views of the John Templeton Foundation.

\clearpage
\printbibliography

\end{document}

%%% Local Variables:
%%% mode: latex
%%% TeX-master: t
%%% End: